\documentclass[11pt,nosumlimits]{article}

\usepackage[margin=1in,footskip=0.25in]{geometry}

\author{Chris Camaño\thanks{Department of Computing and Mathematical Sciences, California Institute of Technology, Pasadena, CA 91125 USA (\email{ccamano@caltech.edu}, \email{eepperly@berkeley.edu},  \email{ram900@berkeley.edu}, \email{jtropp@caltech.edu})}\and Ethan N. Epperly\footnotemark[1] \and Raphael A. Meyer\footnotemark[1] \and Joel A. Tropp\footnotemark[1]}

\title{Faster Linear Algebra Algorithms \\ with Structured Random Matrices}

\date{Date: 5 July 2025.  Revised: 5 October 2025.}

\newcommand{\email}[1]{\href{mailto:#1}{#1}}

\usepackage[dvipsnames,table]{xcolor}
\usepackage[sort&compress,square,comma,numbers]{natbib}

\newif\ifsubmission
\submissiontrue

\usepackage[pagebackref]{hyperref}
\hypersetup{
    colorlinks=true,      
    linkcolor=PineGreen, 
    citecolor=NavyBlue,      
    filecolor=Purple,    
    urlcolor=Purple         
}

\renewcommand*{\backref}[1]{}
\renewcommand*{\backrefalt}[4]{%
  \ifcase #1 %
  (No citations.)
  \or
  (Cited p.~#2.)
  \else
  (Cited pp.~#2.)
  \fi
}

\usepackage[utf8]{inputenc}
\usepackage[T1]{fontenc}

\usepackage{fourier}
\DeclareMathAlphabet{\mathcal}{OMS}{cmsy}{m}{n}
\usepackage{bm}

\usepackage{amsmath}
\usepackage{amsthm}

\usepackage{algorithm, algpseudocode, algorithmicx}
\usepackage{bbm}
\usepackage{booktabs}
\usepackage{cancel}
\usepackage{caption}
\usepackage[capitalize, nameinlink, noabbrev]{cleveref}
\usepackage{commath}
\usepackage{enumitem}
\usepackage{float}
\usepackage{graphicx}
\usepackage{mathtools}
\usepackage{nicefrac}       
\usepackage{relsize} 
\usepackage{subcaption}
\usepackage{xspace}
\usepackage{url}
\usepackage{breakcites}
\usepackage{stmaryrd}

\usepackage{titlesec}
\titlespacing*{\section}{0pt}{12pt}{5pt}
\titlespacing*{\subsection}{0pt}{11pt}{5pt}
\titlespacing*{\subsubsection}{0pt}{11pt}{5pt}
\titlespacing*{\paragraph}{0pt}{6pt}{1em}

\floatstyle{plaintop}
\restylefloat{table}

\usepackage{crossreftools}
\renewcommand{\algorithmicrequire}{\textbf{Input:}}
\renewcommand{\algorithmicensure}{\textbf{Output:}}

\algrenewcommand\alglinenumber[1]{\sf\scriptsize\color{NavyBlue}{#1}}
\algrenewcommand\algorithmicrequire{\textbf{Input:}}
\algrenewcommand\algorithmicensure{\textbf{Output:}}

\usepackage{contour}
\usepackage[normalem]{ulem}

\contourlength{0.8pt}

\usepackage{todonotes}

\makeatletter
\def\hlinewd#1{%
	\noalign{\ifnum0=`}\fi\hrule \@height #1 \futurelet
	\reserved@a\@xhline}
\makeatother

\newtheorem{theorem}{Theorem}
\newtheorem{proposition}[theorem]{Proposition}

\newtheorem{openproblem}[theorem]{Open Problem}
\newtheorem{corollary}[theorem]{Corollary}
\newtheorem{lemma}[theorem]{Lemma}

\newtheorem{importedtheorem}[theorem]{Imported Theorem}
\newtheorem{importedlemma}[theorem]{Imported Lemma}

\numberwithin{theorem}{section}
\numberwithin{equation}{section}

\theoremstyle{definition}
\newtheorem{definition}[theorem]{Definition}

\theoremstyle{remark}
\newtheorem{remark}[theorem]{Remark}

\crefname{setting}{setting}{settings}
\Crefname{setting}{Setting}{Settings}
\crefname{problem}{problem}{problems}
\Crefname{problem}{Problem}{Problems}
\crefname{openproblem}{Open Problem}{Open Problems}
\Crefname{openproblem}{Open Problem}{Open Problems}
\crefname{equation}{}{}
\Crefname{equation}{Equation}{Equations}
\Crefname{importedtheorem}{Imported Theorem}{Imported Theorems}
\Crefname{importedlemma}{Imported Lemma}{Imported Lemmas}

\makeatletter
\newtheorem*{rep@theorem}{\rep@title}
\newcommand{\newreptheorem}[2]{%
\newenvironment{rep#1}[1]{%
 \def\rep@title{\Cref{##1}, restated}%
 \begin{rep@theorem}}%
 {\end{rep@theorem}}}
\makeatother
\makeatletter
\newtheorem*{rep@lemma}{\rep@title}
\newcommand{\newreplemma}[2]{%
\newenvironment{rep#1}[1]{%
 \def\rep@title{\Cref{##1} Restated}%
 \begin{rep@lemma}}%
 {\end{rep@lemma}}}
\makeatother
\makeatletter
\newtheorem*{rep@definition}{\rep@title}
\newcommand{\newrepdefinition}[2]{%
\newenvironment{rep#1}[1]{%
 \def\rep@title{\Cref{##1}, restated}%
 \begin{rep@definition}}%
 {\end{rep@definition}}}
\makeatother
\makeatletter
\newtheorem*{rep@corollary}{\rep@title}
\newcommand{\newrepcorollary}[2]{%
\newenvironment{rep#1}[1]{%
 \def\rep@title{\Cref{##1}, restated}%
 \begin{rep@corollary}}%
 {\end{rep@corollary}}}
\makeatother
\newreptheorem{theorem}{Theorem}
\newreplemma{lemma}{Lemma}
\newrepdefinition{definition}{Definition}
\newrepcorollary{corollary}{Corollary}

\usepackage[most]{tcolorbox}

\newcommand*{\vertbar}{\rule[-1ex]{0.5pt}{2.5ex}}

\newcommand{\e}{\mathrm{e}}

\renewcommand{\Re}{\mathrm{Re}}
\renewcommand{\Im}{\mathrm{Im}}

\renewcommand{\top}{\protect{\smash{*}}} 

\let\daggerFake\dagger
\renewcommand{\dagger}{{\smash{\daggerFake}}}

\newcommand{\defeq}[0]{\ensuremath{\;{\vcentcolon=}\;}\xspace}

\let\hat\relax
\newcommand{\hat}[1]{\smash{\skew{4}\widehat{\smash{\boldsymbol{#1}}\mathstrut}}}

\let\oldtilde\tilde
\let\tilde\relax
\newcommand{\tilde}[1]{\smash{\skew{4}\oldtilde{\smash{\boldsymbol{#1}}\mathstrut}}}

\let\norm\relax
\newcommand{\norm}[1]{\enVert{#1}}

\newcommand{\bignorm}[1]{\bigl\|#1\bigr\|}

\DeclareMathOperator*{\argmin}{arg\,min}

\DeclareMathOperator*{\Var}{Var}
\DeclareMathOperator*{\E}{\mathbb{E}}
\DeclareMathOperator{\tr}{Tr}

\DeclareMathOperator{\rank}{rank}

\DeclareMathOperator{\range}{range}

\DeclareMathOperator{\nnz}{nnz}
\let\Vec\relax
\DeclareMathOperator{\Vec}{vec}

\newcommand{\ie}{\text{i.e.}\xspace}
\newcommand{\etal}{\text{et al.}\xspace}

\newcommand{\eg}{\text{e.g.}\xspace}

\newcommand{\eps}[0]{\ensuremath{\varepsilon}}
\let\epsilon\eps
\newcommand{\rademacher}{\varrho}
\newcommand{\rad}{\rademacher}

\newcommand{\colsparse}{\xi}

\newcommand{\rev}[1]{#1}

\newcommand{\tsfrac}[2]{{\textstyle\frac{#1}{#2}}}

\newcommand{\vecalt}[1]{\boldsymbol{#1}} 

\newcommand{\mat}[1]{\bm{#1}} 
\renewcommand{\vec}[1]{\bm{#1}} 

\newcommand{\bmat}[1]{\begin{bmatrix} #1 \end{bmatrix}} 
\newcommand{\sbmat}[1]{\left[\begin{smallmatrix} #1 \end{smallmatrix}\right]} 

\newcommand{\R}{\bbR}

\newcommand{\C}{\bbC}
\newcommand{\F}{\mathbb{F}}

\newcommand{\coloneqq}{\defeq}
\newcommand{\prob}{\mathbb{P}}
\renewcommand{\Pr}{\prob}
\newcommand{\Id}{\mathbf{I}}
\newcommand{\order}{\cO}
\newcommand{\orderish}{\widetilde\cO}

\DeclareMathOperator*{\Mom}{Mom}
\DeclareMathOperator{\Span}{span} 
\newcommand{\lra}[2]{\ensuremath{\llbracket #1 \rrbracket_{#2}}}
\newcommand{\ptop}{{\vphantom{*}}} 
\newcommand{\nahutchpp}{\text{NA-Hutch\raisebox{0.35ex}{\relscale{0.75}++}}}

\definecolor{color1}{HTML}{2437E6}
\definecolor{color2}{HTML}{D12757}

\newcommand{\mA}{\ensuremath{\mat{A}}\xspace}
\newcommand{\mB}{\ensuremath{\mat{B}}\xspace}
\newcommand{\mC}{\ensuremath{\mat{C}}\xspace}
\newcommand{\mD}{\ensuremath{\mat{D}}\xspace}

\newcommand{\mF}{\ensuremath{\mat{F}}\xspace}
\newcommand{\mG}{\ensuremath{\mat{G}}\xspace}
\newcommand{\mH}{\ensuremath{\mat{H}}\xspace}
\newcommand{\mI}{\ensuremath{\mathbf{I}}\xspace}

\newcommand{\mM}{\ensuremath{\mat{M}}\xspace}
\newcommand{\mN}{\ensuremath{\mat{N}}\xspace}

\newcommand{\mP}{\ensuremath{\mat{P}}\xspace}
\newcommand{\mQ}{\ensuremath{\mat{Q}}\xspace}

\newcommand{\mS}{\ensuremath{\mat{S}}\xspace}
\newcommand{\mT}{\ensuremath{\mat{T}}\xspace}
\newcommand{\mU}{\ensuremath{\mat{U}}\xspace}
\newcommand{\mV}{\ensuremath{\mat{V}}\xspace}
\newcommand{\mW}{\ensuremath{\mat{W}}\xspace}
\newcommand{\mX}{\ensuremath{\mat{X}}\xspace}
\newcommand{\mY}{\ensuremath{\mat{Y}}\xspace}
\newcommand{\mZ}{\ensuremath{\mat{Z}}\xspace}
\newcommand{\mLambda}{\ensuremath{\mat{\Lambda}}\xspace}

\newcommand{\mSigma}{\ensuremath{\mat{\Sigma}}\xspace}
\newcommand{\mOmega}{\ensuremath{\mat{\Omega}}\xspace}

\newcommand{\mPsi}{\ensuremath{\mat{\Psi}}\xspace}
\newcommand{\mPhi}{\ensuremath{\mat{\Phi}}\xspace}

\renewcommand{\eqref}{\cref}

\let\oldthebibliography\thebibliography
\renewcommand{\thebibliography}[1]{%
  \oldthebibliography{#1}%
  \setlength{\itemsep}{0pt}%
  \setlength{\parskip}{0pt}%
}

\newcommand{\va}{\ensuremath{\vec{a}}\xspace}
\newcommand{\vb}{\ensuremath{\vec{b}}\xspace}

\newcommand{\ve}{\ensuremath{\mathbf{e}}\xspace}

\newcommand{\vg}{\ensuremath{\vec{g}}\xspace}

\newcommand{\vq}{\ensuremath{\vec{q}}\xspace}

\newcommand{\vu}{\ensuremath{\vec{u}}\xspace}

\newcommand{\vw}{\ensuremath{\vec{w}}\xspace}
\newcommand{\vx}{\ensuremath{\vec{x}}\xspace}
\newcommand{\vy}{\ensuremath{\vec{y}}\xspace}
\newcommand{\vz}{\ensuremath{\vec{z}}\xspace}

\newcommand{\vomega}{\ensuremath{\vecalt{\omega}}\xspace}
\newcommand{\vnu}{\ensuremath{\vecalt{\nu}}\xspace}

\newcommand{\cF}{\ensuremath{{\mathcal F}}\xspace}

\newcommand{\cH}{\ensuremath{{\mathcal H}}\xspace}

\newcommand{\cN}{\ensuremath{{\mathcal N}}\xspace}
\newcommand{\cO}{\ensuremath{{\mathcal O}}\xspace}

\newcommand{\cR}{\ensuremath{{\mathcal R}}\xspace}

\newcommand{\cT}{\ensuremath{{\mathcal T}}\xspace}

\newcommand{\cV}{\ensuremath{{\mathcal V}}\xspace}

\newcommand{\cX}{\ensuremath{{\mathcal X}}\xspace}

\newcommand{\bbC}{\ensuremath{{\mathbb C}}\xspace}

\newcommand{\bbF}{\ensuremath{{\mathbb F}}\xspace}

\newcommand{\bbN}{\ensuremath{{\mathbb N}}\xspace}

\newcommand{\bbR}{\ensuremath{{\mathbb R}}\xspace}

\newcommand{\rC}{\ensuremath{\mathrm{C}}\xspace}

\usepackage{silence}

\usepackage[dvipsnames]{xcolor}
\newcommand{\newtext}[1]{#1}

\begin{document}

\pagenumbering{gobble}

\maketitle

\begin{abstract}
To achieve the greatest possible speed,
practitioners regularly implement
randomized algorithms for low-rank approximation
and least-squares regression
with structured dimension reduction maps.
Despite significant research effort,
basic questions remain about the design and analysis
of randomized linear algebra algorithms
that employ structured random matrices.

This paper develops a new perspective on structured dimension reduction,
based on
the \emph{oblivious subspace injection} (OSI) property.
The OSI property is a relatively weak assumption
on a random matrix that holds when the matrix
preserves the length of vectors on average
and, with high probability, does not annihilate
any vector in a low-dimensional subspace.
With the OSI abstraction, the analysis
of a randomized linear algebra algorithm factors into two parts:
(i) proving that the algorithm works when implemented with an OSI;
and (ii) proving that a given random matrix model has the OSI property.

This paper develops both parts of the program.
First, it analyzes standard randomized algorithms for low-rank
approximation and least-squares regression under the OSI assumption.
Second, it identifies many examples of OSIs, including
random sparse matrices, randomized trigonometric transforms,
and random matrices with tensor product structure.
These theoretical results imply faster, near-optimal
runtimes for several fundamental linear algebra tasks.
The paper also provides guidance on implementation,
along with empirical evidence that structured
random matrices offer exemplary performance for
a range of synthetic problems
and contemporary scientific applications.
\end{abstract}

\setcounter{tocdepth}{2} 
\tableofcontents

\newpage
\setcounter{page}{0}
\pagenumbering{arabic}

\section{Introduction}

Randomized algorithms allow us to solve
certain linear algebra problems faster and
more reliably than ever before,
with real implications for the practice of 
machine learning and scientific computing.

In the field of randomized linear algebra,
many fundamental algorithms are based
on randomized linear dimension reduction,
often called \emph{sketching}
\cite{sarlos06,halko11,woodruff2014sketching,Drineas17,MT20,murray23,KT23:Randomized-Matrix}.
In this context, sketching algorithms compress%
a large
input matrix $\mA \in \R^{n \times d}$ by forming
the product $\mY = \mA \mOmega \in \R^{n \times k}$
with a smaller random \emph{test matrix}
$\mOmega \in \R^{d \times k}$,
where the \emph{embedding dimension} $k \ll d$.
We can quickly extract salient information about
the large input matrix $\mA$ by manipulating the
small \emph{sketch} matrix $\mY$.
This strategy leads to powerful algorithms for%
several important linear algebra problems,
including low-rank approximation~\cite{FKV98:Fast-Monte-Carlo,Williams00,DKM06:Fast-Monte-Carlo-II,Clarkson09,halko11,Clarkson13,Naka2020}
and least-squares regression~\cite{sarlos06,rokhlin08,avron10}.

For many sketching algorithms,
the cost of computing the sketch $\mY=\mA\mOmega$
dominates the overall runtime.
If we use an \emph{unstructured} test matrix $\mOmega$,
such as the standard Gaussian test matrix in~\cite{halko11},
the sketching step requires an expensive dense matrix--matrix product.
We can%
reduce the runtime by using a
\emph{structured} test matrix $\mOmega$
that supplies
a faster algorithm for the matrix--matrix product.
The most popular examples are
sparse random matrices~\cite{achlioptas03,Clarkson09,nelson13,KN14}
and randomized trigonometric transforms~\cite{ailon09,woolfe08,drineas11,tropp11SRHT}.
For linear algebra problems with tensor structure,
we can employ test matrices with a compatible tensor
structure, such as random Khatri--Rao products~\cite{biaginoni15,battaglino2018,rakhshan20,bujanovic25,camano25,ahle20}.

According to the cognoscenti~\cite{kapralov16,MT20}, sketching with a Gaussian test matrix is the gold standard for quality of output.  
Moreover, we can derive precise theoretical results that%
predict the empirical performance of Gaussian sketching methods~\cite{halko11,tropp2017a,tropp17b,tropp19,MT20,derezinski23,bartan20}.
In practice, algorithms that sketch with
a \emph{structured} test matrix
produce outputs whose quality matches
the Gaussian standard---at a dramatically lower
computational cost~\cite{tropp19,Epp25a,dong23,avron10,chen25,murray23,MT20}.
Yet we cannot fully explain the \newtext{performance} of
structured sketching methods through existing
theoretical frameworks.

\paragraph{Theory.}
This paper develops a new theoretical perspective
on sketching with structured random matrices.
We prove that several randomized linear algebra algorithms
succeed%
when implemented with any test matrix that
satisfies the \emph{oblivious subspace injection (OSI)}
property.%
The OSI condition is genuinely weaker than the
oblivious subspace \emph{embedding} (OSE) property that is
standard in the literature.  We explain how to establish
the OSI property%
using modern tools from high-dimensional probability~\cite{oliveira16,koltchinskii2015bounding,tropp23hdp,tropp25,Ver25:High-Dimensional-Probability-2ed}.
See \cref{sec:intro-osi,sec:intro-randlna-via-osi,sec:intro-osi-is-easier-than-ose} for an overview of this work.

\paragraph{Design.}

We make recommendations on how to design and implement structured sketching methods to attain the best empirical performance, with provable guarantees.
For sparse test matrices, our analysis justifies sparsity levels that were previously unsupportable, and%
we can contemplate%
levels that had been considered impossible
(\cref{sec:sparse-intro}).
Second, we construct a randomized
trigonometric transform with the minimal
embedding dimension and the fastest possible
sketching time (\cref{sec:fast-tranform-intro}).
Last, we produce the first qualitatively accurate bounds
on the worst-case behavior of Khatri--Rao sketching methods,
and we identify variations that offer exponential improvements
over the most popular constructions (\cref{sec:khatri-rao-intro}).

\paragraph{Why it matters.} For practical computation,
we must insist on algorithms that produce outputs with the
highest quality, as quickly as possible, and
with reliability guarantees.
Sketching methods implemented with our recommended
test matrices can meet these stringent requirements.
We include two stylized applications
that illustrate the benefits for
scientific computing (\cref{sec:science-applications}).
First, we use sparse test matrices to compress a large PDE simulation quickly and accurately.
Second, we use tensor product test matrices to estimate the partition function of a quantum mechanical system to high precision.

\paragraph{}

\noindent
\hspace{\parindent}
In sum, we hope that our work provides theoreticians
with new insights on sketching methods
for randomized linear algebra
and that it equips practitioners with faster and more reliable
algorithms.

\subsection{Empirical context: The power of structured random matrices}

Let us motivate our investigation by means of an empirical study.
Our experiments align with the conclusions of two decades
of empirical work on sketching algorithms, including \cite{tropp19,Epp25a,dong23,avron10,chen25,murray23}.
\textbf{Structured test matrices provide the same accuracy as Gaussian test matrices, while the structure allows us to design far more efficient algorithms.}

To confirm the first point,
\Cref{fig:rsvd-gaussian-vs-structured} charts the approximation performance of the randomized SVD algorithm (RSVD, \cref{alg:rsvd}) \cite[p.~227]{halko11}, implemented with several structured test matrices.
The%
testbed consists of each matrix from
the \texttt{SuiteSparse} collection~\cite{Davis11}
whose dimensions both fall between $300$ and $500{,}000$;
these $2{,}314$
matrices arise from a range of applications
and exhibit a variety of singular value
distributions.
For each structured test matrix and for each input matrix,
we compare the rank $k = 200$ approximation error to the baseline error achieved with a (scaled) standard Gaussian matrix,
and we perform three trials.%
In most cases,
the approximation quality is almost indistinguishable from
the Gaussian baseline.
Even for the hardest examples, the structured
test matrix yields at most $4 \times$ higher error.

Even though we suffer no loss in output quality by implementing
sketching algorithms with structured test matrices,
we gain dramatic improvements in efficiency.
For example,~\cref{fig:sparseStack-speed} illustrates
that sparse test matrices can reduce the sketching
time by \emph{two orders of magnitude},
as compared with Gaussian test matrices.
The benefits only accrue as the problem instances become larger.

\begin{figure}[t]
    \centering
    \resizebox{1\linewidth}{!}{%
        \input{figures/results/fig1.pgf}%
    }
    \caption{
\textbf{Randomized SVD: Structured vs.\ Gaussian.} 
For each structured test matrix, we form histograms of the
error ratio $\Vert \mA-\hat\mA_{\rm Structured}\Vert_{\rm F} / \Vert \mA-\hat\mA_{\rm Gaussian}\Vert_{\rm F}$ for matrix approximations with rank $k=200$, computed via RSVD (\cref{alg:rsvd}).
Each panel superimposes histograms from three trials.
The testbed consists of all $2{,}314$ matrices from the \texttt{SuiteSparse} collection \cite{Davis11} whose
dimensions
fall
between $300$ and $500{,}000$.
From left to right, we consider the SparseStack test matrix (\cref{def:sparse-stack}), the SparseRTT test matrix (\cref{def:intro-sparsertt}), and the complex spherical Khatri--Rao test matrix (\cref{def:kr-intro}) with tensor order $\ell=\lceil\log_2 n\rceil$.
The worst-case error ratio never exceeds 4,
and structured test matrices sometimes beat Gaussians!
These experiments suggest that structured sketching algorithms reliably produce approximation errors within a small constant factor of their Gaussian counterparts.
}
\label{fig:rsvd-gaussian-vs-structured}
\end{figure}
\subsection{Subspace embeddings (OSEs) and subspace injections (OSIs)}
\label{sec:intro-osi}

To understand the behavior of structured sketching methods for linear algebra,
we adopt a perspective that departs from the existing literature.
Our assumptions are%
genuinely weaker,
\newtext{and they suggest new possibilities for sketching and randomized dimension reduction.}

In his 2006 paper on sketch-and-solve algorithms for least-squares regression~\cite{sarlos06}, Sarl{\'o}s introduced the concept of an
\emph{oblivious subspace embedding (OSE)},
where ``oblivious'' means ``subspace independent.''
He proved that the sketch-and-solve algorithm
produces a near-optimal least-squares solution
whenever it is implemented with a random
test matrix that satisfies the OSE property.
Since then, researchers have employed OSEs to analyze
other randomized linear algebra algorithms,
including the RSVD and the generalized Nystr{\"o}m approximation~\cite{woodruff2014sketching,Drineas17,MT20}.
In our paper, the scalar field $\F \in \{ \R, \C \}$.

\begin{definition}[Oblivious subspace embedding] \label{def:ose}
    A random matrix \(\mOmega\in\F^{d \times k}\) is called an \emph{\((r,\alpha,\beta)\)-OSE} with \emph{subspace dimension} \(r\),
    \emph{embedding dimension} $k \geq r$, \emph{injectivity} \(\alpha\in(0,1]\), and \emph{dilation} \(\beta \geq 1\) if the following condition holds
    for each fixed \(r\)-dimensional subspace \(\cV\subseteq\F^d\).
    With probability at least $\frac{19}{20}$,
    \begin{equation} \label{eq:ose}
            \alpha \cdot \norm{\vx}_2^2
            \leq \norm{\mOmega^\top\vx}_2^2
            \leq \beta \cdot \norm\vx_2^2
        \quad
        \text{for all }
        \vx\in\cV.
    \end{equation}
\end{definition}

As it happens, the two inequalities in the
OSE property~\eqref{eq:ose} play asymmetric roles.
\newtext{We will argue that} the essential condition is
the lower bound, involving the injection parameter $\alpha$,
\newtext{which} ensures that the random matrix
$\mOmega$ does not annihilate any vector $\vx$ in the
fixed subspace $\cV$.
\newtext{Annihilating a vector destroys information, whereas increasing the length of a vector merely results in some distortion.  Thus,} upper bounds on the dilation $\beta$ \newtext{are less important}.
This insight motivates us to focus on the injectivity
property.

\begin{definition}[Oblivious subspace injection] \label{def:osi}
    A random matrix $\mOmega \in \F^{d\times k}$ is called an \emph{$(r,\alpha)$-OSI} with \emph{subspace dimension} $r$, \emph{embedding dimension} $k \geq r$, and \emph{injectivity} $\alpha \in (0,1]$ when it meets two conditions:
    \begin{enumerate}%
    \setlength{\itemsep}{0pt}
        \item \textbf{{Isotropy}.} \label{item:osi-isotropy}
        On average, the matrix preserves the squared length of each vector:
        \begin{equation} \label{eqn:osi-isotropy}
        \E{} \norm{\mOmega^\top \vx}_2^2 = \norm{\vx}_2^2
        \quad\text{for all $\vx \in \F^d$.}
        \end{equation}
        \item \textbf{{Injectivity}.} For each fixed $r$-dimensional subspace $\cV \subseteq \F^d$, with probability at least $\frac{19}
        {20}$,%
        \begin{equation} \label{eq:injectivity}
           \alpha\cdot \norm{ \vx }_2^2 \le \norm{ \mOmega^\top \vx }_2^2 \
            \quad\text{for all }\vx\in\mathcal V.
        \end{equation}
    \end{enumerate}
\end{definition}
\noindent
\newtext{To avoid an extra parameter, we require the injectivity property \cref{eq:injectivity} to hold with a fixed failure probability below $\tfrac{1}{20}$.
The definition is easily modified to treat the failure probability as a parameter.}

Morally, \cref{def:osi} of an OSI is weaker
than \cref{def:ose} of an OSE.
Indeed, the isotropy property of an OSI~\eqref{eqn:osi-isotropy} holds on average, while the dilation property of an OSE in~\eqref{eq:ose} must hold uniformly.
We have also obtained numerical evidence
that the OSI property is more permissive.%
In particular, \cref{fig:sparseStack-injectivity-intro} suggests
that an ultra-sparse test matrix \newtext{with $\zeta = 4$ nonzeros per row can sustain an injectivity far from zero even as the dilation noticeably diverges.}
See \cref{conj:constant-sparsity} for a precise formulation.

\begin{figure}
    \centering
    \resizebox{1\linewidth}{!}{%
        \input{figures/results/fig2.pgf}%
    }
    \caption{\textbf{SparseStack: OSI but not OSE?}
    Estimates for injectivity $\alpha = \sigma$ (\textit{left}) and dilation $\beta$ (\textit{right})
    of SparseStack test matrices (\cref{def:sparse-stack})
    with constant row sparsity $\zeta$ and embedding
    dimension $k = 2r$, applied to \rev{the adversarial subspace spanned by $\mQ = [\ve_1 ~ \ldots ~ \ve_r]$ with
    subspace dimension ranging from $r = 10^2$ to $r = 10^7$ }.%
    The markers track the median over 10 trials;
    shaded regions are bounded by the 10\% and 90\% quantiles.
    The%
    dashed lines mark the asymptotic value
    of the injectivity $(1 - 1/\sqrt{2})^2$
    and of the dilation $(1 + 1/\sqrt{2})^2$
    of a $2r \times r$ Gaussian test matrix (\cref{def:gauss-test})
    as $r \to \infty$.
    The experiment motivates \cref{conj:constant-sparsity}.
    }
\label{fig:sparseStack-injectivity-intro}
\end{figure}

In this paper, all test matrices%
satisfy the isotropy condition~\eqref{eqn:osi-isotropy}.
We sometimes refer to an OSI as an \emph{embedding}, even though its dilation factor $\beta$ is uncontrolled.
\Cref{sec:intro-osi-is-easier-than-ose} outlines
mathematical methods for establishing the OSI property.

\subsection{Randomized linear algebra with OSIs}
\label{sec:osi-program}

With the OSI abstraction (\cref{def:osi}), we can factor the analysis of a randomized sketching algorithm into two independent parts.
First, we argue that the algorithm succeeds when it is implemented using any test matrix $\mOmega \in \F^{d \times k}$ that is an OSI with specific parameters $(r, \alpha)$.
Second, we prove that a particular construction
of the random test matrix $\mOmega \in \F^{d \times k}$ enjoys the $(r,\alpha)$-OSI property.

To take full advantage of this theory, we must develop
constructions of OSI test matrices that have a
favorable computational profile.
Here is a summary of the desiderata for an OSI:

\begin{enumerate}[label=(D\arabic*)]
\setlength{\itemsep}{0pt}
\item   \textbf{OSI parameters.}
\label{item:test-matrix-desire-osi}
The embedding dimension is proportional to the subspace dimension: $k = \cO(r)$.  The injectivity parameter is constant: $\alpha \geq \mathrm{c}$, where $\mathrm{c} \in (0,1)$ is universal.

\item   \textbf{Arithmetic.}
\label{item:test-matrix-desire-arithmatic}
For each vector $\vx \in \F^d$, we can form the sketch $\mOmega^{\top} \vx$ with $\widetilde{\cO}( d )$
arithmetic operations.

\item   \textbf{Implementation.}
\label{item:test-matrix-desire-implement}
The OSI construction must be practical and implementable with available program libraries (e.g., sparse arithmetic, fast trigonometric transforms).  It should perform well in applications.

 \item   \textbf{*Other resources.}
 \label{item:test-matrix-desire-other}
 We also hope to control the amount of storage, randomness, and communication required to build and apply the test matrix $\mOmega$.  Our work does not address these questions.
\end{enumerate}

This paper carries out our program to study OSIs.
First, we analyze randomized linear algebra algorithms that use OSIs
(\cref{sec:intro-randlna-via-osi}).
Then we explore the universe of OSI constructions
(\cref{sec:smorgasbord}).

\subsection{OSIs suffice for (some) matrix computations} \label{sec:intro-randlna-via-osi}

Several widely used sketching algorithms for randomized linear algebra
can be implemented with OSIs.
This section summarizes the main conclusions;
we postpone full descriptions of the algorithms
and the theoretical analysis to \cref{sec:randnla-via-osi}.
Our first result concerns three standard randomized
algorithms for computing a low-rank approximation
of an input matrix.

\begin{theorem}[OSIs for low-rank approximation;
summary of \protect{\Cref{thm:rsvd-via-osi,cor:nystrom-via-osi,thm:gen-nystrom-via-osi}}]
    \label{thm:osi-omnibus}
    \newtext{
        Let \(\mA\in\F^{n \times d}\) be an input matrix, and choose a target rank \(r\leq \min\{n,d\}\).
    }
    Suppose that the random matrix \(\mOmega\in\F^{d \times k}\) is an ($r,\alpha$)-OSI with subspace dimension $r$ and embedding dimension $k \geq r$.
    Then each of the following statements holds with probability at least $\:\nicefrac45$.
    \begin{enumerate} \setlength{\itemsep}{0pt}
        \item
        \textbf{Randomized SVD.} RSVD (\cref{alg:rsvd}) with test matrix $\mOmega$ yields a rank-$k$ approximation $\hat\mA$ where%
        \begin{equation*}
            \norm{\mA-\hat\mA}_{\rm F}^2
                \leq (\rC / \alpha) \cdot
                \min_{\rank \mB \leq r}\  \norm{\mA- \mB}_{\rm F}^2.
        \end{equation*}
        
        \item \textbf{Nystr\"om approximation.}
        If $\mA$ is square and positive semidefinite, then the Nystr\"om approximation method (\cref{alg:nystrom}) with test matrix $\mOmega$ returns a rank-$k$ approximation $\hat\mA$ that satisfies
        \begin{equation*}
            \norm{\mA-\hat\mA}_*
                \leq (\rC / \alpha) \cdot
                \min_{\rank \mB \leq r}\ \norm{\mA- \mB}_*.
        \end{equation*}
        \item \textbf{Generalized Nystr\"om approximation.}
        Let $\mPsi\in \F^{n \times p}$ be a $(k,\alpha)$-OSI
        with%
        subspace dimension $k$ and embedding dimension $p \geq k$.
        With test matrices $\mOmega$ and $\mPsi$, the generalized Nystr\"om approximation method (\cref{alg:gen_nystrom_outer})
        constructs a rank-$k$ approximation $\hat\mA$ that satisfies
        \begin{equation*}
            \norm{\mA-\hat\mA}_{\rm F}^2
                \leq (\rC / \alpha)^2\cdot
                \min_{\rank \mB \leq r}\
                \norm{\mA-\mB}_{\rm F}^2.
        \end{equation*}
        \vspace{-1.75em}
    \end{enumerate}
    Here and elsewhere, \(\rC \geq 1\) is a universal constant.
    See \cref{table:dense-rates} for a summary of the
    arithmetic costs.
\end{theorem}

We write $\norm{\cdot}_{\rm F}$ for the Frobenius norm,
while $\norm{\cdot}_*$ denotes the Schatten 1-norm.
As is standard for randomized low-rank approximation algorithms~\cite{halko11,MT20,woodruff2014sketching,gittens16},
\cref{thm:osi-omnibus} presents a bicriteria guarantee where
the approximation rank $k$ exceeds the target rank $r$.
If desired, we can promote the rank-$k$ approximations to rank-$r$
approximations using a standard truncation
method~\cite[Thm.~9.3]{halko11}.

Our second result treats the sketch-and-solve algorithm \cite{sarlos06}
for solving a linear regression problem.

\begin{theorem}[OSIs for least-squares; summary of \cref{thm:lsq-via-osi}] \label{thm:osi-sketch-solve-summary}
    \newtext{Let \(\mA\in\F^{n \times d}\) be a design matrix, let \(\mB\in\F^{n \times m}\) be a response matrix,}
    and consider the linear least-squares problem
    \begin{equation*}
        \min\nolimits_{\mX \in \F^{d\times m}}
        \ \norm{\mA \mX - \mB}_{\rm F}^2.
    \end{equation*}
    Let $\mPsi \in \F^{n\times p}$ be a $(d,\alpha)$-OSI with subspace dimension $d$ and embedding dimension $p \geq d$.%
    With probability at least $\nicefrac{9}{10}$,%
    sketch-and-solve%
    (\cref{alg:sketch-and-solve}) with test matrix $\mPsi$ returns an approximate solution%
    $\tilde\mX \in \F^{d \times m}$ that satisfies%
    \begin{equation*}
        \norm{\mA\tilde\mX-\mB}_{\rm F}^2
            \leq (\rC / \alpha) \cdot
            \min\nolimits_{\mX}%
            \ \norm{\mA\mX-\mB}_{\rm F}^2,
    \end{equation*}
    where \(\rC \geq 1\) is a universal constant.
    See \cref{table:dense-rates} for a summary of the
    arithmetic costs.
\end{theorem}

\begin{table}[t]
\centering \small
\begin{tabular}{@{}cllll@{}} \toprule
	Task & Algorithm & Sketching time & Processing time &
    Total (dense data, SparseStack)%
    \\ \midrule
    Low-rank &
	  RSVD & \(\cT_{\rm sk}(\mA, k)\) & $nk^2 + k\cT_{\rm mv}(\mA^*)$
        & $ndr$ \\[0.25em]
	   approximation & Nystr\"om (psd \mA) & \(\cT_{\rm sk}(\mA, k)\) & \(nk^2\)
      & $n^2 \log (r) + nr^2$
\\[0.25em]
	  & Gen. Nystr\"om & \(\cT_{\rm sk}(\mA, k)+\cT_{\rm sk}(\mA^\top, p)\) & \(nk^2 + dkp\)
      & $nd \log (r) + (n+d)r^2$ \\[0.25em]
        \midrule
    Least squares &
    	Sketch-and-solve & \(\cT_{\rm sk}(\mA^\top, p)\) + \(\cT_{\rm sk}(\mB^\top, p)\) & \(p d^2+ p dm\)
        & $n(d+m) \log(d) + d^2(d+m)$
		\\[0.25em]
		\bottomrule
\end{tabular}
\caption{\textbf{Sketching algorithms for linear algebra: Our runtimes.}
    Sketching time and processing time for the algorithms in \cref{thm:osi-omnibus,thm:osi-sketch-solve-summary} applied to input matrix $\mA \in \F^{n \times d}$ and, for linear regression, response matrix $\mB \in \F^{n \times m}$.
    The test matrices $\mOmega \in \F^{d \times k}$ and $\mPsi \in \F^{n \times p}$.
    The function $\cT_{\rm sk}(\mM, q)$ returns the cost of forming
    the sketch $\mM \mOmega$ for a test matrix $\mOmega$ with $q$
    columns, and $\cT_{\rm mv}(\mM)$ is the time to compute the matrix--vector product $\vx \mapsto \mM \vx$.
    The last column lists the runtimes implied by our analysis
    when $\mA$ and $\mB$ are dense arrays%
    and the test matrices are SparseStacks (\cref{def:sparse-stack})
    with row sparsity $\zeta = \cO(\log r)$ and with sketch dimensions $k,p = \order(r)$ for low-rank approximation
    and with $p = \order(d)$ for linear regression.
    For legibility, the order notation $\cO$ is suppressed.}
\label{table:dense-rates}
\end{table}

Together, \cref{thm:osi-omnibus,thm:osi-sketch-solve-summary}
prove%
that we can use OSIs to implement several fundamental
sketching algorithms from randomized linear algebra.
The runtime analysis (\cref{table:dense-rates})
indicates that we can control the total sketching cost
by constructing OSIs that support fast matrix--vector products (\emph{matvecs}),
while we can limit postprocessing costs by
minimizing the embedding dimension of the OSI.
This approach ultimately allows us to obtain
faster algorithms because the OSI property
is weaker than the familiar OSE property
that drives most existing analyses.

A number of prior works have studied the performance of
sketching algorithms under conditions less stringent than the OSE property.
In particular, the tutorial~\cite{Drineas17} establishes parts of \cref{thm:osi-omnibus} under conditions weaker than the OSE.
The results \cite[Thm.~11.2]{halko11} and \cite[Lem.~2]{mahoney16} implicitly justify that RSVD works under an injectivity hypothesis,
without control on the dilation.
Saibaba \& Mi{\k{e}}dlar have undertaken a direct analysis of
test matrices with iid columns that satisfy moment%
conditions~\cite{saibaba2025randomized}.
Nevertheless, the OSI condition has not been isolated or
systematically exploited.

\subsection{A \texorpdfstring{\ifsubmission{sm{\"o}rg{\r a}sbord}\else\protect{\href{https://upload.wikimedia.org/wikipedia/commons/transcoded/e/e5/Sv-Smörgåsbord.oga/Sv-Smörgåsbord.oga.mp3}{sm{\"o}rg{\r a}sbord}}\fi}{smorgasbord} of OSI constructions}
\label{sec:smorgasbord}

\Cref{sec:intro-randlna-via-osi} advances the first part of the OSI
program, establishing that we can implement several randomized sketching
algorithms using OSIs.%
To address the second part of the program,
this section introduces
three examples of structured%
OSIs, based on sparse matrices (\cref{sec:sparse-intro}),
fast trigonometric transforms (\cref{sec:fast-tranform-intro}),
and tensor products (\cref{sec:khatri-rao-intro}).
These constructions meet the desiderata from
\cref{sec:osi-program},
allowing us to perform
fast linear algebra computations
in several environments.

\setcounter{subsubsection}{-1}
\subsubsection{Baseline: Gaussian test matrices}

As the gold standard for output quality,
Gaussian sketching methods provide
a point of comparison
in both the theoretical and the numerical work.

\begin{definition}[Gaussian test matrix]
\label{def:gauss-test}
    Fix ambient dimension \(d \in \bbN\) and embedding dimension \(k \in \bbN\).
    A real \emph{Gaussian test matrix} \(\mOmega\in\bbR^{d \times k}\) has iid \(\cN(0,\nicefrac1k)\) entries.
\end{definition}

The OSI properties of a Gaussian test matrix
reflect the
excellent quality of output.
According to~\cite[Sec.~8.7.2]{MT20},
Gaussian test matrices attain constant injectivity
$\alpha \geq \nicefrac{1}{2}$ for subspaces of
dimension $r$ at an embedding dimension
$k = \cO(r)$, 
which motivates the first desideratum
\ref{item:test-matrix-desire-osi}.
On the other hand, for a vector $\vx \in \R^d$,
forming the Gaussian sketch $\mOmega^\top \vx \in \R^k$
is expensive, requiring $\cO(dk)$ arithmetic operations
(cf.~\ref{item:test-matrix-desire-arithmatic}).
We aim to design test matrices
that match the OSI properties and output
quality associated with a Gaussian test matrix,
while permitting more economical computations.

\subsubsection{Sparse test matrices} \label{sec:sparse-intro}

For our first in-depth demonstration of the OSI framework,
we study \emph{sparse test matrices}.
Sparse test matrices offer compelling theoretical benefits
for both sparse and dense linear algebra,
and they exhibit spectacular empirical performance \cite{tropp19,chen25,Epp25a,melnichenko23,murray23,murray24,dong23}.
Yet there remain core questions about
their%
design and%
analysis.
Our results%
justify the application
of sparse test matrices at levels of sparsity
that were previously unsupportable, and they
hint at%
sparsity levels
that had been considered impossible.

The basic case for sparse sketching is simple.
If the test matrix $\mOmega \in \F^{d \times k}$
has only $\zeta$ nonzero entries \emph{per row},
then we can form the sketch $\mA\mOmega$ with
just $\cO(\zeta \cdot \nnz(\mA))$ arithmetic operations.
For contrast, sketching with an unstructured
dense test matrix typically costs $\cO(k \cdot \nnz(\mA))$ operations.
If the sparsity is much smaller than the
embedding dimension ($\zeta \ll k$),
then sparse test matrices
can offer significant acceleration
(when implemented with a high-performance sparse arithmetic library).
\Cref{fig:sparseStack-speed} provides empirical evidence
of this speedup.

\begin{figure}
\centering
\resizebox{1\linewidth}{!}{%
    \input{figures/results/fig3.pgf}%
}
    \caption{\textbf{SparseStack: Acceleration over Gaussian.}
    Runtime to form the sketch $\mA\mOmega \in \R^{n\times k}$
    of a dense input matrix $\mA \in \R^{n \times n}$
    when the test matrix $\mOmega \in \R^{n \times k}$
    is either a Gaussian (dashed, \cref{def:gauss-test}) or a SparseStack (solid, \cref{def:sparse-stack})
    with row sparsity $\zeta$.%
    The ambient dimension $n = 50,000$;
    the panels compare embedding dimensions $k = 500, 2500, 25000$;
    the series track the median over 10 trials.
    Minimizing the sparsity $\zeta$ is critical to
    achieve the greatest speed. The acceleration over
    the Gaussian baseline
    ranges from $7 \times$ to $88 \times$ at the
    recommended sparsity $\zeta = 4$, where the
    SparseStack \newtext{is empirically a subspace injection up to dimension $r=10^7$ (\cref{fig:sparseStack-injectivity-intro})}.
    }
\label{fig:sparseStack-speed}
\end{figure}

Among several types of sparse test matrices,
we recommend using the \emph{SparseStack},
an existing construction due
to Kane \& Nelson~\cite[Fig.~1(c)]{KN14},
because of its excellent empirical performance.
To analyze this sparse test matrix, we want to
establish an OSI property at the minimal
row sparsity $\zeta$ (to control sketching costs)
and with minimal embedding dimension $k$ (to control postprocessing costs).

\begin{definition}[SparseStack]
    \label{def:sparse-stack}
    Fix the ambient dimension \(d \in \bbN\), the row sparsity $\zeta \in \bbN$, and
    the block size $b \in \bbN$.
    The embedding dimension $k \coloneqq b \: \zeta$.
    A \emph{SparseStack} is a random matrix \(\mOmega\in\F^{d \times k}\)
    of the form
    \begin{equation*}
        \mOmega \coloneqq \frac{1}{\sqrt{\zeta}}\begin{bmatrix}
            \rademacher_{11} \ve_{s_{11}}^\top & \cdots & \rademacher_{1\zeta} \ve_{s_{1\zeta}}^\top \\
            \rademacher_{21} \ve_{s_{21}}^\top & \cdots & \rademacher_{2\zeta} \ve_{s_{2\zeta}}^\top \\ 
            \vdots & \ddots & \vdots \\
            \rademacher_{d1} \ve_{s_{d1}}^\top & \cdots & \rademacher_{d\zeta} \ve_{s_{d\zeta}}^\top
        \end{bmatrix} \quad \text{where}\quad
        \begin{aligned}
        &\text{$\rademacher_{ij} \sim \textsc{rademacher}$ iid;} \\
        &\text{$s_{ij} \sim \textsc{uniform}\{1,\ldots,b\}$ iid.}
        \end{aligned}
    \end{equation*}
    In this display,
    $\mathbf{e}_{i} \in \F^{b}$
    is the $i$th standard basis vector;%
    \emph{iid} means \emph{independent and identically distributed}. 
    
\end{definition}

\noindent
The SparseStack test matrix consists of
$\zeta$ independent copies of a (scaled)
CountSketch test matrix \cite{Clarkson13},
stacked side by side, like books on a shelf.
\Cref{sec:rnla-with-sparse} discusses implementation matters,
while~\cref{sec:sparse-alternatives}
outlines reasons to prefer the SparseStack over other sparse matrix constructions.

In a recent paper~\cite[Sec.~6]{tropp25},
Tropp studied the injectivity properties
of random sparse matrices.
Without%
proof, he stated a result that translates
into an OSI guarantee for the SparseStack matrix.
For reference, we present the full analysis in \cref{app:sprase-stack}.
\begin{importedtheorem}[SparseStacks are OSIs, \protect{\cite[Rem.~6.5]{tropp25}}] \label{thm:sparse-stack-osi}%
    A SparseStack test matrix $\mOmega \in \F^{d \times k}$
     serves as an $(r,\nicefrac12)$-OSI for some embedding dimension $k = \cO(r)$ and row sparsity $\zeta = \cO(\log r)$.
    With these parameters, building the sketch $\mA \mOmega$ has arithmetic cost $\cO( \nnz(\mA) \cdot \log r)$.
\end{importedtheorem}

\Cref{thm:sparse-stack-osi} compares favorably
with results%
for the OSE property.%
Cohen~\cite{cohen16} proved that a SparseStack is an
$(r, \nicefrac12, \nicefrac32)$-OSE
for some $k = \order(r\log r)$ and $\zeta = \order(\log r)$.
Chenakkod \etal~\cite{chenakkod25} proved a similar OSE bound for some
$k = \order(r)$ and $\zeta = \order(\log^3 r)$.
\rev{In a work~\cite{chenakkod25a} concurrent with this manuscript, the same authors obtained OSE bounds with $k = \order(r\log^{o(1)} (r))$ and $\zeta = \order(\log^{1+o(1)}(r))$.}
By shifting the goal from OSE to OSI,
\cref{thm:sparse-stack-osi}%
achieves \rev{favorable} embedding dimension
\emph{and} row sparsity.

We can now establish end-to-end guarantees for sketching algorithms
based on sparse test matrices by combining Tropp's result (\cref{thm:sparse-stack-osi}) with our new analysis (\cref{thm:osi-omnibus,thm:osi-sketch-solve-summary}).
We learn that randomized linear algebra algorithms implemented with the SparseStack test matrix \newtext{achieve error bounds within a constant factor of optimal.}
For sparse input matrices, our approach yields sketching times
within a single logarithmic factor of the best possible.
See \cref{table:sparse-runtimes-intro} for a summary
of our theoretical results and a comparison
with existing research.

\paragraph{Empirical results.}

This paper also contains new evidence that SparseStack matrices
have exemplary empirical performance.
\Cref{fig:sparseStack-injectivity-intro,fig:sparseStack-speed}
illustrate that a SparseStack test matrix
with constant row sparsity $\zeta = 4$
achieves an OSI guarantee \newtext{for an adversarial subspace},
while the sketching time
is up to \textbf{88$\times$ faster}
than the Gaussian benchmark.
\Cref{sec:sparse-experiments} documents additional experiments
on difficult synthetic test problems,
where SparseStack matrices%
achieve near-Gaussian output quality on
low-rank approximation tasks.

These benefits persist in practical computations.
\Cref{sec:science-pod-modes} showcases an
application that employs the generalized
Nystr{\"o}m method to compress and analyze
scientific simulation data.
For this problem, SparseStack matrices
offer up to a \textbf{12$\times$ speedup}
over the Gaussian baseline, while
attaining the same accuracy.
Together, these results demonstrate
that SparseStack test matrices are highly
effective for linear algebra computations,
while they are significantly faster than alternatives.

\paragraph{\newtext{What is the minimal sparsity?}}
While \cref{thm:sparse-stack-osi} allows for row sparsity $\zeta = \cO(\log r)$,
practitioners~\cite{epperly24a,melnichenko23,dong23,chen25,tropp19} have long observed that sparse test matrices \newtext{are already reliable with as few as $\zeta = 4$ nonzeros per row}.
\newtext{%
On the basis of the accumulated evidence and our
experiments in \cref{fig:sparseStack-injectivity-intro},
the initial draft of this manuscript \cite[Conjecture 1.9]{camano2025fasterlinearalgebraalgorithms}
proposed that a SparseStack with {constant} row sparsity $\zeta$ might achieve the OSI property.
Later, Han et al.~\cite{HRT26} refuted this conjecture by proving that the injectivity $\alpha$ must vanish in the limit $d,k,r\to\infty$ when the row sparsity $\zeta$ stays constant.
Their work raises a more refined theoretical question; cf.~\cite[Prob.~7.2]{HRT26}.}

\begin{openproblem}[\newtext{Oblivious SparseStack: Optimal sparsity}]
    \label{conj:constant-sparsity}
    \newtext{What is the minimal value $c \in (0,1]$ for which a SparseStack with some sparsity $\zeta = \order(\log^{c+o(1)}(r))$ and embedding dimension $k = \order(r)$ is an $(r,1/2)$-OSI for each $r\ge 1$?}
\end{openproblem}

\newtext{The resolution of \cref{conj:constant-sparsity} would provide firmer foundations for implementations of ultra-sparse dimension reduction methods.}

\begin{table}[t]
\centering \small
\begin{tabular}{@{}llll@{}} \toprule
    & \multicolumn{3}{c}{Best available runtime}  \\ \cmidrule(r){2-4}
    Algorithm & Chenakkod \etal \cite{chenakkod25} & Cohen \cite{cohen16} & Ours \\
    \midrule
    Sketch \& solve (Alg.~\ref{alg:sketch-and-solve})
        & \(\textcolor{color1}{\nnz([\mA ~ \mB])\log^3(d)} + \textcolor{color2}{d^3}\)
        & \(\textcolor{color1}{\nnz([\mA ~ \mB])\log(d)} + \textcolor{color2}{d^3\log d}\)
        & \(\textcolor{color1}{\nnz([\mA ~ \mB])\log(d)} + \textcolor{color2}{d^3}\) \\ 
    \midrule
    RSVD (Alg.~\ref{alg:rsvd})
        & \(\textcolor{color2}{\nnz(\mA)r + nr^2}\)
        & \(\textcolor{color2}{\nnz(\mA)r + nr^2\log^2 r}\)
        & \(\textcolor{color2}{\nnz(\mA)r + nr^2}\) \\
    Nystr{\"o}m (psd \mA, Alg.~\ref{alg:nystrom})
        & \(\textcolor{color1}{\nnz(\mA)\log^3(r)} + \textcolor{color2}{nr^2}\)
        & \(\textcolor{color1}{\nnz(\mA)\log(r)} + \textcolor{color2}{nr^2\log^2 r}\)
        & \(\textcolor{color1}{\nnz(\mA)\log(r)} + \textcolor{color2}{nr^2}\) \\ 
    Gen.~Nystr{\"o}m (Alg.~\ref{alg:gen_nystrom_outer})
        & \(\textcolor{color1}{\nnz(\mA)\log^3(r)} + \textcolor{color2}{(n+d)r^2}\)
        & \(\textcolor{color1}{\nnz(\mA)\log(r)} + \textcolor{color2}{(n+d)r^2\log^3 r}\)
        & \(\textcolor{color1}{\nnz(\mA)\log(r)} + \textcolor{color2}{(n+d)r^2}\) \\
    \bottomrule
\end{tabular}
\caption{\textbf{Sketching algorithms for sparse linear algebra: Runtimes.}
For sketching algorithms with the SparseStack test matrix (\cref{def:sparse-stack}),
we compare our results with the existing analyses
of Cohen~\cite{cohen16}
and of Chenakkod \etal~\cite{chenakkod25}.
\rev{We report the runtime to solve a least-squares or low-rank approximation problem to constant-factor accuracy (cf.\ \cref{sec:constant-factor}).}
\Cref{thm:sparse-stack-osi} justifies a row sparsity \(\zeta=\cO(\log d)\) at an embedding dimension \(k=\cO(r)\),%
so we obtain \emph{both} the faster sketching time of Cohen \emph{and} the faster postprocessing time of Chenakkod \etal
Matrices $\mA \in \F^{n \times d}$ and \rev{$\mB \in \F^{n \times m}$} can be dense or sparse, 
$\nnz(\cdot)$ returns the number of nonzero entries,
and we employ the standard matrix multiplication algorithm.
\textcolor{color1}{Blue terms} arise from sketching costs, while \textcolor{color2}{magenta terms} arise from postprocessing.%
We omit \cO symbols \rev{and assume $m\le d$} for legibility.
See \cref{table:sparse-runtimes} for another comparison.
}
\label{table:sparse-runtimes-intro}
\end{table}

\subsubsection{Test matrices based on fast trigonometric transforms} \label{sec:fast-tranform-intro}

Next, we study test matrices built from trigonometric transforms
with fast algorithms,
such as the discrete Fourier transform (DFT), the discrete
cosine transform (DCT), or the Walsh--Hadamard transform (WHT).
Historically, these constructions were among the first examples of structured test matrices~\cite{woolfe08,ailon09,tropp11SRHT,halko11}.
Sketching methods based on trigonometric transforms
remain popular because they show strong empirical performance~\cite{tropp2017a,Naka2020,nakatsuksa24}.
They also benefit from the wide availability
of software~\cite{fftw}
and hardware~\cite{slade13} for%
fast transforms\rev{, and block variants are available for parallel computing environments \cite{balabanov23}.}
Within this ecosystem, we will
design an effective test matrix%
that achieves
superior theoretical guarantees
by exploiting the OSI framework.

The standard template for a transform-based test matrix
consists of three matrix factors:
\begin{equation} \label{eqn:intro-srtt}
\mOmega \coloneqq \mD \mF \mS
\in \F^{d \times k}
\quad\text{where}\quad
\begin{aligned}
&\text{$\mD \in \F^{d\times d}$ is a random diagonal matrix;} \\
&\text{$\mF \in \F^{d\times d}$ is a trigonometric transform;} \\
&\text{$\mS \in \F^{d\times k}$ is a sparse random matrix.}
\end{aligned}
\end{equation}
The diagonal matrix $\mD$ is usually populated with iid Rademacher variables.
Typical choices for the trigonometric transform $\mF$ include the
WHT or the DCT (when $\F \in \{ \R, \C \}$)
or the DFT (when $\F = \C$).
The sparse matrix $\mS$ enacts a fast sampling operation that
reduces the dimension and rescales the output so that $\mOmega$ is isotropic.%
For an input matrix $\mA \in \F^{n \times d}$,
computing the sketch $\mA \mOmega$ usually
costs $\cO(nd \log d)$ operations for the
trigonometric transform, plus the cost of
applying the sparse sampling matrix.

To complete the construction~\eqref{eqn:intro-srtt},
we must design a sparse sampling matrix $\mS \in \F^{d \times k}$
that delivers fast sketching times along with injectivity guarantees.
The most common choice for $\mS$ is a restriction matrix
that selects $k$ coordinates from $\{1,\ldots,d\}$ uniformly at random;
the resulting test matrix%
is called a
\emph{subsampled randomized trigonometric transform (SRTT)}
~\cite[Sec.~9.3]{MT20}.
To achieve an OSI with subspace dimension $r$,
standard SRTTs require embedding dimension $k = \Omega( r \log r)$
to guard against the worst-case input matrix~\cite[Sec.~3.3]{tropp11SRHT}.
This fact has prompted researchers to explore
alternative forms of the sparse matrix~\cite{cartis21,chenakkod24,chenakkod25}
that allow for the minimal
embedding dimension $k = \cO(r)$.

We recommend using the \emph{SparseRTT}, a transform-based test matrix that achieves the optimal embedding dimension and the minimal sketching cost within the class~\eqref{eqn:intro-srtt}.

\begin{definition}[SparseRTT] \label{def:intro-sparsertt}
    Fix a column sparsity parameter $\colsparse \in \bbN$.  A \emph{SparseRTT} is a random matrix $\mOmega \coloneqq \mD\mF\mS \in \F^{d \times k}$,%
    where the diagonal matrix $\mD \in \F^{d \times d}$ has iid Rademacher entries,
    and the trigonometric transform $\mF \in \F^{d\times d}$ is a unitary WHT, DCT, or DFT.
    The sparse matrix $\mS \in \F^{d \times k}$ is a
    {SparseCol} random matrix (\cref{def:sparsecol})
    with exactly $\colsparse$ nonzero entries per column.
\end{definition}

\begin{theorem}[SparseRTTs are OSIs, variant of~\cref{thm:sparsertt}]
    \label{thm:sparsertt-intro}
    Let the subspace dimension $r\ge \log d$.  A SparseRTT test matrix \(\mOmega\in\F^{d \times k}\) serves as an $(r, \nicefrac{1}{2})$-OSI for some embedding dimension
    \(k = \cO(r)\)
    and some column sparsity $\colsparse = \cO(\log r)$.
    For a matrix $\mA \in \F^{n \times d}$,
    the sketch $\mA\mOmega$ has arithmetic cost $\cO(nd\log r)$.
\end{theorem}

In the setting of \cref{thm:sparsertt-intro},
the random sparse matrix $\mS$ has
$\Theta( r \log r )$ entries.
This is the minimum sparsity level possible because
of the SRTT counterexample~\cite[Sec.~3.3]{tropp11SRHT}.
In particular, the sketching cost cannot be reduced further
with an (oblivious) test matrix of the form~\cref{eqn:intro-srtt}.
For a comparison with prior theoretical work~\cite{ailon09,chenakkod24,chenakkod25},
\cref{table:sparsertt-comparison-body}
lists the%
costs of the sketch-and-solve method (\cref{alg:sketch-and-solve})
with various trigonometric test matrices of the
form~\eqref{eqn:intro-srtt}.%
Our numerical work (\cref{fig:SparseRTT})
confirms that the SparseRTT
outperforms other trigonometric transform test matrices,
with accuracy matching a Gaussian test matrix.
We discuss additional implementation details in \cref{sec:sparsertt-implementation}.

For a dense unstructured input matrix $\mA$,
the SparseRTT%
and the SparseStack (\cref{def:sparse-stack})
yield OSIs with the same arithmetic costs,
up to constant factors.%
Sparse test matrices are typically faster than those based on randomized trigonometric transforms because they exploit hardware-level parallelism more effectively and involve less data movement.
Therefore, for most applications,
we recommend the SparseStack over the SparseRTT,%
but SparseRTT test matrices remain useful in environments that lack
sparse linear algebra libraries.

\subsubsection{Test matrices with tensor structure}
\label{sec:khatri-rao-intro}

Many linear algebra problems in machine learning and scientific computing
exhibit \emph{tensor structure}.
Because of the towering scale of these problems, 
randomized dimension reduction serves as a powerful ally.
Nevertheless, sketching methods for tensors remain
poorly understood.
This section employs the OSI framework to develop
and analyze sketching methods for tensor data.

\paragraph{Khatri--Rao test matrices.}
The \emph{Kronecker matvec access model} \cite{meyer23,meyer25} offers a useful abstraction for operations on tensor data.
Here is the simplest version of the model.
Fix the \emph{base dimension} $d_0 \in \bbN$
and the \emph{tensor order} $\ell \in \bbN$.
We interact with an input matrix
$\mA \in \F^{n \times d_0^\ell}$ by
forming the matvec $\mA \vomega \in \F^{n}$
with a vector $\vomega \in \F^{d_0^{\ell}}$
that has \emph{Kronecker product} structure:
\begin{equation*}
    \vomega = \vomega^{(1)} \otimes \cdots \otimes \vomega^{(\ell)}
    \hspace{1cm}
    \text{where}
    \hspace{1cm}
    \vomega^{(1)},\,\ldots,\,\vomega^{(\ell)} \in \F^{d_0}.
\end{equation*}
\Cref{sec:khatri-rao-applications} lists several \rev{problems}%
that are well described by the Kronecker matvec access model \rev{including Sylvester equations, tensor network computations, and CP decomposition of tensors.}

Let us outline a natural construction of a random test matrix that
we can implement in the Kronecker matvec access model.
First, choose an isotropic random vector $\vnu \in \F^{d_0}$,
called the \emph{base distribution}.
The columns of the test matrix are
tensor products of independent copies of the base distribution.
Common examples of base distributions include
the real and complex Rademacher distributions,
as well as the real and complex Gaussian distributions.
We also treat the real and complex \emph{spherical distributions};
that is, $\vnu \sim \textsc{uniform}\big\{ \vx \in \F^{d_0} : \norm{ \vx }_2 = \sqrt{d_0} \big\}$.
For details \newtext{on these distributions}, see~\cref{sec:prelims}.

\begin{definition}[Khatri--Rao test matrix] \label{def:kr-intro}
    Fix base dimension $d_0 \in \bbN$, tensor order $\ell \in \bbN$,
    and embedding dimension $k \in \bbN$.
    Select an isotropic base distribution $\vnu$ on $\F^{d_0}$.
    Define the \emph{Khatri--Rao} test matrix
    \begin{equation*}
    \mOmega \coloneqq \frac{1}{\sqrt{k}}\, \begin{bmatrix}
        \vertbar & & \vertbar \\
        \vomega_1 & \cdots & \vomega_k \\
        \vertbar & & \vertbar
    \end{bmatrix} \in \bbF^{d_0^\ell \times k}
    \quad \text{where}\quad
    \vomega_j\coloneqq\vomega_j^{(1)} \otimes \cdots \otimes \vomega_j^{(\ell)} \in\F^{d_0^\ell}
    \quad\text{and}\quad
    \text{$\vomega_j^{(i)} \sim \vnu$ iid.}
    \end{equation*}
\end{definition}

For an input matrix $\mA \in \F^{n \times d_0^\ell}$,
we can form the sketch $\mA \mOmega \in \F^{n \times k}$
by performing $k$ queries in the Kronecker matvec
access model, one for each column of the test matrix.
Moreover, the test matrix $\mOmega$ inherits the
isotropy property~\eqref{eqn:osi-isotropy} from the
base distribution $\vnu$.
The matrix $\mOmega$ can also be written as a Khatri--Rao product~\cite[p.~462]{kolda09} of independent random matrices,
but we avoid this formalism.

\paragraph{OSI properties.}
There is an extensive literature on%
randomized linear algebra computations
with Khatri--Rao test matrices;
\eg, see~\cite{avron14,ahle20,rakhshan20,chen21,bujanovic25}.
To advance our understanding,
we establish
that Khatri--Rao test matrices are OSIs.
Here is an essential special case,
extracted from
\cref{corol:kr-large-sketch-osi,corol:kr-small-sketch}.

\begin{theorem}[Gaussian Khatri--Rao test matrices are OSIs]
    \label{thm:khatri-rao-intro}
    Assume that the base distribution $\vnu \in \R^{d_0}$
    is the \emph{real} standard normal distribution\newtext{.}%
    For an embedding dimension $k$ \newtext{and tensor order \(\ell\)},
    construct the Khatri--Rao test matrix $\mOmega \in \R^{d_0^{\ell} \times k}$ described by \cref{def:kr-intro}.
    \begin{enumerate}[label=(K\arabic*)] \setlength{\itemsep}{0pt}
        \item
        \label{item:kr-intro-large-sketch}
        \textbf{Exponential embedding dimension, constant injectivity.}
        The test matrix $\mOmega$ serves as an $(r,\nicefrac{1}{2})$-OSI for some embedding dimension $k = \order(3^{\ell} r)$.
        \item
        \label{item:kr-intro-small-sketch}
        \textbf{Linear embedding dimension, exponentially-small injectivity.}
        The test matrix $\mOmega$ serves as an $(r,\mathrm{c}^\ell)$-OSI for some embedding dimension $k = \order(r)$,
        where $\mathrm{c} \in (0,1)$ is a universal constant.
    \end{enumerate}

\noindent
    \ref{item:kr-intro-large-sketch} extends to many other
    distributions (\cref{corol:kr-large-sketch-osi}), while
    \ref{item:kr-intro-small-sketch} is
    more delicate (\cref{corol:kr-small-sketch}).
\end{theorem}

\Cref{thm:khatri-rao-intro} suggests that the quality
of a Khatri--Rao test matrix degrades quickly%
as the tensor order \(\ell\) increases.
Conclusion~\ref{item:kr-intro-large-sketch} states that we can achieve constant injectivity $\alpha$ when the embedding dimension $k$ is exponential in the tensor order $\ell$.
Result~\ref{item:kr-intro-small-sketch} states that we can reach
the optimal embedding dimension $k = \cO(r)$
at the cost of exponentially small injectivity $\alpha$.
The exponential scaling is not an artifact of the
analysis---it is necessary for worst-case examples~\cite{meyer25}.
Khatri--Rao sketching also has the surprising feature~\cite{meyer23,meyer25} that
changing the base distribution can result in an
\emph{exponentially large difference}
in the embedding dimension $k$ required
for constant injectivity $\alpha \ge \mathrm{c}$.
See \cref{sec:khatri-rao-top-level-section} for details.

We are not aware of any result similar to \cref{thm:khatri-rao-intro}, aside from the concurrent research \cite{saibaba2025improved},
although the literature contains OSE guarantees in special cases.
For example, Bujanovi{\'c} \etal
\cite{bujanovic25} prove that a Khatri--Rao test matrix with a real Gaussian base distribution and tensor order $\ell = 2$ is an OSE when the embedding dimension $k = \cO(r^{3/2})$.

\paragraph{Practitioner recommendations.}

In the literature, the most common base distributions
for Khatri--Rao test matrices are the
real Gaussian and \newtext{(}real\newtext{)} Rademacher distributions.
We \emph{advise against} using these base distributions.
Instead, we recommend the complex spherical
distribution.
If complex arithmetic is not available,
then we recommend using the real spherical distribution.

Why?
To guarantee the OSI property with constant injectivity
$\alpha \geq \nicefrac{1}{2}$,
the real Gaussian distribution requires embedding
dimension $k = \order(\smash{3^{\ell}} r)$.
By using a spherical distribution, the base of the exponent shrinks
by a factor depending on the base dimension \(d_0\).
In particular,
in the important case $d_0 =2$,
the real spherical distribution allows for
$k = \cO(\smash{1.5^{\ell}} r)$ and the complex spherical
distribution permits $k = \cO(\smash{1.\overline{33}}^{\ell}r)$.
The argument against using a Rademacher base distribution
is even more brutal, as Khatri--Rao test matrices over this distribution can fail catastrophically.
Indeed, for a worst-case subspace,
the Rademacher Khatri--Rao test matrix has
injectivity $\alpha = 0$ \emph{unless}
the embedding dimension $k = \Omega(2^{\ell})$; cf.~\cite{meyer25}.
\cref{sec:kr-rad-bad-numerics} presents
numerical experiments that document these effects.

\paragraph{Computational consequences.}

While the%
exponential factors in \cref{thm:khatri-rao-intro} may seem deadly,%
Khatri--Rao test matrices remain powerful tools for scientific computing.

First, many applications of tensor linear algebra concern tensors
with small order, such as $\ell = 2$ or $\ell = 3$.
In these cases, the exponential factors in $\ell$ become constants,
and we recover \newtext{OSI guarantees}
for some embedding dimension $k = \order(r)$.
As a concrete example, the next theorem treats
the problem of structured matrix approximation
under a bilinear query model.

\begin{theorem}[Matrix recovery from bilinear queries, variant of \protect{\cref{thm:matrix-recovery}}]
\label{thm:matrxi-recovery-intro}
    Choose a $d$-dimensional subspace of matrices \(\cF \coloneqq \Span\{\mM_1, \ldots, \mM_d\} \subseteq \R^{n \times n} \).
    \newtext{Let \(\mB\in\bbR^{n \times n}\) be a matrix that is}
    accessible with bilinear queries: $(\vx,\vy) \mapsto \vy^{\sf T} \mB \vx$.
    \Cref{alg:matrix-recovery} performs \(\cO(d)\)
    bilinear queries on $\mB$, and it returns a matrix $\tilde\mB \in \R^{n \times n}$ that is a near-optimal approximation to $\mB$ within the class $\cF$
    with probability at least $\:\nicefrac{9}{10}$.  That is,
    \[
        \norm{\smash{\mB - \tilde\mB}}_{\rm F}^2 \leq \rC \cdot \min_{\mM\in\cF}\ \norm{\mB - \mM}_{\rm F}^2 \quad \text{\newtext{for a universal constant $\rC \ge 1$}}.
    \]
\end{theorem}

The subspace $\cF$ of structured matrices is sometimes called a \emph{linearly parameterized family}.
Typical examples include the set of Toeplitz matrices or the
set of \newtext{sparse} matrices with fixed bandwidth.
The bilinear query model arises in many settings
\cite{swartworth2023optimal,chen24,andoni2013eigenvalues,li2014sketching,swartworth2023optimal}, and it has been the focus of
several theoretical studies
\cite{wimmer2014optimal,rashtchian2020vector,needell2022testing}.
To prove~\cref{thm:matrxi-recovery-intro},
we observe that a family of random bilinear queries
amounts to a Khatri--Rao sketch of (the vectorization of)
the target matrix \mB.
\Cref{thm:khatri-rao-intro} provides
an OSI for this test matrix,
which implies the \emph{optimal} query complexity $\cO(d)$.

Second, although Khatri--Rao test matrices behave exponentially
badly for worst-case problem instances, their empirical
performance is often substantially
better than the worst-case theory suggests.
As evidence, \cref{fig:rsvd-gaussian-vs-structured}
applies (complex spherical) Khatri--Rao test matrices with large tensor
order $\ell$ to accurately approximate thousands of%
matrices from \texttt{SuiteSparse}.
In \cref{sec:science-trace}, we successfully perform
variance-reduced trace estimation~\cite{meyer2021hutch++,persson22,epperly24trace}
to compute the partition function of a quantum system.
These experiments require Khatri--Rao products
with tensor order $\ell = 16$,
and yet the methods
are%
powerful enough to achieve \textbf{12 digits of precision}
in regimes where baseline methods \textbf{fail to achieve
even 1 digit of precision}.

\subsection{How to establish the OSI property}
\label{sec:intro-osi-is-easier-than-ose}

So far, we have argued that many random test matrices
exhibit the OSI property.  Moreover, when we use these
test matrices to implement randomized linear algebra
algorithms, we can achieve \newtext{satisfactory} error bounds \newtext{and}
state-of-the-art runtimes.
Still, some conceptual questions nag us.
\textit{How is the OSI property fundamentally
different from the OSE property?
Why is the OSI property more likely to hold?
How can we prove that a test matrix is an OSI?}

Throughout this section, we assume that the test matrix
has iid columns.
More precisely,
let $\vomega \in \F^d$ be an isotropic random vector:
$\E{} [ \vomega \vomega^\top ] = \Id$.
Consider the random test matrix
\begin{equation} \label{eqn:test-iid-cols}
\mOmega \coloneqq \frac{1}{\sqrt{k}} \begin{bmatrix} \vomega_1 & \dots & \vomega_k \end{bmatrix} \in \F^{d \times k}
\quad\text{where $\vomega_j \sim \vomega$ iid.}
\end{equation}

\paragraph{Intuition: Dilation versus injection.}

On closer inspection of the embedding property~\eqref{eq:ose},
we discover a fundamental asymmetry
between the dilation and injection conditions.
To see why, write%
\begin{equation} \label{eqn:test-mtx-action-sum}
\norm{ \mOmega^\top \vx }_2^2
    = \frac{1}{k} \sum_{i=1}^k | \langle \vomega_i, \vx \rangle |^2
    \quad\text{for each fixed vector $\vx \in \F^d$.}
\end{equation}
Note that the right-hand side is the average
of $k$ iid \emph{nonnegative} summands,
each with expectation $\norm{\vx}_2^2$.

We realize that a \emph{sufficient} condition
for the test matrix $\mOmega$ to stretch out
the vector $\vx$ by a large factor is \emph{for one} summand
in~\eqref{eqn:test-mtx-action-sum}
to be unusually large.
Thus, when linear marginals of $\vomega$ have heavy tails,
the dilation parameter $\beta$ tends to be large.
On the other hand, a \emph{necessary} condition for the
test matrix $\mOmega$ to annihilate the vector $\vx$
is \emph{for all} summands in~\cref{eqn:test-mtx-action-sum}
to equal zero.
Therefore, the injectivity parameter $\alpha$ hates to be zero.
In summary: \textbf{It is almost a {law of nature} that a test matrix
is an OSI (\cref{def:osi}) when the embedding dimension $k$ exceeds
the subspace dimension $r$ by a constant factor};
cf.~\cite[p.~12994]{koltchinskii2015bounding}.

\paragraph{Spectral formulations.}

We still need to identify strategies for establishing subspace injectivity.
To that end, observe that the OSI property admits an
equivalent formulation as a spectral inequality.

\begin{proposition}[OSI: Spectral formulation] \label{prop:osi-spec}
Suppose that $\mOmega \in \F^{d \times k}$ is an isotropic random matrix, as in~\eqref{eqn:osi-isotropy}.  The following conditions are equivalent.
\begin{enumerate}%
\setlength{\itemsep}{0pt}
\item   The test matrix $\mOmega$ is an $(r, \alpha)$-OSI.

\item   For each fixed matrix $\mQ \in \F^{d \times r}$ with orthonormal columns,%
\[
    \sigma_{\min}^2(\mOmega^\top \mQ) = \lambda_{\min}(\mQ^\top\mOmega\mOmega^\top\mQ) \geq \alpha > 0
    \quad\text{with probability at least $\:\nicefrac{19}{20}$.}
\]
\end{enumerate}
\end{proposition}

We define the minimum singular value and the minimum eigenvalue as
\[
\sigma_{\min}(\mB) \coloneqq \min_{\norm{\vu}_2 = 1}\ \norm{\mB \vu}_2
\quad\text{and}\quad
\lambda_{\min}(\mA) \coloneqq \min_{\norm{\vu}_2 = 1}\ \vu^\top \mA \vu.
\]
In particular, the minimum singular value of a matrix is zero if and only if the matrix has a \newtext{non}trivial null space.

Therefore, we can check the OSI property by obtaining lower bounds
for the minimum singular value%
of certain random matrices.
Fortuitously, the literature on high-dimensional probability%
contains several powerful tools for studying the minimum
singular value of a random matrix.

\paragraph{Fourth-moment bounds (\cref{sec:fourth-to-second}).}
Oliveira~\cite[Thm.~1.1]{oliveira16} established that fourth-moment bounds
suffice to control the minimum eigenvalue of a sample covariance matrix.
In detail, let $\vomega \in \F^d$ be an isotropic random vector
where the fourth moments of the linear marginals are uniformly bounded.
That is, for a fixed number $h \geq 1$, we assume that
\[
\E \big[ | \langle \vomega, \vu \rangle |^4 \big]
    \leq h
    \quad\text{for each unit vector $\vu \in \F^d$.}
\]
Consider a random test matrix $\mOmega \in \F^{d \times k}$ with iid columns, as in~\eqref{eqn:test-iid-cols}.
Then the random matrix
$\mOmega$ is an $(r, \nicefrac{1}{2})$-OSI
for some embedding dimension $k = \cO(hr)$.
We use this strategy to establish \newtext{conclusion}~\ref{item:kr-intro-large-sketch} \newtext{in} \cref{thm:khatri-rao-intro} on Khatri--Rao test matrices
with constant injectivity.

\paragraph{Small-ball probabilities (\Cref{sec:small-ball}).}

Koltchinskii \& Mendelson~\cite[Thm.~1.3]{koltchinskii2015bounding}
have demonstrated that small-ball probability bounds
allow us to control the minimum singular value of a random matrix.
We employ a variant of their approach,
adapted from~\cite{GKK+22:Geometry-Polytopes,tropp23hdp,Ver25:High-Dimensional-Probability-2ed}.
Consider an isotropic random vector $\vomega \in \F^d$ whose linear marginals admit a uniform bound on the small-ball probability.
That is, for a fixed number $\tau \in (0,1)$, we assume that
\[
\prob\big\{ | \langle\vomega, \vu \rangle |^2 \leq \tau \big\}
    \leq 0.01
    \quad\text{for each unit vector $\vu \in \F^d$.}
\]
Consider a random test matrix $\mOmega \in \F^{d \times k}$
with iid columns, as in~\eqref{eqn:test-iid-cols}.
Then the test matrix is an $(r, \nicefrac\tau{10})$-OSI
for some embedding dimension $k = \cO(r)$.
We apply this strategy to prove
\newtext{conclusion}~\ref{item:kr-intro-small-sketch} \newtext{in} \cref{thm:khatri-rao-intro}
on Khatri--Rao test matrices
with proportional embedding dimension. 

\paragraph{Gaussian comparison (\Cref{sec:gaussian-compare}).}
Tropp~\cite[Thm.~2.3]{tropp25} recently devised a method for
bounding the minimum eigenvalue of a random psd matrix by
comparison with the minimum eigenvalue of a suitable Gaussian
random matrix.  This approach is the most elaborate of the
three techniques we consider.
Here is the simplest setting.
Consider a random test matrix $\mOmega \in \F^{d \times k}$
with iid columns, as in~\eqref{eqn:test-iid-cols}.
For a fixed matrix $\mQ \in \F^{d \times r}$ with orthonormal columns,
we introduce the random psd matrix
\[
\mY \coloneqq \mQ^\top \mOmega \mOmega^\top \mQ
    = \frac{1}{k} \sum_{j=1}^k \mW_j
    \quad\text{where the summands $\mW_j = \mQ^\top \vomega_j \vomega_j^{\top} \mQ$ are iid.}
\]
The idea is to construct a Gaussian matrix $\mZ \in \F^{r \times r}$
whose statistics are determined by the summands:
\[
\E[ \mZ ] = \E[ \mW_1 ]%
\quad\text{and}\quad
\Var[ \tr( \mM \mZ) ] = k^{-1} \cdot \E\big[ \tr( \mM \mW_1 )^2 \big]
\quad\text{for all self-adjoint $\mM \in \F^{r\times r}$.}
\]
In an appropriate sense, the minimum eigenvalue $\lambda_{\min}(\mY)$
is stochastically bounded below by the minimum eigenvalue
$\lambda_{\min}(\mZ)$.
Since $\mZ$ is Gaussian, we can employ specialized methods
to study its minimum eigenvalue.
This strategy is effective for analyzing sparse test matrices.
We use it to establish \cref{thm:sparsertt-intro} about SparseRTTs
and to confirm \cref{thm:sparse-stack-osi} about SparseStacks.

\subsection{Outline}

This paper carries out the OSI program
for studying randomized linear algebra algorithms
with structured test matrices.
\Cref{sec:randnla-via-osi} shows that the OSI property
suffices to justify several fundamental algorithms.
Afterward, we establish OSI properties and discuss 
constructions of sparse test matrices (\cref{sec:sparse-sketching}),
fast trigonometric transforms (\cref{sec:sparse_trig_transforms}),
and Khatri--Rao products (\cref{sec:khatri-rao-top-level-section}).
\Cref{sec:science-applications} highlights two scientific applications
of randomized linear algebra that benefit from structured test matrices.
Last, \cref{sec:osi-proof-strats} elaborates on the proof strategies
that we use to check the OSI property, and \cref{sec:conclusion} summarizes and concludes.

\subsection{Notation \newtext{and} terminology}
\label{sec:prelims}

Our results are valid in either the real field $\R$ or the complex field $\C$; the symbol $\F$ refers to either field.
The bar $\overline{z}$ denotes the complex conjugate of $z \in \C$.
For vectors and matrices, the asterisk ${}^\top$ denotes the conjugate transpose,
while the superscript \(\smash{{}^{\sf T}}\) denotes the plain transpose.
The vector space \(\bbF^{d}\) is equipped with the canonical
inner product $\langle  \vx, \vy \rangle \coloneqq \vx^\top \vy$,
which induces the $\ell_2$ norm $\norm{ \cdot }_2$.
We define the standard Kronecker product \(\otimes\)%
and the vectorization operation \(\Vec(\cdot)\).

A \emph{positive semidefinite (psd)} matrix
is self-adjoint and has nonnegative eigenvalues.
An \emph{orthonormal} matrix $\mQ \in \F^{n\times r}$
satisfies the condition $\mQ^\top\mQ = \mI$.
The \emph{coherence} of an orthonormal matrix
is the numerical quantity
$\mu(\mQ) \coloneqq \max_{1\le i \le n}\ \norm{\smash{\ve_i^*\mQ}}_2^2$,
where $\ve_i$ denotes the $i$th canonical basis vector.
We work with the spectral norm $\norm{\cdot}_2$, the Frobenius norm $\norm{\cdot}_{\rm F}$, and the nuclear norm $\norm{\cdot}_*$.
For a matrix $\mA$, the symbol \(\lra{\mA}{r}\) denotes a best rank-$r$
approximation with respect to the Frobenius norm.
The adjective ``rank-$r$'' means ``with rank at most $r$.''

We employ two \newtext{(independent)} sketching matrices, which are always called \(\mOmega\) and \(\mPsi\), each of which is tall and skinny.
Consequently, \mOmega always acts on the right ($\mA \mapsto \mA\mOmega$), while $\mPsi^\top$ acts on the left ($\mA \mapsto \mPsi^\top \mA$).
The symbol $\zeta$ represents the (expected) number of nonzeros per \textit{row} of a sparse random matrix,
while $\colsparse$ is the (expected) number of nonzeros per \textit{column}.

The operator $\prob$ returns the probability of an event,
while $\E$ and $\Var$ compute the expectation and variance
of a random variable.
The symbol $\sim$ means ``has the distribution.''
As always, iid abbreviates ``independent and identically distributed.''
The real standard normal distribution \(\cN(0,1)\)
describes a real Gaussian random variable
with expectation zero and variance one.
The complex standard normal distribution \(\cN_{\bbC}(0,1)\) induces a complex random variable of the form $(Z_1 + \mathrm{i} Z_2)/\sqrt{2}$,
where $Z_1, Z_2 \sim \cN(0,1)$ iid.
The real Rademacher distribution $\textsc{rademacher} = \textsc{uniform}\{\pm 1\}$.
The complex Rademacher distribution $\textsc{rademacher}_{\C}$
induces a random variable of the form $(\varrho_1 + \mathrm{i} \varrho_2)/\sqrt{2}$ where $\varrho_1, \varrho_2 \sim \textsc{rademacher}$ iid.
We also define another complex distribution, $\textsc{steinhaus} = \textsc{uniform}\{ z \in \C : |z|=1 \}$.
The spherical distribution on $\F^d$ is uniform over the set of vectors with $\ell_2$ norm $\sqrt{d_0}$.

We employ the standard computer-science definitions
of the order symbols $\cO, \Omega, \Theta$.
The notation $\orderish(\cdot)$ suppresses polylogarithmic factors.%
As usual, the upright letters $\mathrm{c}, \mathrm{C}$%
denote positive universal constants that may change from appearance to appearance.

\section{Randomized matrix computations via OSIs}
\label{sec:randnla-via-osi}

In this section, we review standard algorithms for least-squares regression and low-rank approximation that employ randomized dimension reduction.  When these algorithms are implemented with OSIs, we can prove that they achieve errors within a constant factor of the optimal value.
The runtime of these algorithms also depends on the choice of the OSI,
so it is productive to isolate the sketching costs from the auxiliary arithmetic performed by each algorithm; see \cref{table:dense-rates}.

\Cref{sec:osis-respect-orth} contains our main technical observation,
which states that OSIs ``respect orthogonality.''
In \cref{sec:rvsd-via-osi}, we discuss the randomized SVD algorithm, which is a two-pass algorithm for computing a low-rank approximation.
\Cref{sec:nystrom,sec:gen-nystrom-via-osi} discuss the Nystr\"om and generalized Nystr\"om approximations, which produce low-rank approximations after a single pass over the matrix.
\Cref{sec:lsq-via-osi} treats the sketch-and-solve algorithm for least-squares, a fundamental method
that also figures in the analysis of the generalized Nystr\"om approximation.

\subsection{Basic technical observation: OSIs respect orthogonality} \label{sec:osis-respect-orth}

Before treating specific algorithms, we highlight a core theoretical insight:

\begin{lemma}[OSIs respect orthogonality]
    \label{lem:osi-preserve-orthogonality}
    \newtext{
        Let $\mQ \in \F^{n \times r}$ and $\mQ_{\perp} \in \F^{n \times s}$ be orthogonal matrices%
    }
    whose ranges are mutually orthogonal: $\mQ^\top \mQ_\perp = \bm{0}$.    
        Choose an arbitrary matrix $\mB \in \bbF^{t \times s}$.
    Draw a random $(r, \alpha)$-OSI test matrix $\mOmega \in \bbF^{n\times k}$. 
    With probability at least \(\frac9{10}\), the matrix
    $\mQ^\top \mOmega$ has full row rank, and
    \[
        \norm{\mB(\mQ_\perp^\top\mOmega)(\mQ^\top\mOmega)^\dagger}_{\rm F}^2
        \leq%
        (\rC / \alpha) \cdot \norm{\mB}_{\rm F}^2.
    \]
\end{lemma}

We can interpret the mysterious quantity $F \coloneqq \norm{\smash{\mB(\mQ_\perp^\top\mOmega)(\mQ^\top\mOmega)^\dagger}}_{\rm F}$ as the minimum Frobenius norm of a solution to the sketched least-squares problem
\[
\min\nolimits_{\mX}\ \norm{ \mOmega^\top ( \mQ \mX - \mQ_{\perp} \mB^\top ) }_{\rm F}^2.
\]
In words, $F$ describes how well we can \newtext{approximate the matrix} $\mQ_{\perp} \mB^\top$ by a \newtext{matrix} in the
\emph{complementary} subspace $\range(\mQ)$
after sketching both with $\mOmega$.
The lemma states that an OSI $\mOmega$ roughly preserves the
orthogonality between the matrix and the complementary subspace.
As we will discover, the result is valuable because
$F$ provides an exact expression
for the residual in the sketch-and-solve method for least-squares regression.  It also controls the error in randomized algorithms
for low-rank approximation of matrices.

\begin{proof}[\newtext{Proof of \cref{lem:osi-preserve-orthogonality}}]%
\newtext{Using the norm bound $\norm{\smash{\mM\mN}}_{\rm F} \le \norm{\smash{\mM}}_{\rm F}\norm{\smash{\mN}}_2$ for the Frobenius and spectral norms},
    we have%
    \begin{equation} \label{eqn:osi-orth-proof-1}
        \norm{\mB(\mQ_\perp^\top\mOmega)(\mQ^\top\mOmega)^\dagger}_{\rm F}^2
        \leq \norm{\mB\mQ_\perp^\top\mOmega}_{\rm F}^2 \cdot \norm{(\mQ^\top\mOmega)^\dagger}_2^2
        \leq \frac{\norm{\mB\smash{\mQ_\perp^\top}\mOmega}_{\rm F}^2}{\sigma_{\rm min}^2(\mOmega^\top \mQ)}.
    \end{equation}
    Since $\mQ$ has $r$ columns and $\mOmega$ is an $(r, \alpha)$-OSI, \cref{prop:osi-spec}%
    ensures that \(\sigma_{\rm min}^2(\mOmega^\top \mQ) \geq \alpha\) with probability at least 
    \newtext{$\frac{19}{20}$}
    .
    After we expand the numerator, the isotropy property~\eqref{eqn:osi-isotropy} of the OSI yields
    \[
        \E{}\bigl[\norm{\mB\mQ_\perp^\top\mOmega}_{\rm F}^2\bigr]
        = \E\big[\tr(\mB\mQ_\perp^\top\mOmega\mOmega^\top\mQ_\perp\mB^\top)\big]
        = \tr(\mB\mQ_\perp^\top\mQ_\perp\mB^\top)
        = \norm{\mB}_{\rm F}^2.
    \]
    Indeed, $\E{} [ \mOmega \mOmega^{\top} ] = \mI$.
    By Markov's inequality,%
    \begin{equation} \label{eq:markov}
    \norm{\mB\mQ_\perp^\top\mOmega}_{\rm F}^2 \leq 20 \cdot \norm{\mB}_{\rm F}^2 \quad \text{with probability at least 
    \newtext{$\frac{19}{20}$}
    }.
    \end{equation}
    With a union bound, we can control the numerator and denominator in~\eqref{eqn:osi-orth-proof-1} simultaneously
    with probability at least 
    \newtext{
    $
    \frac{9}{10}
    $
}.
    On this event, we can combine the inequalities to reach the stated result.    
\end{proof}

The rest of this section establishes that we can use OSIs
to implement RSVD, Nystr{\"o}m, generalized Nystr{\"o}m,
and sketch-and-solve regression.  The proofs all hinge
on \cref{lem:osi-preserve-orthogonality}. 

\begin{remark}[\newtext{Dependence on the failure probability}]
    \newtext{If we assume the subspace injection condition holds with probability $1-\delta/2$, then we may repeat the proof of \cref{lem:sketch-and-solve-error-formula} and conclude that $$\norm{\smash{\mB(\mQ_\perp^\top\mOmega)(\mQ^\top\mOmega)^\dagger}}_{\rm F}^2 \leq%
    (\rC / \alpha\delta) \cdot \norm{\mB}_{\rm F}^2.$$
    The error bound is inversely proportional to the failure probability $\delta$, a consequence of the use of Markov's inequality in \cref{eq:markov}.
    This $\order(\delta^{-1})$ scaling carries over to the remaining proofs in this section.
    To obtain sharper bounds, one needs only a tighter version of the estimate \cref{eq:markov}.
    For instance, if the columns of $\mat{\Omega} = k^{-1/2} [\vec{\omega}_1,\ldots,\vec{\omega}_k]$ are independent, we can decompose $\norm{\smash{\mB\mQ_\perp^\top\mOmega}}_{\rm F}^2 = k^{-1}\sum_{i=1}^k \norm{\smash{\mB\mQ_\perp^\top\vec{\omega}_i}}^2_2$ and employ a standard (scalar) concentration inequality to control the norm with high probability.
    We leave the exploration of such refinements to the interested reader.}
\end{remark}

\subsection{The randomized SVD}
\label{sec:rvsd-via-osi}

The randomized SVD \cite{halko11}, or RSVD, is a basic and widely used randomized algorithm for computing a low-rank approximation of a matrix.

Fix an input matrix $\mA \in \F^{n \times d}$ and an approximation rank  $k \leq \min\{n, d\}$.  We draw a random test matrix $\mOmega \in \F^{d \times k}$,
and we acquire information about the input matrix via the sketch $\mY \coloneqq \mA \mOmega \in \F^{n \times k}$.
The range of the sketch $\mY$ captures important directions in the column space of $\mA$.
The orthogonal projection of the input matrix $\mA$ onto the range of the sketch $\mY$ induces a rank-$k$ approximation of the form
$
\widehat{\mA} \coloneqq \mY \mY^{\dagger} \mA \in \F^{n \times d}.
$
The RSVD algorithm reports a compact SVD of this approximation: $\widehat{\mA} = \mU \mSigma \mV^\top$.

For numerical stability, the RSVD algorithm accomplishes this task by forming an orthonormal basis for the range of the test matrix:
\begin{equation} \label{eqn:rdm-rangefinder}
    \mQ \coloneqq \operatorname{orth}(\mA\mOmega) \in \F^{n \times k'}
    \qquad\text{where $k' \leq k$.}
\end{equation}
The rank-$k$ approximation $\widehat{\mA}$ of the input matrix can be expressed with the%
\textsf{QB} factorization:
\begin{equation} \label{eqn:rsvd-approx}
\hat\mA \coloneqq \mQ \mB \in \F^{n \times d} \quad \text{for }\mB\defeq \mQ^\top \mA.
\end{equation}
To obtain the approximation,
we need to compute the matrix--matrix
product $\mQ^* \mA = (\mA^\top \mQ)^\top \in \F^{k' \times d}$, which is often the dominant cost of the RSVD algorithm.
Afterward, the RSVD algorithm manipulates the \textsf{QB}
factorization to obtain a compact SVD.

See \cref{alg:rsvd} for reliable RSVD pseudocode.
The RSVD algorithm performs
\emph{two passes} over the input matrix $\mA$, including one sketching operation
($\mA \mOmega$) and one matrix--matrix product ($\mQ^* \mA$).  It also requires $\cO(k^2 n)$ operations for orthogonalization and SVD computations.   While we can exploit properties of the
input matrix $\mA$ in the matrix--matrix product $\mQ^* \mA$,
this step does not benefit from the structure of the test matrix $\mOmega$.  See \cref{table:dense-rates} for a summary of the arithmetic costs. \Cref{sec:nystrom,sec:gen-nystrom-via-osi} describe low-rank approximation algorithms that can achieve arithmetic costs strictly smaller than the RSVD.

\begin{algorithm}[t]
    \caption{Randomized singular value decomposition (RSVD) with structured sketching}
    \label{alg:rsvd}
    \begin{algorithmic}[1]
        \Require Matrix $\mA \in \F^{n \times d}$, approximation rank $k$
        \Ensure Orthonormal $\mU \in \F^{n \times k}$, $\mV \in \F^{d \times k}$ and nonnegative diagonal $\mSigma \in \F^{k \times k}$ defining $\widehat{\mA} \coloneqq \mU\mSigma\mV^\top$         \State Draw a \textit{structured} random test matrix $\mOmega \in \F^{d \times k}$ \Comment{ \textit{e.g., SparseStack, Khatri--Rao, etc.}}
        \State $\mY \leftarrow \mA\mOmega$
        \Comment{Sketch the input matrix}
        \State $\mQ \gets \texttt{orth}(\mY)$ \Comment{Orthogonalize columns of \mY}
        \State $\mB \leftarrow \mQ^\top\mA$
        \Comment{$\widehat{\mA} = \mQ \mB$ is the \textsf{QB} approximation}
        \State $[\widehat{\mU},\,\mSigma,\,\mV] \leftarrow \texttt{svd\_econ}(\mB)$
        \State $\mU \leftarrow \mQ\widehat{\mU}$
        \Comment{$\widehat{\mA} = \mU \mSigma \mV^\top$ is the RSVD approximation}
    \end{algorithmic}
\end{algorithm}

Here is our main theoretical result for the RSVD.
Recall that \(\lra{\mA}{r}\) denotes a best rank-\(r\)
approximation of a matrix $\mA$ with respect to the
Frobenius norm.

\begin{theorem}[Randomized RSVD with an OSI]
    \label{thm:rsvd-via-osi}
    \newtext{
        Let \(\mA\in\F^{n \times d}\) be an input matrix, and choose a target rank \(r\leq\min\{n,d\}\).
    }
    Draw a random $(r,\alpha)$-OSI test matrix \(\mOmega \in \F^{d \times k}\) where%
    $k \geq r$.
    Construct the rank-$k$ RSVD approximation $\widehat{\mA} \in \F^{n \times d}$ determined by~\eqref{eqn:rdm-rangefinder} and \eqref{eqn:rsvd-approx}.
    Then, with probability at least \(\frac9{10}\),
    \[
        \norm{\mA - \hat\mA}_{\rm F}^2 \leq%
        (\rC / \alpha) \cdot \norm{\smash{\mA-\lra{\mA}{r}}}_{\rm F}^2.
    \]
    \end{theorem}

\Cref{thm:rsvd-via-osi} provides a \emph{bicriteria} guarantee for the RSVD approximation.  The approximation error for the \emph{rank-$k$} output of the RSVD lies within a constant factor of the best \emph{rank-$r$} approximation error.
This type of bound is typical for the RSVD and related algorithms \cite{halko11,tropp23,gittens16,woodruff2014sketching}.
We can  upgrade the rank-$k$ approximation to a rank-$r$
approximation by a standard truncation technique~\cite[Thm.~9.3]{halko11}.

The proof of \cref{thm:rsvd-via-osi} combines the standard RSVD error bound
with the core technical observation about OSIs from \cref{sec:osis-respect-orth}.
Halko \etal%
employed a similar approach in their analysis of the RSVD with SRFT test matrices \cite[Thm.~11.2]{halko11}, as do Saibaba \& Mi{\k{e}}dlar \cite{saibaba2025randomized}.
A related idea appears in the tutorial of Drineas \& Mahoney~\cite[Rem.~98]{Drineas17}.

\begin{importedtheorem}[RSVD error decomposition \protect{\cite[Thm.~9.1]{halko11}, \cite[Lem.~2]{mahoney16}}]
\label{impthm:hmt-error-decomp}
    \newtext{
        Let \(\mA\in\F^{n \times d}\) be an input matrix.%
    }
Choose an arbitrary test matrix \(\mOmega \in \F^{d \times k}\), let \(\mQ \in \F^{n \times k}\) be an orthonormal matrix%
whose range contains the range of \(\mA\mOmega\),
and construct the approximation \(\hat\mA = \mQ\mQ^\top \mA \in \F^{n \times n}\).

To analyze the RSVD, fix $r \leq k$,
and introduce the partitioned SVD of the input matrix:
\[
	\mA = \mU \bmat{ \mSigma_1 \\ & \mSigma_2 } \bmat{\mV_1^\top \\ \mV_2^\top},
\]
where \(\mSigma_1\in\F^{r \times r}\) and \(\mV_1\in\F^{d \times r}\) contain the top \(r\) singular values and right singular vectors of \mA, while \(\mSigma_2\) and \(\mV_2\) contain the other singular values and right singular vectors.
Provided that $\mV_1^\top \mOmega$ has full row-rank,
\[
    \norm{\mA-\hat\mA_{\phantom{\hspace{0cm}}}}_{\rm F}^2 \leq
    \norm{\mSigma^\ptop_2}_{\rm F}^2
    + \norm{\mSigma^\ptop_2(\mV_2^\top \mOmega)(\mV_1^\top \mOmega)^\dagger}_{\rm F}^2.
\]
\end{importedtheorem}

\begin{proof}[Proof of \cref{thm:rsvd-via-osi}]
To proceed, we simply invoke \cref{lem:osi-preserve-orthogonality}
to bound the excess error term from \cref{impthm:hmt-error-decomp}.
Thus, with probability at least \(\frac9{10}\), we can compute
\begin{align*}
    \norm{\mA - \hat\mA}_{\rm F}^2
    &\leq \norm{\mSigma^\ptop_2}_{\rm F}^2%
    + \norm{\mSigma^\ptop_2(\mV_2^\top\mOmega)(\mV_1^\top\mOmega)^\dagger}_{\rm F}^2 \tag{\cref{impthm:hmt-error-decomp}}\\
    &\leq \norm{\mSigma^\ptop_2}_{\rm F}^2%
    + (\rC/\alpha) \cdot \norm{\mSigma^\ptop_2}_{\rm F}^2  \tag{\cref{lem:osi-preserve-orthogonality}}\\
    &= (1 + \rC/\alpha)\cdot%
    \norm{\mA - \lra{\mA}{r}}_{\rm F}^2.
\end{align*}
In the last line, we applied the Eckart--Young theorem to rewrite \(\norm{\mSigma_2}_{\rm F}^2 = \norm{\mA - \lra{\mA}{r}}_{\rm F}^2\).
Finally, note that $1 + \rC/\alpha \leq \rC'/\alpha$,
where $\rC'$ is a larger constant.
\end{proof}

\subsection{Nystr\"om approximation} \label{sec:nystrom}

\begin{algorithm}[t]
    \caption{Randomized Nystr\"om approximation with structured sketching}
    \label{alg:nystrom}
    \begin{algorithmic}[1]
        \Require Psd target matrix $\mA \in \F^{n\times n}$, approximation rank $k$
        \Ensure Orthonormal $\mU \in \F^{n \times k}$ and nonnegative diagonal $\mLambda \in \F^{k\times k}$ defining $\hat\mA = \mU \mLambda \mU^\top$ 
        \State Draw a \textit{structured} random test matrix $\mOmega \in \F^{n \times k}$ \Comment{ \textit{e.g., SparseStack, Khatri--Rao, etc.}}
        \State $\mY \leftarrow \mA \mOmega$
        \State $\nu \leftarrow \varepsilon_{\rm mach}\newtext{\|\mY\|_{\rm F}}$ \Comment{Compute shift where $\eps_{\rm mach}$ is floating-point unit precision}
        \State $\mY_{\nu} \leftarrow \mY + \nu \mOmega$ \Comment{Sketch of shifted matrix}
        \State $\mC \leftarrow \texttt{chol}(\mOmega^\top \mY_{\nu})$
        \Comment{Or pivoted Cholesky, to be extra safe}
        \State $\mB \leftarrow \mY_{\nu} \mC^{-1}$ \Comment{Triangular solve}
        \State $[\mU, \mSigma, \sim] \leftarrow \texttt{svd\_econ}(\mB)$ \Comment{Dense SVD}
        \State $\mLambda \leftarrow \max\{0, \mSigma^{2} - \nu \mI\}$ \Comment{Remove shift}
        \State $\mU \leftarrow \mU(:,1:k)$ \Comment{Truncate rank to $k$}
        \State $\mLambda \leftarrow \mLambda(1:k,1:k)$
    \end{algorithmic}
\end{algorithm}

When the input matrix is psd, we can obtain a low-rank approximation in a single pass using 
the \emph{Nystr\"om approximation} \cite{tropp17b,gittens16,drineas05,halko11}.

Fix a psd input matrix $\mA \in \F^{n \times n}$ and an approximation rank $k \leq n$.  Draw a random test matrix $\mOmega \in \F^{n \times k}$, and form the sketch $\mY \coloneqq \mA \mOmega \in \F^{n \times k}$.  The Nystr{\"o}m approximation is the rank-$k$ psd matrix
\begin{equation} \label{eqn:nystrom-approx}
    \hat\mA = \mY (\mOmega^\top \mY)^\dagger \mY^\top \in \F^{n \times n}.
\end{equation}
The Nystr{\"o}m approximation algorithm manipulates the approximation~\eqref{eqn:nystrom-approx} to obtain its eigenvalue decomposition $\widehat{\mA} = \mU \mLambda \mU^\top$.  A numerically stable implementation
of this procedure demands care.%
\Cref{alg:nystrom} presents reliable pseudocode, adapted from \cite{li17,tropp17b}.  The algorithm performs a single pass over the input matrix $\mA$.  It requires one sketch of the input matrix, a sketch of intermediate data, and $\order(n k^2)$ additional arithmetic operations.
See \cref{table:dense-rates}.

To analyze the Nystr{\"o}m approximation,%
we exploit an algebraic relationship with the RSVD approximation~\cite{gittens11,gittens16}.  This equivalence, recently termed the \textit{Gram correspondence} \cite{Epp25a}, allows us to transfer \emph{Frobenius-norm} error bounds for the RSVD to \emph{nuclear-norm} guarantees for the Nystr\"om approximation.

\begin{importedtheorem}[Gram correspondence \protect{\cite[Lem.~1]{gittens16}}]
\label{thm:gram-corr}
    Let \(\mA\in\F^{n \times n}\) be a psd matrix.  Choose an arbitrary test matrix \(\mOmega\in\F^{n \times k}\), and form the orthonormal matrix $\mQ = \operatorname{orth}(\mA^{1/2} \mOmega)$.  Then the Nystr{\"o}m approximation $\widehat{\mA} \in \F^{n \times n}$, defined in~\eqref{eqn:nystrom-approx}, satisfies
    \[
    \norm{ \mA^\ptop - \smash{\widehat{\mA}} }_{*}
        = \bignorm{{\mA^{1/2} - \mQ\mQ^\top\mA^{1/2}}}_{\rm F}^2.
    \]
    \end{importedtheorem}

Put simply, the Nystr\"om approximation applied to $\mA$ is equivalent to the RSVD applied to $\mA^{1/2}$.  Combining the Gram correspondence (\cref{thm:gram-corr}) and the RSVD analysis (\cref{thm:rsvd-via-osi}),%
we immediately obtain an error bound for the single-pass Nystr\"om
approximation.
This argument is entirely standard~\cite{gittens11,gittens16,tropp17b,MT20,tropp23,Epp25a,persson25}.

\begin{corollary}[Randomized Nystr\"om approximation with an OSI]
    \label{cor:nystrom-via-osi}
    \newtext{Let \(\mA\in\F^{n \times n}\) be a psd input matrix, and choose a target rank \(r \leq n\).}
    Draw a random $(r,\alpha)$-OSI test matrix \(\mOmega \in \F^{n \times k}\) where $k \geq r$.
    Construct the rank-$k$ psd Nystr{\"o}m approximation
    $\widehat{\mA} \in \F^{n \times n}$
    determined by~\eqref{eqn:nystrom-approx}.
    Then, with probability at least \(\frac9{10}\),
    \[
        \norm{\mA^\ptop - \smash{\widehat{\mA}}}_{*} \leq%
        (\rC / \alpha) \cdot \norm{\mA-\lra{\mA}{r}}_*.
    \]
\end{corollary}

\subsection{Sketch-and-solve for least-squares regression}
\label{sec:lsq-via-osi}

\begin{algorithm}[t]
    \caption{Sketch-and-solve with structured sketching}
    \label{alg:sketch-and-solve}
    \begin{algorithmic}[1]
        \Require Design matrix $\mA \in \F^{n\times d}$, response matrix $\mB \in \F^{n\times m}$, sketching dimension $p \geq d$
        \Ensure Approximate solution $\tilde\mX \in \F^{d \times m}$ to the least-squares problem~\eqref{eq:least-squares}         \State Draw a \textit{structured} random test matrix $\mPsi \in \F^{n \times p}$ \Comment{ \textit{e.g., SparseStack, Khatri--Rao, etc.}}
        \State $\tilde\mA \leftarrow \mPsi^\top \mA$ and $\tilde\mB \leftarrow \mPsi^\top \mB$ \Comment{Sketch the problem data}
        \State $[\mU,\mSigma,\mV] \gets \texttt{svd\_econ}(\tilde\mA)$
        \State $r\gets |\{ i : \mSigma(i,i) > 5\varepsilon_{\rm mach}\mSigma(1,1) \}|$ \Comment{Numerical rank of \mA}
        \State $\tilde\mX \gets \mV(:,1:r)(\mSigma(1:r,1:r)^{-1} (\mU(:,1:r)^\top\tilde\mB))$ \Comment{Truncated pseudoinverse}
    \end{algorithmic}
\end{algorithm}

Sketch-and-solve~\cite{sarlos06} is a standard randomized algorithm
for least-squares problems.  If we use an OSI to
implement the sketch-and-solve method,
the residual error lies within a constant factor of optimal.

Suppose that we are given a (tall) design matrix \(\mA\in\F^{n \times d}\) with $n \geq d$ and a response matrix \(\mB\in\F^{n \times m}\).  We wish to solve%
an overdetermined least-squares problem:
\begin{equation} \label{eq:least-squares}
    \mX_{\star} \in \argmin\nolimits_{\mX}\ \norm{\mA\mX-\mB}_{\rm F}^2.
\end{equation}
To quickly obtain an approximate solution $\tilde\mX \in \F^{d \times m}$, we can employ the sketch-and-solve method.  Draw a random test matrix $\mPsi \in \F^{n \times p}$, and extract the two sketches \(\mPsi^\top\mA \in \F^{p \times d}\) and \(\mPsi^\top\mB \in \F^{p \times m}\).  We report the minimum-Frobenius-norm solution to the sketched least-squares problem:
\begin{equation} \label{eq:least-squares-ss}
    \tilde\mX \coloneqq (\mPsi^\top\mA)^\dagger (\mPsi^\top \mB) \in \argmin\nolimits_{\mX}\ \norm{\mPsi^\top (\mA \mX- \mB)}_{\rm F}^2.
\end{equation}
See \cref{alg:sketch-and-solve} for numerically stable pseudocode.
This algorithm requires sketching both \mA and \mB,
plus $\mathcal{O}(d^2 (p+m) )$ arithmetic operations.
See \cref{table:dense-rates}.

To analyze the sketch-and-solve procedure, we make a change of variables.
Construct a reduced SVD $\mA = \mQ \mSigma \mV^\top$,%
where%
the square diagonal matrix $\mSigma$ is invertible.
After the transformation $\mSigma \mV^{\top} \mX \mapsto \mZ$,
the original least-squares problem \eqref{eq:least-squares}
and the sketched-least-squares problem~\eqref{eq:least-squares-ss} become%
\begin{align}
    \mZ_\star \coloneqq \mQ^\top \mB \in &\argmin\nolimits_{\mZ}\ \norm{\mQ\mZ - \mB}_{\rm F}^2;  \label{eq:true-ls} \\
    \tilde\mZ \coloneqq (\mPsi^\top\mQ)^\dagger (\mPsi^\top\mB) \in &\argmin\nolimits_{\mZ}\ \norm{ \mPsi^\top (\mQ \mZ - \mB)}_{\rm F}^2. \label{eq:true-ls-ss}    
\end{align}
The (unique) solution $\tilde{\mZ}$ to~\eqref{eq:true-ls-ss} induces
the minimum-norm solution $\tilde\mX = \mV \mSigma^{-1} \tilde\mZ$
to the sketched least-squares problem~\eqref{eq:least-squares-ss}.
The next lemma provides an expression for the difference between
the solutions of~\eqref{eq:true-ls} and~\eqref{eq:true-ls-ss}.
The argument is adapted from \cite[Lemma A.4]{tropp2017a},
where it appears in the analysis of the generalized Nystr{\"o}m
approximation.

\begin{lemma}[Sketch-and-solve error formula \protect{\cite[Lemma A.4]{tropp2017a}}] \label{lem:sketch-and-solve-error-formula}
                    
    Instate the prevailing notation.  Let $\mQ_{\perp}$ be an orthonormal
    matrix that satisfies $\mQ^{\vphantom{\top}}_{\perp} \mQ_{\perp}^\top = \mI - \mQ \mQ^\top$.
    Define the matrices \(\mPsi_1 \defeq \mPsi^\top \mQ \)
    and \(\mPsi_2 \defeq \mPsi^\top\mQ_\perp \).         Provided that \(\mPsi_1\) has full column rank,     \[
        \tilde\mZ - \mZ_\star = \mPsi_1^\dagger \mPsi^{\vphantom{\dagger}}_2\mQ_\perp^\top\mB.
    \]

\end{lemma}
\begin{proof}
Starting from the definition~\eqref{eq:true-ls-ss} of $\tilde\mZ$, we calculate that%
\begin{align*}
    \tilde\mZ
    &= (\mPsi^\top\mQ)^\dagger\mPsi^\top\mB \\
    &= \mPsi_1^\dagger\mPsi^\top(\mQ\mQ^\top+\mQ^{\vphantom{\dagger}}_\perp\mQ_\perp^\top)\mB \tag*{since \(\mI = \mQ\mQ^\top + \mQ^{\vphantom{\dagger}}_\perp\mQ_\perp^\top\)} \\
    &= \mPsi_1^\dagger\mPsi^{\vphantom{\dagger}}_1 \mQ^\top\mB+ \mPsi_1^\dagger\mPsi^{\vphantom{\dagger}}_2\mQ_\perp^\top\mB \\
    &= \mQ^\top\mB + \mPsi_1^\dagger\mPsi^{\vphantom{\dagger}}_2\mQ_\perp^\top\mB \tag*{since \(\mPsi_1^\dagger\mPsi_1=\mI\)}
    = \mZ^{\vphantom{\dagger}}_\star + \mPsi_1^\dagger\mPsi^{\vphantom{\dagger}}_2\mQ_\perp^\top\mB.
\end{align*}
This is the advertised statement.
         \end{proof}

With this lemma in place, we can now prove the central result on sketch-and-solve.
Several prior works~\cite{woodruff2014sketching,drineas11,parulekar21,Clarkson13} have obtained similar guarantees under alternative hypotheses on $\mPsi$.
The standard analysis assumes that $\mPsi$
is an OSE for the range of the concatenated matrix $[\mA ~ \mB]$.
Parulekar et al.~\cite{parulekar21} assume that
$\mPsi$ is an isotropic matrix that
is also an OSE for the range of $\mA$.
Drineas et al.~\cite{drineas11}%
assume that $\mPsi$ is a subspace injection for the range of $\mA$ and that $\mPsi$ has an approximate matrix multiplication property.
We are not aware of any sources which recognize that the OSI property is enough.

\begin{theorem}[Sketch-and-solve with an OSI]
    \label{thm:lsq-via-osi}
    \newtext{Let} \(\mA\in\F^{n \times d}\) and \(\mB\in\F^{n \times m}\).
    Draw a random $(d, \alpha)$-OSI test matrix \(\mPsi\in\F^{n \times p}\) where $p \geq d$, and let \(\tilde\mX\in\F^{d \times m}\) be the output~\eqref{eq:least-squares-ss} of the sketch-and-solve procedure.                     Then, with probability at least $\tfrac{9}{10}$,
    \[
          \norm{\mA\tilde\mX - \mB}_{\rm F}^2 \leq%
          (\rC / \alpha) \cdot \min\nolimits_{\mX}\ \norm{\mA\mX-\mB}_{\rm F}^2.%
    \]
\end{theorem}
\begin{proof}
    We adopt the notation and the reparameterization outlined above.
    The goal is to show that
    \[
        \norm{\mQ\tilde\mZ - \mB}_{\rm F}^2
        \leq%
        (\rC / \alpha) \cdot \norm{\mQ\mZ_\star - \mB}_{\rm F}^2.
    \]
    The optimal least-squares residual \(\mQ\mZ_\star-\mB\) is orthogonal to the range of \mQ, so%
    \begin{align*}
        \norm{\mQ\tilde\mZ-\mB}_{\rm F}^2
        &= \norm{\mQ\mZ_\star-\mB}_{\rm F}^2 + \norm{\mQ\tilde\mZ-\mQ\mZ_\star}_{\rm F}^2 \\
        &= \norm{\mQ\mZ_\star-\mB}_{\rm F}^2 + \norm{\tilde\mZ-\mZ_\star}_{\rm F}^2
        = \norm{\mQ\mZ_\star-\mB}_{\rm F}^2 + \norm{\mPsi_1^\dagger\mPsi_2^\ptop\mQ_\perp^\top\mB}_{\rm F}^2.
    \end{align*}
    The last identity follows from \cref{lem:sketch-and-solve-error-formula}.
    Finally, apply \cref{lem:osi-preserve-orthogonality}
    to deduce that
    \[
        \norm{\mPsi_1^\dagger\mPsi_2^\ptop\mQ_\perp^\top\mB}_{\rm F}^2
        \leq%
        (\rC / \alpha) \cdot \norm{\mQ_\perp^\top\mB}_{\rm F}^2
        = (\rC / \alpha) \cdot \norm{\mQ_\perp^\ptop\mQ_\perp^\top\mB}_{\rm F}^2
        = (\rC / \alpha) \cdot \norm{\mQ\mZ_\star-\mB}_{\rm F}^2,
    \]
    where the bound holds with probability at least \(\frac{9}{10}\).
\end{proof}
\newtext{
    If \mA does not have full column rank, we can take the matrix \mQ to have \(\rank(\mA)<d\) columns, so we may replace the \((d,\alpha)\)-OSI requirement with the weaker \((\rank(\mA),\alpha)\)-OSI requirement.
}

\subsection{Generalized Nystr\"om approximation}
\label{sec:gen-nystrom-via-osi}

\begin{algorithm}[t]
    \caption{Generalized Nyström with structured sketching: Output in outer-product form}
    \label{alg:gen_nystrom_outer}
    \begin{algorithmic}[1]
        \Require Matrix $\mA \in \F^{n \times d}$, approximation rank $k$
        \Ensure Matrices $\mF \in \F^{n \times k}$ and $\mG \in \F^{d \times k}$
        determining $\hat\mA=\mF\mG^\top$
        \State Draw \emph{structured} random test matrices $\mOmega \in \F^{d \times k}$ and $\mPsi \in \F^{n \times p}$ \Comment{Recommendation: $p = \lceil 1.5 k\rceil$}
        \State $\mY \leftarrow \mA\mOmega$ and $\mX \leftarrow \mA^\top\mPsi$ \Comment{Compute sketches}
        \State $[\mU,\mSigma,\mV] \gets \texttt{svd\_econ}(\mX^\top \mOmega)$
        \State $r\gets |\{ i : \mSigma(i,i) > 5\varepsilon_{\rm mach}\mSigma(1,1) \}|$ \Comment{Numerical rank of $\mPsi^\top \mY$}
        \State $\mF \gets \mY\mV(:,1:r)\mSigma(1:r,1:r)^{-1}$ and $\mG \gets \mX\mU(:,1:r)$ 
    \end{algorithmic}
\end{algorithm}

\begin{algorithm}[t]
    \caption{Generalized Nyström with structured sketching: Output in SVD form}
    \label{alg:gen_nystrom_svd}
    \begin{algorithmic}[1]
        \Require Matrix $\mA \in \F^{n \times d}$, approximation rank $k$
        \Ensure Orthonormal $\mU \in \F^{n \times k}$, $\mV \in \F^{d \times k}$ and nonnegative diagonal $\mSigma \in \F^{k \times k}$
                determining $\hat\mA= \mU\mSigma\mV^\top$
        \State Draw \emph{structured} random test matrices $\mOmega \in \F^{d \times k}$ and $\mPsi \in \F^{n \times p}$ \Comment{Recommendation: $p = \lceil 1.5 k\rceil$}
        \State $\mY \leftarrow \mA\mOmega$ and $\mX \leftarrow \mA^\top\mPsi$ \Comment{Compute sketches}
        \State $[\mQ,\sim] \leftarrow \texttt{qr\_econ}(\mY)$ and $[\mP,\mT] \leftarrow \texttt{qr\_econ}(\mX)$ \Comment{Orthogonalize}
        \State $[\mU_1,\mSigma_1,\mV_1] \gets \texttt{svd\_econ}(\mPsi^\top\mQ)$
        \State $r\gets |\{ i : \mSigma_1(i,i) > 5\varepsilon_{\rm mach}\mSigma_1(1,1) \}|$ \Comment{Numerical rank of $\mPsi^\top\mQ$}
        \State $\mC \leftarrow \mV_1(:,1:r)(\mSigma_1(1:r,1:r)^{-1}(\mU_1(:,1:r)^\top\mT^\top))$ \Comment{$\mC = (\mPsi^\top\mQ)^\dagger \mT^\top $}
        \State $[\widehat{\mU},\mSigma,\widehat{\mV}] \leftarrow \texttt{svd\_econ}(\mC)$
        \State $\mU \leftarrow \mQ\widehat{\mU}\vphantom{\widehat{\widehat{\mU}}}$ and $\mV \leftarrow \mP\widehat{\mV}$
    \end{algorithmic}
\end{algorithm}

Last, we treat low-rank approximation of a matrix via the generalized Nystr{\"o}m method~\cite{woolfe08,Clarkson09,Naka2020,tropp2017a}.

Fix an (arbitrary) input matrix $\mA \in \F^{n \times d}$
and an approximation rank parameter $k \leq \min \{n, d\}$.
We draw two random test matrices $\mOmega \in \F^{d \times k}$
and $\mPsi \in \F^{n \times p}$ where the embedding dimensions $k \leq p \leq \min\{n, d\}$.
Collect two sketches of the input matrix: $\mY \coloneqq \mA \mOmega \in \F^{n \times k}$ and $\mX \coloneqq \mPsi^\top \mA \in \F^{p \times d}$.  The generalized Nystr{\"o}m approximation takes the form
\begin{equation} \label{eq:generalized-nystrom}
	\hat\mA \coloneqq \mY (\mPsi^\top \mY)^\dagger \mX \in \F^{n \times d}.
\end{equation}

We present two numerically stable implementations of~\eqref{eq:generalized-nystrom} that prepare the approximation
in outer-product form (\cref{alg:gen_nystrom_outer})
or in compact SVD form (\cref{alg:gen_nystrom_svd}).
Each algorithm employs two sketches of the input matrix,
as well as a sketch of intermediate data.  The outer-product algorithm
requires $\cO(n k^2 + dpk)$ extra arithmetic,
while the SVD algorithm requires $\cO(nk^2 + dp^2)$ extra arithmetic.
Since $p \geq k$, both algorithms are more efficient
when $\mA$ is wide ($n \geq d$);
otherwise, it is faster to apply them to the adjoint $\mA^\top$.
See \cref{table:dense-rates} for a summary of arithmetic costs.

To motivate the construction of the generalized Nystr\"om approximation \cref{eq:generalized-nystrom} and our analysis,
we adopt the perspective from \cite{woodruff2014sketching,tropp2017a,Naka2020}
that the RSVD approximation~\eqref{eqn:rsvd-approx} can be written as
\begin{equation} \label{eqn:rsvd-approx-gn}
    \hat\mA_{\rm SVD} = (\mA\mOmega)(\mA\mOmega)^\dagger \mA.
\end{equation}
We recognize the matrix $(\mA\mOmega)^\dagger \mA$ as a solution to the least-squares problem
\begin{equation*}%
    (\mA\mOmega)^\dagger \mA \in \argmin\nolimits_{\mZ \in \F^{k\times d}}\ \norm{(\mA\mOmega)\mZ - \mA}_{\rm F}.
\end{equation*}
The generalized Nystr\"om method can be derived by approximating the solution of this least-squares problem with sketch-and-solve, resulting in a further approximation
\begin{equation} \label{eq:sketched-ls-gen-nys}
    (\mA\mOmega)^\dagger \mA \approx
    (\mPsi^\top \mA\mOmega)^\dagger (\mPsi^\top \mA)
    \in \argmin\nolimits_{\mZ \in \F^{k\times d}} \ \norm{(\mPsi^\top \mA\mOmega)\mZ - \mPsi^\top \mA}_{\rm F}.
\end{equation}
Hence, the generalized Nystr{\"o}m approximation~\eqref{eq:generalized-nystrom} is a proxy for the RSVD approximation~\eqref{eqn:rsvd-approx-gn}.

With this observation, the error bound for the generalized Nystr\"om approximation is an immediate corollary of our analyses of RSVD (\cref{thm:rsvd-via-osi}) and sketch-and-solve (\cref{thm:lsq-via-osi}).
Our argument parallels~\cite[Thm.~4.3]{tropp2017a},
although the cited work only treats Gaussian test matrices.

\begin{corollary}[Generalized Nystr\"om with OSIs]
    \label{thm:gen-nystrom-via-osi}
    \newtext{Let \(\mA\in\F^{n \times d}\) be an input matrix, and choose a target rank \(r \leq \min\{n, d\}\).}
    Draw a random $(r,\alpha)$-OSI test matrix \(\mOmega \in \newtext{\bbF}^{d \times k}\) and a random $(k,\alpha)$-OSI test matrix \(\mPsi \in \newtext{\bbF}^{n \times p}\)
    where \newtext{\mPsi is independent of \mOmega and} $r \leq k \leq p$.
    Construct the generalized Nystr{\"o}m approximation $\widehat{\mA} \in \F^{n \times d}$ determined by~\eqref{eq:generalized-nystrom}.
    Then, with probability at least \(\frac45\),
    \begin{equation} \label{eq:gen-nys-guarantee}
        \norm{\mA-\hat\mA_\ptop}_{\rm F}^2 \leq%
        (\rC / \alpha)^2 \cdot \norm{\mA-\lra{\mA}{r}}_{\rm F}^2.
    \end{equation}
\end{corollary}

\begin{proof}
    Using the interpretation \cref{eq:sketched-ls-gen-nys} of generalized Nystr\"om as the solution to a sketched least-squares problem,
    the bound follows from our results for sketch-and-solve and RSVD with OSIs:
    \begin{align*}
        \norm{\hat\mA - \mA}_{\rm F}^2 &= \norm{(\mA\mOmega)(\mPsi^\top \mA\mOmega)^\dagger (\mPsi^\top \mA) - \mA}_{\rm F}^2 \\
        &\le (\rC / \alpha) \cdot \norm{(\mA\mOmega)(\mA\mOmega)^\dagger \mA - \mA}_{\rm F}^2 \tag{Sketch-and-solve, \cref{thm:lsq-via-osi}}\\
        &\le (\rC / \alpha)^2 \cdot \norm{\mA - \llbracket\mA\rrbracket_r}_{\rm F}^2. \tag{RSVD, \cref{thm:rsvd-via-osi}}
    \end{align*}
    The sketch-and-solve bound and the RSVD bound each holds with probability $\tfrac{9}{10}$, so they hold simultaneously with probability $\tfrac{4}{5}$.                                                                 \end{proof}

\newtext{To interpret this result, we may define the \emph{oversampling factor} of a class of $(r,\alpha)$-OSI matrices $\mOmega \in \bbF^{d\times k}$ as the ratio $\gamma \coloneqq k/r$.
Using this terminology, the result states that if $k\ge \gamma r$ and if $p \ge \gamma k \ge \gamma^2 r$, then the generalized Nystr\"om approximation satisfies the guarantee \cref{eq:gen-nys-guarantee}.
In particular, we must oversample twice: The first embedding dimension $k$ must be a factor $\gamma$ larger than the target rank $r$, and the second embedding dimension must be a factor $\gamma$ larger than the first embedding dimension.
In case
the oversampling parameter \(\gamma = \cO(1)\), as with the SparseStack (\cref{def:sparse-stack}) and SparseRTT (\cref{def:intro-sparsertt}) test matrices, 
the larger embedding dimension still satisfies \(p = \cO(r)\).}

\begin{remark}[Truncated generalized Nystr{\"o}m]
For test matrices with oversampling \(\gamma \gg 1\),
such as an SRTT \cite{tropp11SRHT} or a Khatri--Rao construction 
(\cref{def:kr-intro}),%
the larger embedding dimension \(p = \gamma^2 r \) may be prohibitively expensive.
In theory, we can avoid this cost by using
the \emph{truncated} generalized
Nystr{\"o}m approximation~\cite[Thm.~4.3]{woodruff2014sketching},
which we can analyze by adapting the proof of~\cref{thm:gen-nystrom-via-osi}.  This algorithm allows the test matrix $\mPsi$ to have embedding dimension $p = \cO(\gamma r)$.
In practice, we cannot endorse this algorithm
because the truncation step severely impairs the quality of the resulting low-rank approximation.
Designing an algorithm with the theoretical guarantees of truncated generalized Nystr\"om estimator and the accuracy of the standard generalized Nystr\"om estimator is an open problem.
\end{remark}

\subsection{\newtext{Discussion: Constant-factor guarantees}} \label{sec:constant-factor}

\Cref{thm:rsvd-via-osi,thm:lsq-via-osi,thm:gen-nystrom-via-osi,cor:nystrom-via-osi} demonstrate
that several sketching algorithms \newtext{enjoy error bounds} within a \emph{constant factor} of the optimal error bound.
Many papers in the randomized linear algebra literature
seek stronger approximation guarantees, within a factor
of $1 + \eps$ of the optimal error, while the algorithmic
resources exhibit a $\mathrm{poly}(\eps^{-1})$ dependence
in the parameter $\eps > 0$.
\newtext{Following the initial release of our manuscript, Townsend \& Wang demonstrated that the OSI condition alone does not suffice to obtain
the stronger $1 + \varepsilon$ guarantees \cite{townsend26}.
However, OSIs equipped with an additional \emph{approximate matrix multiplication property} can achieve this type of bound;
the latter claim follows from arguments in \cite[Thm.~4.2]{woodruff2014sketching} and \cite[Lem.~1]{drineas11}.}

Regardless, a constant-factor error bound suffices
for typical applications.
Indeed, it is most fruitful
to deploy sketching%
algorithms%
in settings where the
optimal error is tiny~\cite[Sec.~1.1]{nakatsuksa24}.
To reach high accuracy (\ie, small $\eps$),
practitioners do not use simple sketching algorithms.
For low-rank approximation,%
block Krylov methods~\cite{musco2015randomized,tropp23,CEM+26} are preferred.
For linear regression, sketch-and-precondition methods~\cite{rokhlin08,avron10,epperly24a}
are the standard.

\subsection{\newtext{Discussion: OSI vs.\ OSE}}

\newtext{This paper argues that the OSI condition justifies the use of sketching algorithms for most practical applications.
One exception concerns algorithms that use sketching to precondition or to approximately orthogonalize a matrix \cite{rokhlin08,avron10,melnichenko23,FGL21,BG22,Ips25}.
The analysis of these algorithms appears to require
    bounds on both the injectivity and dilation parameters.}

\section{Sparse test matrices}
\label{sec:sparse-sketching}

\Cref{sec:randnla-via-osi} shows that we can implement
several sketching algorithms with OSIs.
As a first example, we study the implications
for randomized linear algebra with sparse test matrices, \newtext{which achieve excellent speed for dense and sparse linear algebra tasks in both theory and practice}~\cite{Epp25a,epperly24a,dong23,chen25,melnichenko23,tropp19,Clarkson13,nelson13,cohen16}.
See \cref{fig:sparseStack-speed}.

This section supplements the discussion of sparse
test matrices from~\cref{sec:sparse-intro}
with additional information about theoretical guarantees,
implementation, and empirical performance.
\Cref{sec:rnla-with-sparse} reviews the SparseStack
construction; \cref{sec:sparse-alternatives}
presents alternative constructions;
and \cref{sec:sparse-experiments} provides
an empirical study. \newtext{For reference, a CountSketch test matrix is the special case
$\zeta=1$ of the SparseStack construction; see
\cref{eq:countsketch} for further details}.

\subsection{Randomized linear algebra with SparseStacks} \label{sec:rnla-with-sparse}

We highly recommend the SparseStack test matrix,
introduced in \cref{def:sparse-stack},
for%
randomized linear algebra.
\newtext{As a reminder, a SparseStack is a random matrix where each row is blocked into $\zeta$ equally sized pieces, each containing a single nonzero entry with uniformly random value $\pm \zeta^{-1/2}$ in a uniformly random position.
For complex variables, a natural variant is to draw the nonzero values instead from a complex Rademacher or Steinhaus distribution scaled by the factor $\zeta^{-1/2}$.
Our theoretical guarantees remain valid for these variants as well.}

\Cref{thm:sparse-stack-osi}, adapted from~\cite[Rem.~6.5]{tropp25},
states that the SparseStack serves as an $(r, \nicefrac{1}{2})$-OSI
for some embedding dimension $k = \order(r)$ and row sparsity $\zeta = \order(\log r)$.
When we implement \cref{alg:rsvd,alg:nystrom,alg:sketch-and-solve,alg:gen_nystrom_outer,alg:gen_nystrom_svd}
using SparseStacks with these parameter choices,
we can certify that the algorithms run faster than
previously known.  
\Cref{table:sparse-runtimes-intro,table:sparse-runtimes}
provide a comparison with existing work.

In addition to the OSI guarantee and the superior runtime bounds,
the SparseStack test matrix has several other merits
(and one demerit):

\begin{description} \setlength{\itemsep}{0pt}
\item[\textbf{($+$) Construction.}]  We form the SparseStack by drawing and storing $2d\zeta$ simple random variables.  The storage cost is fixed in advance by the parameters, which aids with memory management.  In addition, the SparseStack is \emph{extensible},
which means that we can increase the ambient dimension $d$
on the fly,
just by tacking on additional independent rows of the same form.

\item[\textbf{($+$) Reliability.}]  The SparseStack is extremely reliable in practice.  Numerical experiments in \cref{fig:sparseStack-injectivity-intro} \newtext{suggest} that \emph{constant} row sparsity $\zeta = 4$ and proportional embedding dimension $k = 2r$ suffice to obtain an OSI \newtext{for subspace dimensions up to $r=10^7$}.  See \cref{sec:sparse-experiments}.

\item[\textbf{($-$) Library support.}]
SparseStack matrices are practical only when
they are implemented with high-quality sparse linear
algebra libraries.  Sadly, these software libraries are less widely available than one would hope~\cite{sparseblas}. When fast sparse linear algebra primitives are unavailable, we recommend the SparseRTT test matrix (\cref{def:intro-sparsertt}) as an alternative.
\end{description}

\noindent
When fast sparse linear algebra primitives are not available,
we recommend the SparseRTT test matrix (\cref{def:intro-sparsertt})
as an alternative to the SparseStack.

\begin{table}[t]
\centering \small
\begin{tabular}{@{}lllll@{}} \toprule
    Guarantee & Test matrix & Embedding dim.\ \(k\) & Row sparsity \(\zeta\) & Gen.\ Nystr\"om runtime \\ \midrule
    OSE & CountSketch \cite{meng13,nelson13} & \(\cO(r^2)\) & \(1\) & \(\cO(\nnz(\mA) + nr^4)\) \\[0.25em]
    & SparseStack \cite{cohen16} & \(\cO(r \log r)\) & \(\cO(\log r)\) & \(\cO(\nnz(\mA)\log(r) + nr^2\log^3 r)\) \\[0.25em]
    & SparseStack \cite{chenakkod24} & \(\cO(r)\) & \(\cO(\log^3 r)\) & \(\cO(\nnz(\mA)\log^3(r) + nr^2)\) \\[0.25em] \midrule
    OSI & SparseIID \cite{tropp25} & \(\cO(r)\) & \(\cO(\log r)\) & \(\cO(\nnz(\mA)\log(r) + nr^2)\) \\[0.25em]
    & SparseStack (Thm.~\ref{thm:sparse-stack-osi}) & \(\cO(r)\) & \(\cO(\log r)\) & \(\cO(\nnz(\mA)\log(r) + nr^2)\) \\[0.25em] \midrule
    None & SparseStack (Practice) & \(2r\) & \(4\) & \(\cO(\nnz(\mA) + nr^2)\) \\[0.25em]
    \bottomrule
\end{tabular}
\caption{
    \textbf{Sparse test matrices: Comparison.}
    Theoretical results for sparse test matrices,
    including the type of guarantee (OSE or OSI), the embedding dimension $k$ and row sparsity $\zeta$ for sketching an $r$-dimensional subspace, and the runtime for generalized Nystr\"om (\cref{alg:gen_nystrom_outer}).
    OSI guarantees for the SparseStack test matrix (\cref{def:sparse-stack}) permit \emph{both} the smaller row sparsity $\zeta = \order(\log r)$ of Cohen's result \cite{cohen16} \emph{and} the smaller sketch dimension $k = \order(r)$ of Chenakkod \etal \cite{chenakkod25}.
}
\label{table:sparse-runtimes}
\end{table}

\subsection{Alternative sparse test matrices}
\label{sec:sparse-alternatives}

The literature contains many constructions for sparse dimension reduction maps; for instance, see \cite{chenakkod24,KN14,dasgupta10,tropp25,Clarkson13,meng13,nelson13}.
This section reviews two natural alternatives: SparseUniform (\cref{sec:sparse-uniform}) and SparseIID (\cref{sec:stack-vs-iid}).
We state reasons for (and against) preferring the SparseStack construction.

\subsubsection{The SparseUniform test matrix}
\label{sec:sparse-uniform}

By and large,
practitioners~\cite{murray23,epperly24a,chen25,dong23,tropp19}
have worked with the \emph{SparseUniform} test matrix,
also proposed by Kane \& Nelson~\cite[Fig.~1(b)]{KN14}.
\begin{definition}[SparseUniform]\label{def:sparse_unif}
    \newtext{
        Let \(d\) denote the ambient dimension, \(k\) the embedding dimension, and \(\zeta\) the row sparsity.
    }
The \emph{SparseUniform} test matrix $\mOmega\in\F^{d \times k}$ is a random sparse matrix with exactly $\zeta$ independent copies of the random variable \(\zeta^{-1/2}\rad\) in each row,
placed in uniformly random positions,
sampled without replacement,
where \(\rad\sim\textsc{rademacher}\).
\end{definition}

\Cref{fig:sparsity_structure} illustrates how the sparsity patterns
of the SparseUniform and SparseStack constructions differ.
In our experience, the SparseUniform and SparseStack test matrices have similar injectivity $\alpha$ and dilation $\beta$, and they attain similar accuracy on downstream tasks, such as low-rank approximation.

Although the performance is comparable,
we prefer the SparseStack test matrix for two reasons.
First, the SparseUniform construction involves repeated sampling
of $\zeta$ indices \emph{without replacement},
which is hard to implement efficiently~\cite{chen25,Epp25a,murray23}.
Second, we presently have stronger theoretical guarantees for the SparseStack construction.
Indeed, the best OSI result for a SparseStack (\cref{thm:sparse-stack-osi})
allows for an embedding dimension $k = \cO(r)$
and a row sparsity $\zeta = \cO(\log r)$,
while SparseUniform admits an OSE guarantee~\cite{cohen16}
for some embedding dimension $k = \cO(r \log r)$
and row sparsity $\zeta = \cO(\log r)$.
\rev{We believe that the SparseUniform model should satisfy the same guarantees as the SparseStack model, 
but we lack theoretical tools that can certify this conjecture.}

\begin{figure}[t]
    \centering
    \resizebox{0.7\linewidth}{!}{%
    \input{figures/results/fig4.pgf}%
    }
    \caption{\textbf{Sparse test matrices: Sparsity patterns.} Illustration of the sparsity patterns for several sparse sketching matrices with $k = 9$ columns and row sparsity $\zeta = 3$. From left to right: SparseStack (\cref{def:sparse-stack}), SparseUniform (\cref{def:sparse_unif}), and SparseIID (\cref{def:sparse-iid}). \label{fig:sparsity_structure}
    }
\end{figure}

\subsubsection{The SparseIID test matrix}

\label{sec:stack-vs-iid}

Another alternative%
is the SparseIID model,
a construction with roots in the work of Achlioptas~\cite{achlioptas03}.

\begin{definition}[SparseIID]
    \label{def:sparse-iid}
    \newtext{Let \(d\) denote the ambient dimension, \(k\) the embedding dimension, and \(\zeta\in(0,k]\) the \emph{expected} row sparsity.}
    Define a random variable $\varphi \coloneqq \zeta^{-1/2} \varrho \delta$,
    where \(\rad\sim\textsc{rademacher}\) and 
    $\delta \sim \textsc{bernoulli}(\zeta/k)$.
    The entries of the \emph{SparseIID} test matrix $\mOmega\in\F^{d \times k}$ are iid copies of $\varphi$.
\end{definition}

\cref{fig:sparsity_structure} illustrates the sparsity pattern
of the SparseIID test matrix.
Tropp~\cite[Thm.~6.3]{tropp25} proved that the
SparseIID model admits an OSI guarantee
when the expected row sparsity $\zeta$
is adapted to the \emph{subspace coherence},
defined in \cref{sec:prelims}.

\begin{importedtheorem}[SparseIID matrices are OSIs, \protect{\cite[Thm.~6.3]{tropp25}}]
    \label{impthm:sparse-iid-coherence}
    \newtext{
        Let \(\mQ\in\bbF^{d \times r}\) be an orthogonal matrix with coherence \(\mu(\mQ)\in [\nicefrac rd, 1]\).%
    }
    A SparseIID test matrix $\mOmega \in \F^{d \times k}$
    serves as an $(r, \nicefrac{1}{2})$-OSI for some
    embedding dimension $k = \cO(r)$ and some
    expected row sparsity $\zeta = \cO(\mu(\mQ) \log r)$.
\end{importedtheorem}

\noindent
Because of the role of the subspace coherence,
\cref{impthm:sparse-iid-coherence} for SparseIID
provides a \emph{stronger guarantee}
than \cref{thm:sparse-stack-osi} offers for SparseStack.

Nevertheless, we deprecate the use of the
SparseIID construction because of three significant flaws.
First, it is hard to develop efficient implementations
of the iid sampling procedure that have expected
runtime $\cO(d \zeta)$.
Second, the storage costs are variable.
Third, the SparseIID model cannot serve as an OSI
with constant expected row sparsity $\zeta = \cO(1)$
because there is a high probability that the random
matrix has rows that equal zero;
this failure mode is visible in \cref{fig:sparse-map-comparison}.
As a consequence, we believe that the SparseStack
construction remains preferable.

\subsection{Numerical evidence} \label{sec:sparse-experiments}

\begin{figure}[!t]
    \centering
    \resizebox{0.95\linewidth}{!}{%
        \input{figures/results/fig5.pgf}%
    }
    \caption{
        \textbf{SparseStack versus SparseIID.}
        Relative error
        \(\norm{\mA - \hat\mA}_{\rm F}^2 / \norm{\mA}_{\rm F}^2\)
        for approximations $\hat\mA$ obtained by the RSVD (\cref{alg:rsvd})
        using a SparseStack (\cref{def:sparse-stack}), a SparseIID
        (\cref{def:sparse-iid}), or a Gaussian test matrix
        (\cref{def:gauss-test}).
        We consider four diagonal input matrices
        $\mA \in \R^{1024 \times 1024}$ from the testbed
        in~\cite{tropp17b,tropp2017a}.
        For each instance, we perform $100$ trials.
        The data markers track the median errors; shaded regions are bounded
        by the 10\% and 90\% quantiles.
        The SparseStack test matrix with row sparsity $\zeta=3$ or $\zeta=4$
        consistently matches the accuracy benchmark of a Gaussian test matrix
        for problems of this size, whereas the SparseIID test matrix often
        fails to produce accurate approximations when the spectrum decays
        rapidly.
    }
    \label{fig:sparse-map-comparison}
\end{figure}

\Cref{fig:sparse-map-comparison} charts the error
in the RSVD approximation (\cref{alg:rsvd}) using the
SparseStack test matrix with a range of row sparsity
parameters $\zeta$ and embedding dimensions $k$.
\newtext{A} SparseStack with $\zeta = 1$ \newtext{is known as a}
CountSketch
test matrix~\cite{Clarkson13,meng13,nelson13,charikar04}.
The plot also includes comparisons with SparseIID
test matrices and the Gaussian baseline.
The input matrices $\mA$ are diagonal matrices
from the testbed in~\cite{tropp17b,tropp2017a}, as
these instances are adversarial%
for sparse test matrices~\cite{nelson13b}.
The experiments suggest that \emph{constant} row sparsity
$\zeta = 4$ is sufficient for the SparseStack test
matrix to achieve the same low-rank approximation
error as a Gaussian test matrix \newtext{ for problems of this size}.

For more targeted evidence, \cref{fig:sparseStack-injectivity-intro}
tracks numerical approximations of the injectivity and dilation
parameters of a SparseStack matrix as the subspace
dimension $r$ varies and the embedding dimension $k = 2r$.
We estimate both parameters by applying
the test matrix to the adversarial orthonormal matrix
\(\mQ = [\ve_1 ~ \cdots ~ \ve_r] \in \bbR^{d \times r}\),
which \rev{has maximal coherence $\mu(\mQ) = 1$
and is believed to be a challenging instance \cite{nelson13b}}.
When \(\zeta \geq 4\), the \emph{injectivity} of SparseStack
appears to plateau \rev{\emph{at a value bounded away from zero}} as \(r\) increases.
When $\zeta = 8$, the SparseStack \emph{injectivity} almost matches
the asymptotic injectivity $\alpha = (1 - 1/\sqrt{2})^2$
of a Gaussian test matrix (dashed line).
However, for all values of $\zeta$, the SparseStack \emph{dilation}
deviates from the asymptotic dilation $\beta = (1 + 1/\sqrt{2})^2$
of a Gaussian.
This disparity \newtext{shows how the OSI properties of SparseStack matrices are more favorable than their OSE properties}.

In \cref{sec:science-pod-modes}, we showcase a scientific application
of SparseStack matrices for the compression and analysis of a numerical simulation of a Bose--Einstein condensate.
In this context, the generalized Nystr{\"o}m method runs up to
\textbf{12$\times$ faster} with SparseStack test matrices,
as compared with the baseline cost using Gaussian matrices.

\section{SparseRTT: \newtext{Optimized} fast trigonometric transform test \newtext{matrices}} \label{sec:sparse_trig_transforms}

This section elaborates on the SparseRTT test matrix
(\cref{def:intro-sparsertt}).
Trigonometric transform-based test matrices,
such as the SparseRTT, are widely used in
practice~\cite{tropp2017a,Naka2020,nakatsuksa24},
and they are valuable for computing
environments that lack high-quality sparse arithmetic libraries.

\Cref{sec:sparsertt-construction} summarizes
the design and analysis of SparseRTTs;
\cref{sec:sparsertt-cost} contains an accounting
of their computational cost;
and \cref{sec:sparsertt-implementation} offers implementation guidance.
\Cref{sec:sparsertt-alternatives} outlines several related constructions,
along with numerical comparisons.
Last, \cref{sec:sparsertt-proof} contains a proof
of the main theorem on SparseRTTs.

\subsection{The SparseRTT construction and its analysis} \label{sec:sparsertt-construction}

In \cref{sec:fast-tranform-intro}, we introduced
random test matrices based on fast trigonometric transforms.
\newtext{These matrices $\mOmega = \mD\mF\mS$ are constructed by composing a diagonal matrix \(\mD\in\bbF^{d \times d}\), a fast trigonometric transform \(\mF\in\bbF^{d \times d}\), and a sparse random matrix \(\mS\in\bbR^{d \times k}\).}
Among constructions of this type,
we recommend using the SparseRTT test matrix, \newtext{as it achieves the} fastest
possible sketching time in \newtext{this} class
and it enjoys the \((r,\nicefrac12)\)-OSI property
at a proportional embedding dimension \(k=\cO(r)\).

In the SparseRTT construction (\cref{def:intro-sparsertt}), we populate the random diagonal matrix $\mD$ with iid \(\textsc{rademacher}\) random variables; the trigonometric transform $\mF$ is a WHT or DCT (when \(\newtext{\bbF} \in \{\bbR,\bbC\}\)) or a DFT (when \(\bbF = \bbC\)); the sparse matrix $\mS$ follows the \emph{SparseCol} distribution, defined below.

\begin{definition}[SparseCol]
    \label{def:sparsecol}
    \newtext{Let \(d\) denote the ambient dimension, \(k\) the embedding dimension, and \(\colsparse\in\bbN\) the \emph{column sparsity}.}
    Construct \newtext{the} random sparse vector
    \[
        \vomega \coloneqq \sqrt{\frac{d}{\colsparse}} \sum_{i=1}^\colsparse \rademacher_i \ve_{s_i} \in \F^{d}.
    \]
    The random variables $\rad_1,...,\rad_\xi \sim \textsc{rademacher}$ iid, and the selectors $s_1,...,s_\xi$ are sampled uniformly from $\{1, \dots, d\}$ \emph{without} replacement.
    As usual, $\mathbf{e}_i \in \F^d$ is the $i$th standard basis vector.%
    Up to scaling, the columns of the \emph{SparseCol} matrix
    \(\mS \in \F^{d \times k}\)
    are iid copies of the random sparse vector $\vomega$:
    \[
        \mS \coloneqq \frac1{\sqrt k} \begin{bmatrix}
            \vomega_1 & \cdots & \vomega_k \end{bmatrix}
            \in \F^{d \times k}
        \quad \text{where $\vomega_j \sim \vomega$ iid.}
    \]
\end{definition}

Our main result contains an OSI guarantee for
the SparseRTT test matrix.

\begin{theorem}[SparseRTTs are OSIs]
    \label{thm:sparsertt}
    The SparseRTT test matrix $\mOmega \in \F^{d \times k}$
    is an $(r,\nicefrac12)$-OSI for some  embedding dimension \(k = \cO(r)\) and some column sparsity $\smash{\colsparse = \cO\big(\big(1+\tfrac{\log d}{r}\big) \cdot \log r\big)}$.
    With these parameters,
    when the input matrix $\mA \in \F^{n \times d}$, the sketch \(\mA\mOmega\) has arithmetic cost \(\cO(nd \log r)\).
\end{theorem}

We justify the bound on the sketching cost
in \cref{sec:sparsertt-cost};
the proof of the OSI guarantee appears
in~\cref{sec:sparsertt-proof}.
An example of Tropp \cite[Sec.~3.3]{tropp11SRHT}
implies that a column sparsity $\colsparse = \order(\log r)$
is \emph{optimal} for the SparseRTT construction
at embedding dimension $k = \order(r)$.
For trigonometric transform embeddings
\newtext{of the form $\mOmega = \mD\mF\mS$}
with embedding dimension $k = \cO(r)$,
the best available results~\cite{chenakkod24,chenakkod25}
establish the OSE property
with%
a weaker column sparsity
guarantee $\colsparse = \order(\log^3 r)$
using a slightly different sparse matrix $\mS$.

\subsection{Arithmetic costs} \label{sec:sparsertt-cost}

This section justifies the sketching cost for the SparseRTT,
stated in \cref{thm:sparsertt}.
Consider a SparseRTT matrix $\mOmega \in \F^{d \times k}$
with embedding dimension $k \leq d$.
It suffices to work out the cost of forming
$\vy = \mOmega^\top \vx$ for a single vector $\vx \in \F^d$.
A direct implementation using a standard fast trigonometric transform
requires $\cO(d \log d)$ operations.  Here is an accounting of the cost:
\begin{enumerate} \setlength{\itemsep}{0pt}
\item   Form $\vy_1 \gets \mD^\top \vx$. Since $\mD$ is diagonal,
this step expends $\cO(d)$ operations.

\item   Compute $\vy_2 \gets \mF^{\top} \vy_1$ using a fast trigonometric  transform algorithm, at a cost of $\cO(d \log d)$ operations.

\item   Evaluate $\vy \gets \mS^\top \vy_2$.  This step requires $\cO(k \xi)$%
operations.
\end{enumerate}
Assuming the optimal column sparsity $\xi = \cO(\log k)$,
we obtain a final operation count of $\cO( d \log d )$.
Practitioners almost always use the implementation just described.

To optimize the asymptotic cost,
we can employ a fast \emph{subsampled} trigonometric transform.
See~\cite{sorensen93} for the subsampled DFT and DCT;
see~\cite{ailon09a} for the subsampled WHT.
We adjust steps 2 and 3:

\begin{enumerate} \setlength{\itemsep}{0pt}
\item[2'.]   Let $\cR$ be the set of nonzero rows of $\mS$.
Compute $\vy_2' \gets (\mF(:, \cR))^\top \vy_1$ using the
fast subsampled trigonometric transform.  The cost is $\cO(d \log k)$
operations.

\item[3'.]   Evaluate $\vy \gets (\mS(\cR, :))^{\top} \vy_2'$ at a cost
of $\cO(k \xi)$ operations.
\end{enumerate}
Indeed, the set $\cR$ contains at most $k \xi$ coordinates.
Assuming the optimal column sparsity $\xi = \cO(\log k)$,
the total cost is $\cO(d \log k)$%
operations,
consistent with~\cref{thm:sparsertt}.

\subsection{Implementation guidance} \label{sec:sparsertt-implementation}

This section offers some advice on the proper implementation
of SparseRTT test matrices.
First, let us discuss the distribution of the
iid entries of the random diagonal matrix $\mD$.
We remark that real Rademacher variables are admissible,
regardless of the scalar field $\F$.
Unfortunately, real and complex Rademacher random variables
both have failure modes for discrete data; cf.~\cite[Sec.~2.4]{murray23}.
It is often preferable to employ a continuous distribution,
such as $\textsc{uniform}[-\sqrt{3},+\sqrt{3}]$ when $\F = \R$
or the Steinhaus distribution when $\F = \C$.
Theoretical results for SparseRTTs with these distributions
are the same as the results for real Rademachers, modulo constant factors.

For the distribution of entries in the sparse matrix $\mS$,
we can make similar recommendations.%
When \(\bbF=\bbC\), we endorse the use of \(\textsc{rademacher}_\bbC\) or \(\textsc{steinhaus}\) variables in place of real Rademachers.
The complex distributions increase the reliability
and yield equivalent theoretical guarantees.

For the trigonometric transform $\mF$,
we generally recommend the DCT when $\F = \R$ 
and the DFT when $\F = \C$.
Although the WHT is simpler arithmetically than the DCT or DFT,
it is only defined when the dimension $d = 2^{m}$.
In practice, the cost of applying the SparseRTT is
usually dominated by the cost of the standard
fast trigonometric transform.
While subsampled trigonometric transforms are potentially
faster, high-quality implementations of these
algorithms are not available at the time of writing.

To construct the sparse matrix $\mS$, the user must
select an explicit value for the column sparsity $\xi$.
We recommend $\xi = \lceil 1.5 \log k \rceil$
to avoid worst-case examples identified by
Tropp~\cite[Rem.~1.5]{tropp11SRHT}.
When the subspace dimension $r$ is known,
we generally set the embedding dimension $k = 2r$.
The numerical evidence in \cref{fig:SparseRTT}
suggests that more aggressive parameter choices
can be effective.

In the available runtime comparisons~\cite{dong23,Epp25a,chen25},
sparse test matrices are significantly faster than test matrices
based on trigonometric transforms.
On the other hand, to achieve competitive performance,
sparse test matrices require a high-quality
sparse linear algebra library, and it may be necessary to
implement the sketching method in a low-level programming
language, such as \texttt{C}.
Absent a sparse linear algebra library, the SparseRTT
construction can be the better alternative.
Although the SparseRTT involves a sparse
matrix multiplication, the matrix $\mS$ is so sparse that
special libraries are unnecessary.

\subsection{Alternative constructions} \label{sec:sparsertt-alternatives}

The literature describes many trigonometric transform test
matrices that fit the template~\eqref{eqn:intro-srtt}.
These constructions make distinct choices for the sparse matrix $\mS$.
In this section, we summarize these alternatives, their
theoretical guarantees (\cref{table:sparsertt-comparison-body}),
and their empirical performance (\cref{fig:SparseRTT}).

\begin{table}[t]
\centering \small
\begin{tabular}{@{}llll@{}} \toprule
	Sparse matrix \mS & Embedding dim.\ $k$
              & Sparsity $\nnz(\mS)$ & Sketch-and-solve time \\ \midrule
   Subsampling \cite{ailon09} & \(\cO(r \log r)\)
        & \(\cO(r \log r)\) & $\order(nd \log(d) + d^3\log d)$
		\\[0.25em]
    CountSketch \cite{cartis21} & $\order(r)$ & $d$ & $\order(nd\log(n) + d^3)$\\[0.25em]
    LESS \cite{chenakkod24,chenakkod25} & $\order(r)$ & $\order(r\log^3 r)$ & $\order(nd\log^3(d) + d^3)$ \\ \midrule 
    SparseCol (Def.~\ref{def:sparsecol}) & $\order(r)$ & $\order(r\log r)$ & $\order(nd\log(d) + d^3)$ \\
		\bottomrule
\end{tabular}
\caption{\textbf{Fast trigonometric transform test matrices: Comparison.}
Choices for the sparse matrix $\mS$ used within the trigonometric transform test matrix $\mOmega \in \F^{d \times k}$, described in~\eqref{eqn:intro-srtt}.%
We report the embedding dimension $k$ and sparsity $\nnz(\mS)$ that suffice to justify the $(r, \alpha)$-OSI property.
The last column lists the runtimes for the sketch-and-solve method (\cref{alg:sketch-and-solve}) for an input matrix $\mA \in \F^{n \times d}$ and response matrix $\mB \in \F^{n \times d}$.
The SparseCol test matrix (\cref{def:sparsecol}), used in our SparseRTT construction (\cref{def:intro-sparsertt}), simultaneously achieves the minimum embedding dimension $k$ and sparsity $\nnz(\mS)$, resulting in improved runtimes for sketch-and-solve.
}
\label{table:sparsertt-comparison-body}
\end{table}

\begin{figure}[t!]
    \centering
\resizebox{.95\linewidth}{!}{%
\input{figures/results/fig6.pgf}
}
\caption{\textbf{Fast trigonometric transform test matrices: Comparison.} Relative error
\(\norm{\mA - \hat\mA}_{\rm F}^2 / \norm{\mA}_{\rm F}^2\)
for the rank-$k$ approximation $\hat\mA$ obtained by the RSVD (\cref{alg:rsvd}).
Sketching is performed with a fast transform test matrix \eqref{eqn:intro-srtt}
with Rademacher $\mD$ and DCT $\mF$;
the sampling matrix $\mS$ is SparseIID (\emph{left}, \cref{def:sparse-iid}) or SparseCol (\emph{right}, \cref{def:sparsecol}).
All plots compare against a Gaussian test matrix (\cref{def:gauss-test})
and a fast transform test matrix~\eqref{eqn:intro-srtt}
where $\mS$ is subsampling or CountSketch (\cref{sec:sparsertt-alternatives}).
We consider four diagonal input matrices $\mA \in \R^{1024 \times 1024}$ from the testbed in \cite{tropp2017a,tropp17b}.  For each instance, we perform $100$ trials. The data markers track the median errors; shaded regions are bounded by the 10\% and 90\% quantiles.
On these examples, the SparseCol matrix achieves better accuracy than SparseIID matrix for every value of the column sparsity $\colsparse$.
}
    \label{fig:SparseRTT}%
\end{figure}

Here are several options for the sparse matrix $\mS$
in a fast trigonometric transform~\eqref{eqn:intro-srtt}.
The tag describes the style of the line and data
marker in \cref{fig:SparseRTT}.

\begin{enumerate} \setlength{\itemsep}{0pt}
\item   \textbf{Subsampling (\emph{solid line, open square}).}
    The classic subsampled trigonometric transform test matrix~\cite{ailon09,woolfe08} arises when $\mS$
    is a (scaled) restriction matrix.
    This matrix has minimal sparsity $\nnz(\mS) = k$,
    but it requires a larger embedding dimension
    $k = \Omega(r \log r)$ to serve as an OSI.

\item   \textbf{CountSketch (\emph{solid line, solid circle}).}
    Cartis \etal \cite{cartis21} choose $\mS$
    to be a CountSketch embedding.
    This test matrix is reliable,
    but it has higher sparsity ($\nnz(\mS) = d$)
    than the SparseCol.
    It is also more expensive, requiring $\cO(nd \log d)$
    operations to form the sketch $\mA \mOmega$
    of the matrix $\mA \in \F^{n \times d}$.

\item   \textbf{SparseCol (\emph{dashed line, solid circle}).}
    The SparseRTT construction (\cref{def:intro-sparsertt})
    relies on the SparseCol matrix (\cref{def:sparsecol}).

\item   \textbf{SparseIID (\emph{solid line, solid square}).}
    If we replace SparseCol with the SparseIID map
    defined in \cref{sec:stack-vs-iid},
    our main result for the SparseRTT (\cref{thm:sparsertt})
    holds with the same embedding dimension $k$ and
    (expected) column sparsity $\xi$.
    Empirically, the SparseIID map requires higher
    column sparsity $\colsparse$ than SparseCol
    to perform reliably.
\end{enumerate}

\Cref{table:sparsertt-comparison-body} summarizes
these constructions and the influence
on the arithmetic cost of the 
sketch-and-solve method (\cref{alg:sketch-and-solve}).
We see that the SparseRTT test matrix simultaneously
allows the minimum embedding dimension $k = \cO(r)$,
while the SparseCol matrix $\mS$ attains the minimum
possible sparsity $\nnz(\mS) = \cO(r \log r)$.
In addition, among the trigonometric transform sketches,
the SparseRTT offers the fastest runtime for
sketch-and-solve.

\Cref{fig:SparseRTT} provides an empirical comparison
of the RSVD (\cref{alg:rsvd}) implemented with trigonometric transform
sketches \newtext{ $\mOmega = \mD\mF\mS$ }
as we vary the sparse matrix $\mS$.
The left panels show the performance of when $\mS$
is a SparseIID matrix, while the right panels
show the recommended SparseCol matrix.
For reference, all panels display the results
for the existing constructions
(subsampling $\mS$, CountSketch $\mS$, Gaussian $\mOmega$).
For these examples, the recommended SparseRTT test matrix
with column sparsity $\colsparse = 4$%
achieves
accuracy comparable with a Gaussian test matrix.
The SparseIID sampling method is less effective than
the SparseCol sampling method at every value of $\colsparse$.

\subsection{Proof of \cref{thm:sparsertt}} \label{sec:sparsertt-proof}

The proof of \cref{thm:sparsertt} parallels the analysis
of the SRTT construction with the WHT matrix.
We rely on the following incoherence result,
adapted from~\cite[Lem.~3.3]{tropp11SRHT}.

\begin{importedtheorem}[Randomized trigonometric transforms create incoherence, \protect{\cite[Lem.~3.3]{tropp11SRHT}}]
    \label{impthm:srft-coherence}
    \newtext{Let $\mQ \in \F^{d\times r}$ be an orthonormal matrix.}
    Select $\mD \in \F^{d \times d}$ and $\mF \in \F^{d\times d}$
    as in \cref{def:intro-sparsertt}.
    With probability at least \(1-\delta\),
    the coherence of the orthonormal matrix
    \(\mQ_{\rm F} \defeq (\mD\mF)^\top \mQ \in \R^{d\times k}\) satisfies
    \begin{align}
        \mu(\mQ_{\rm F})
        = \cO\left(\frac{r + \log(d/\delta)}{d}\right).
        \label{eq:srht-coherence}
    \end{align}
\end{importedtheorem}

\noindent
\Cref{impthm:srft-coherence} is based on
the result~\cite[Lem.~3.3]{tropp11SRHT},
which is stated under the assumptions
that $\mD \in \F^{d\times d}$ consists
of real Rademacher variables and
that $\mF \in \F^{d \times d}$ is a WHT.
Both conditions can be relaxed without
any additional insight.
To handle complex Rademacher variables,
we split them into real and imaginary parts.
The DCT and the DFT (like the WHT)
are unitary matrices whose
entries are bounded in magnitude
by $\cO(d^{-1/2})$.
The proof only depends on these
two features of the WHT.

\Cref{impthm:srft-coherence} states that%
the (unsampled) random trigonometric transform $(\mD \mF)^\top$
converts an arbitrary orthonormal
matrix into an incoherent orthonormal matrix.%
Thus, to prove that a SparseRTT is an OSI,
it suffices to confirm that a SparseCol matrix
$\mS$ acts as an injection on an incoherent subspace.%
This is the content of the next result.

\begin{theorem}[SparseCol: Injectivity for incoherent matrices]
    \label{thm:sparse-indep-cols-osi}
    \newtext{
        Let $\mQ\in\F^{d \times r}$ be an orthonormal matrix with coherence \(\mu(\mQ)\).
    }
    For some choice of embedding dimension $k = \cO(r)$
    and column sparsity $\colsparse = \cO(\frac dr \mu(\mQ) \log r)$,
    a random SparseCol matrix $\mS\in\F^{d \times k}$
    satisfies
    \[
        \sigma_{\rm min}^2(\mS^{\top} \mQ) \geq \nicefrac12
        \quad\text{with probability at least~\newtext{\nicefrac{19}{20}}}
    \]
\end{theorem}

\Cref{sec:sparsecol-pf} contains a proof of 
\cref{thm:sparse-indep-cols-osi} based on
the Gaussian comparison theorem (\cref{sec:gaussian-compare}).
With this preparation, we can establish \cref{thm:sparsertt}.

\begin{proof}[Proof of \cref{thm:sparsertt}.]
    Introduce the orthonormal matrix
    $\mQ_{\rm F} \defeq (\mD\mF)^\top \mQ \in \F^{d \times k}$.
    \newtext{Taking $\delta=\nicefrac{1}{40}$}, \Cref{impthm:srft-coherence} states that this matrix has coherence
    $\mu(\mQ_{\rm F}) = \cO((r + \log d)/d)$
    with probability at least \(\nicefrac{39}{40}\).
    We construct a SparseCol matrix
    $\mS \in \F^{d \times k}$ with column sparsity
    $
    \colsparse
        = \cO\big( \big(1 + \frac{\log d}{r} \big) \cdot \log r \big)
        = \cO\big( \frac{d}{r} \mu(\mQ_{\rm F}) \cdot \log r \big).
    $
    According to \cref{thm:sparse-indep-cols-osi},
    this choice of $\xi$ is sufficient to ensure that
    the SparseRTT acts injectively on $\range(\mQ)$: 
    \[
    \sigma_{\min}^2(\mOmega^{\top} \mQ)
        = \sigma_{\min}^2((\mD \mF \mS)^\top \mQ)
        = \sigma_{\min}^2(\mS^{\top} \mQ_{\rm F})
        \geq \nicefrac12
    \]
    with probability at least \(\nicefrac{39}{40}\).
    Take a union bound over the two events to complete the argument.
\end{proof}

\section{Khatri--Rao test matrices and their applications} \label{sec:khatri-rao-top-level-section}

In recent years, we have witnessed an explosion of research
on linear algebra problems with \emph{tensor structure}
\cite{camano25,aldaas23,bharadwaj23,ahle20,bujanovic25,jin21,bamberger22,jin24,meyer23,meyer25,pham13,feldman22,ma2022,rakhshan20,rakhshan22}.
For these problems, the simplest 
sketching techniques are grossly inefficient or entirely inapplicable.  Instead, we must resort to special test matrices that are compatible with the tensor structure.

\newtext{This section discusses the Khatri--Rao test matrix, defined and introduced in \cref{sec:khatri-rao-intro}, as a sketching matrix for tensor-structured data.}
\Cref{sec:khatri-rao-applications} \newtext{describes} some settings where tensor-structured sketching is essential.
\newtext{Next,} in~\cref{sec:khatri-rao-overwhelming-orthogonality}, \newtext{we discuss how the performance of Khatri--Rao test matrices can degrade exponentially as the tensor order increases.}
\newtext{This paves the way for \cref{sec:khatri-rao-large-sketch,sec:khatri-rao-small-sketch}, which present}
two theorems that quantify the exponential loss in
complementary ways.
Throughout, we will argue in favor of constructing Khatri--Rao test matrices with a \emph{spherical} base distribution.
We conclude with an application of this theory
to a matrix recovery problem in~\cref{sec:matrix-recovery}.

\subsection{\newtext{Uses for Khatri--Rao test matrices}} \label{sec:khatri-rao-applications}

Khatri--Rao test matrices fill a different niche in the computational ecosystem, as compared with sparse test matrices and trigonometric transform test matrices.
Indeed, for certain applications,%
our best option is to use Khatri--Rao sketching.
Here are several examples.

\begin{enumerate}
\item   \textbf{Sylvester equations.}  For fixed matrices $\mB, \mC, \mY$, the linear equation $\mB \mX + \mX \mC = \mY$ in the matrix variable $\mX$ is called a \emph{Sylvester equation}.  Equivalently,
\[
\Vec(\mX) = \mA \cdot \Vec(\mY)
\quad\text{where}\quad
\mA \coloneqq (\mI \otimes \mB + \mC^{\mathsf{T}} \otimes \mI)^{-1}.
\]
The solution operator $\mA$ is often numerically low rank,
so we can approximate it by means of the generalized
Nystr{\"o}m method (\cref{alg:gen_nystrom_outer}).
To sketch the solution operator, we can solve the Sylvester
equation with (random) rank-one matrices $\mY$,
for example, via the fADI algorithm~\cite{benner09}.
The vectorization of a rank-one matrix has
tensor-product structure (with order $\ell = 2$),
so this is an instance \newtext{of the Kronecker matvec access model introduced in \cref{sec:khatri-rao-intro}}.
For other applications of Khatri--Rao sketching
to matrix equations, see~\cite{bujanovic25}.

\item   \textbf{Tensor networks.}  When simulating entangled quantum systems, scientists often employ tensor networks to represent exponentially large matrices~\cite{orus19}.  The tensor network format is an instance of the Kronecker matvec access model, so Khatri--Rao sketching offers an efficient method for compressing tensor networks.
For examples,
see the recent papers~\cite{camano25,feldman22}.

\item   \textbf{CP decompositions.}  The canonical polyadic decomposition (CPD) is a powerful tool for analyzing tensor-structured data~\cite[Sec.~3.5]{kolda09}.
Algorithms for computing CPDs repeatedly solve overdetermined least-squares problems with tensor structure~\cite[Sec.~3.4]{kolda09}.
Existing randomized algorithms depend on leverage-score sampling~\cite{larsen22,bharadwaj23}.
Khatri--Rao sketching offers a natural alternative.
\end{enumerate}

To complement these examples,  \cref{sec:matrix-recovery} describes an elementary matrix recovery problem that involves Khatri--Rao test matrices.
\Cref{sec:science-trace} highlights a scientific
application of Khatri--Rao test matrices for estimating the
partition function of a quantum-mechanical system.
Both applications benefit from a deeper theoretical understanding of tensor-structured sketches.

\subsection{Overwhelming orthogonality: When Khatri--Rao goes bad} \label{sec:khatri-rao-overwhelming-orthogonality}

\Cref{fig:rsvd-gaussian-vs-structured}, in the introduction, provides empirical evidence that Khatri--Rao test matrices serve almost as well as Gaussian test matrices for computing low-rank approximations
of sparse matrices that arise in applications.
Nevertheless, tensor-structured test matrices exhibit behaviors that
are fundamentally different from the constructions we studied before.
To explain why, let us describe a phenomenon called \emph{overwhelming orthogonality}~\cite{meyer25}.  Here is the most vivid example.

Consider the $2$-dimensional Rademacher base distribution
$\vnu \sim \textsc{uniform}\{(\pm 1, \pm 1)\} \in \R^2$.
Note that
\[
\prob\big\{ | \langle \vnu, \bm{1} \rangle | = 0 \big\}
    = \tfrac{1}{2}
    \quad\text{where $\bm{1} \coloneq (1,1)^\top \in \R^2$.}
\]
Form the tensor product $\vomega = \vomega^{(1)} \otimes \cdots \otimes \vomega^{(\ell)}$ where the factors $\vomega^{(i)} \sim \vnu$ iid.
For the vector of ones $\bm{1} \in \R^{2^{\ell}}$,
\[
\langle \vomega, \bm{1} \rangle
    = \big\langle \vomega^{(1)} \otimes \cdots \otimes \vomega^{(\ell)},
    \bm{1} \otimes \cdots \otimes \bm{1} \big\rangle
    = \prod_{i=1}^{\ell} \big\langle \vomega^{(i)}, \bm{1} \big\rangle.
\]
Note that the inner product equals zero when
any one of the
random vectors $\vomega^{(i)}$ is orthogonal to $\bm{1}$.
It follows that
the orthogonality probability is almost one:
\[
\prob\big\{ | \langle \vomega, \bm{1} \rangle | = 0 \big\}
    = 1 - 2^{-\ell}.
\]
Thus, to certify that a vector in $\R^{2^{\ell}}$ is nonzero,
it requires $\Omega(2^{\ell})$ tensor Rademacher measurements!
The same phenomenon persists for any base distribution $\vnu \in \F^{d_0}$ that is sometimes orthogonal to a fixed vector.

\newtext{Even when the base distribution is orthogonal to each fixed direction with probability zero, tensor products $\vomega = \vomega^{(1)} \otimes \cdots \otimes \vomega^{(\ell)}$ of independent isotropic vectors are still exponentially likely to be nearly orthogonal to certain vectors.}
For example, suppose that the base distribution $\vnu \in \R^{d_0}$ is the real spherical distribution and that $\vq = \va \otimes \cdots \otimes \va$ for
any fixed unit vector $\va \in \R^{d_0}$.
Meyer \etal~\cite[Lem.~8]{meyer25} proved that \(| \langle \vomega, \vq \rangle| \leq \tau^\ell\) with probability at least \(1- \tau^{\ell}\) for some \( \tau \in (0,1)\).
As a consequence, Khatri--Rao embeddings have trouble
picking up information about vectors with Kronecker structure.
For these worst-case instances, we must anticipate an
exponential dependence on the tensor order $\ell$ in
either the embedding dimension $k$ or in the
injectivity parameter $\alpha$.

\begin{figure}[t!]
    \centering
\resizebox{.9\linewidth}{!}{%
\input{figures/results/fig7.pgf}
}
\caption{\textbf{Khatri--Rao test matrices: Comparison of base distributions.}%
 (\emph{left}) Injectivity \(\alpha \coloneqq \sigma_{\rm min}^2(\mOmega^\top \mQ)\) and (\emph{right}) RSVD relative error \(\norm{\mA - \hat\mA}_{\rm F} / \norm{\mA}_{\rm F}\)
for Khatri--Rao test matrices (\cref{def:kr-intro}) with several base distributions
where the base dimension $d_0 = 2$ and the tensor order $\ell = 10$; see \cref{sec:prelims} for definitions.
        The orthonormal matrix $\mQ \in \R^{2^{10} \times 50}$ consists of (\emph{top row}) the first 50 columns of the WHT matrix or (\emph{bottom row}) 50 orthonormalized Kronecker--Gaussian vectors.
        The RSVD algorithm is applied to the matrix
        $\mA = \mQ\mQ^\top + 10^{-4}\mI$.
        Markers track the median of 100 trials;
        the shaded region is bounded by the 10\% and 90\% quantiles.
        The Gaussian baseline comes from a full Gaussian test matrix (\cref{def:gauss-test}).
        In the top-left panel, the real Rademacher distribution is not visible because the median value of $\sigma_{\rm min}(\mOmega^\top \mQ) = 0$ for $k \le 10^3$.
    }%
    \label{fig:kr-hadamard}
\end{figure}

\subsubsection{Numerical demonstration} \label{sec:kr-rad-bad-numerics}

\Cref{fig:kr-hadamard} illustrates the problem of overwhelming orthogonality and compares several base distributions.
In this experiment, we fix an orthonormal matrix $\mQ \in \F^{d_0^{\ell} \times r}$ whose columns have Kronecker product structure.
The base dimension $d_0 = 2$, the tensor order $\ell = 10$,
and the subspace dimension $r = 50$.
We construct a Khatri--Rao embedding $\mOmega \in \F^{d_0^{\ell} \times k}$
with each of the seven base distributions $\vnu$.
The left panel charts the injectivity $\sigma_{\min}^2(\mOmega^\top \mQ)$
as a function of the embedding dimension $k$;
the right panel shows the relative error in the RSVD approximation
of the matrix $\mA = \mQ\mQ^\top + 10^{-4}\mI$.

Since the columns of the input matrix $\mQ$ have Kronecker structure,
overwhelming orthogonality manifests fiercely.
For all base distributions, observe that embedding dimension $k \gg r$ is required to achieve injectivity similar to that of a Gaussian test matrix.
Nevertheless, we can achieve competitive RSVD performance
using a Khatri--Rao embedding with the complex spherical base distribution.
We also observe that the injectivity is very sensitive to the base distribution\rev{, as documented in \cref{table:KR-constants}}.
In particular, the complex Rademacher distribution
fails to achieve injectivity $\alpha > 0$ until the embedding dimension $k \geq 200$, and the real Rademacher distribution fails catastrophically even when $k = 1000$.
Notably, the Steinhaus base distribution,
another complex generalization of the Rademacher,
avoids this failure mode entirely.

\subsection{First result: Khatri--Rao embeddings with constant injectivity} \label{sec:khatri-rao-large-sketch}

In this section, we take the first step toward understanding
the OSI properties of a Khatri--Rao test matrix.
We obtain a sufficient condition on the embedding dimension
$k$ to obtain a \emph{constant} bound for the
injectivity, say $\alpha \geq \tfrac{1}{2}$.
The result reflects the choice of base distribution,
and---because of overwhelming orthogonality---%
the embedding dimension depends exponentially
on the tensor order $\ell$.

\begin{theorem}[Khatri--Rao: Constant injectivity]
    \label{corol:kr-large-sketch-osi}
    Consider an isotropic base distribution $\vnu \in \F^{d_0}$
    that satisfies the fourth-moment bound
    \begin{equation} \label{eqn:kr-4th}
    \E \big[ |\langle \vnu, \va \rangle|^4 \big]
        \leq \rC_{\vnu}
    \quad\text{for each unit-norm vector $\va \in \F^{d_0}$.}
    \end{equation}
    For tensor order $\ell$, the Khatri--Rao test matrix $\mOmega \in \F^{d_0^\ell \times k}$ enjoys the $(r,\nicefrac12)$-OSI property
    for some embedding dimension $k = \cO({\rC}_{\vnu}^{\ell} r )$.
\end{theorem}
\noindent 
The proof of~\cref{corol:kr-large-sketch-osi}
appears in~\cref{sec:fourth-to-second}.
The argument depends on the fourth-moment theorem
of Oliveira~\cite[Thm.~1.1]{oliveira16}
and variance bounds for Kronecker quadratic forms due to
Meyer \& Avron \cite{meyer23}.
For the embedding dimension columns, we omit \(\cO(\cdot)\) symbols for visual clarity.

\begin{table}[t]
\centering
\begin{tabular}{@{}cllrr@{}} \toprule
    & & & \multicolumn{2}{c}{Embedding dimension \(k = \cO(\cdots)\)}  \\ \cmidrule(r){4-5}
	Field ($\F$) & Base distribution ($\vnu$) & $\rC_{\vnu}$ & General \(d_0\) & \(d_0=2\) \\ \midrule
	\bbR
        &\(\cN(\vec0,\mI)\)
        & \(3\)
		& \(3^\ell r\)
		& \(3^\ell r\)
		\\[0.25em]
	&Real Rademacher
        & \(3-\frac2{d_0}\)
		& \hspace{2pc}\((3-\frac2{d_0})^\ell r\)
		& \(2^\ell r\)
		\\[0.25em]
	&Real spherical
		& \(3-\frac{6}{d_0+2}\)
		& \((3-\frac{6}{d_0+2})^\ell r\)
		& \(1.5^\ell r\)
		\\[0.25em]
	\bbC
        &\(\cN_{\C}(\vec0,\mI)\)
		& \(2\)
		& \(2^\ell r\)
		& \(2^\ell r\)
		\\[0.25em]
	&Complex Rademacher
		& \(2-\frac1{d_0}\)
		& \((2-\frac1{d_0})^\ell r\)
		& \(1.5^\ell r\)
		\\[0.25em]
	&Steinhaus
		& \(2-\frac1{d_0}\)
		& \((2-\frac1{d_0})^\ell r\)
		& \(1.5^\ell r\)
		\\[0.25em]
	&Complex spherical
		& \(2-\frac{2}{d_0+1}\)
		& \((2-\frac{2}{d_0+1})^\ell r\)
		& \(\smash{1.\overline{33}}^\ell r\)
		\\[0.25em]
		\bottomrule
\end{tabular}
\caption{
    \textbf{Base distributions: Fourth-moment bounds,
    \protect{\cite[Tab.~1]{meyer23}}.}
    The exact fourth-moment constant \(\rC_{\vnu}\)%
    for several base distributions $\vnu \in \F^{d_0}$, including standard normal distributions, Rademacher distributions, the Steinhaus distribution, and the uniform distribution on the sphere of radius $\sqrt{d_0}$.  The last two columns list the embedding dimension $k$ that suffices to obtain an $(r, \nicefrac{1}{2})$-OSI, where the $\cO$ is omitted for legibility.
    See~\cref{corol:kr-large-sketch-osi} for the definition of $\rC_{\vnu}$ \newtext{and \cref{sec:prelims} for} the definitions
    of the distributions.%
}
\label{table:KR-constants}
\end{table}

Meyer \& Avron~\cite{meyer23} have calculated
the exact value of the fourth-moment constant
$\rC_{\vnu}$ for each of several base distributions;
\cref{table:KR-constants} summarizes their results.
The relative ordering of the fourth-moment constants
is wholly consistent with the numerical
experiments on injectivity in~\cref{fig:kr-hadamard}.

In contrast, the RSVD error
induces a different ordering of the base distributions
in \cref{fig:kr-hadamard}.  The reason for the discrepancy
is that RSVD (\cref{alg:rsvd})
and psd Nystr{\"o}m (\cref{alg:nystrom}) only
depend on the range of the test matrix $\mOmega$,
so Gaussian and spherical base distributions result
in the same performance in exact arithmetic.
On the other hand, sketch-and-solve (\cref{alg:sketch-and-solve}) and 
generalized Nystr{\"o}m (\cref{alg:gen_nystrom_outer}) are sensitive
to the scaling of columns of the test matrix $\mOmega$,
so spherical distributions are the best choice
for general-purpose use.

In quantum science, we frequently
work in the regime where%
the base dimension $d_0 = 2$ and the tensor order
$\ell$ is very large~\cite{orus19,feldman22}.
In this case, real standard normal vectors have the
worst predicted embedding dimension bound $k = \cO(3^{\ell} r)$,
and the complex spherical distribution has the
best embedding dimension bound $k = \cO(\smash{1.\overline{33}{}^\ell r})$.
It is \newtext{notable} that changes in the scalar field
and the base distribution can lead to exponential
changes in the embedding dimension of a Khatri--Rao test matrix
(and the downstream cost of linear algebra algorithms).

\subsection{Second result: Khatri--Rao embeddings with proportional embedding dimension}
\label{sec:khatri-rao-small-sketch}

\Cref{corol:kr-large-sketch-osi} requires the embedding
dimension $k$ of a Khatri--Rao test matrix
to scale exponentially with the tensor order $\ell$.
This cost is acceptable for small $\ell$,
but we can rarely afford the exponential
scaling when $\ell$ is large.
Instead, we would prefer to employ a Khatri--Rao test matrix
whose embedding dimension is proportional to the subspace
dimension: $k = \cO(r)$, without dependence on $\ell$.

To achieve this goal, we must avoid base distributions
that have \newtext{zero probability of exact orthogonality}, \newtext{such as the real or complex Rademacher distributions}.
Therefore, we restrict our attention to continuous
base distributions that have an anticoncentration property.
The resulting Khatri--Rao embedding achieves the
linear scaling of the embedding dimension with
the subspace dimension.
Nevertheless, as predicted by overwhelming orthogonality,
we must accept an exponentially small injectivity parameter.

\begin{theorem}[Khatri--Rao: Proportional embedding dimension]
    \label{corol:kr-small-sketch}
    \newtext{Let \(\vnu\in\bbR^{d_0}\) be a \emph{real} isotropic vector}
    whose entries are independent, continuous random variables,
    with densities uniformly bounded by $M$.
    The associated Khatri--Rao embedding
    $\mOmega \in \R^{d_0^\ell \times k}$
    serves as an $(r, \alpha)$-OSI
    with injectivity $\alpha = (\rC M)^{-2\ell}$
    at an embedding dimension $k = \cO(r)$.
    The same result holds with $M = 1$ when the base distribution \(\vnu\) is isotropic and log-concave,
    or when $\vnu$ is the real
    spherical distribution.
\end{theorem}

\noindent 
We prove \cref{corol:kr-small-sketch} using a version of the small-ball method~\cite[Exer.~8.27]{Ver25:High-Dimensional-Probability-2ed}.%
The argument requires small-ball probability bounds
for tensor products due to Hu \& Paouris \cite{hu2024small}.
See~\cref{sec:small-ball} for details.

In particular, a Khatri--Rao embedding with a real standard normal
base distribution is an $(r,\alpha)$-OSI with injectivity
$\alpha = \rC^{-\ell}$ at an embedding dimension $k = \cO(r)$.
Recall that the randomized linear algebra algorithms (\cref{sec:randnla-via-osi}) produce errors that are within a $\mathrm{poly}(\alpha^{-1})$ factor of optimal.  As a consequence, Gaussian Khatri--Rao embeddings with $k = \mathcal{O}(r)$ can yield errors that are suboptimal by an exponential factor $\rC^{\ell}$.
This debility is visible in \cref{fig:kr-hadamard} (left) when
the embedding dimension $k \approx r = 50$.

Observe that discrete base distributions,
such as Rademacher distributions, do not
meet the hypotheses of \cref{corol:kr-small-sketch}.
This limitation reflects a genuine failure mode
of the Khatri--Rao embedding, arising from
the overwhelming orthogonality phenomenon
outlined in \cref{sec:khatri-rao-overwhelming-orthogonality}.
For a discrete base distribution supported on a finite
number of points,%
the associated Khatri--Rao test matrix%
can \emph{never} achieve a proportional embedding dimension
$k = \cO(r)$ that is independent of the tensor order $\ell$.
See~\cite[Thm.~19]{meyer25} for confirmation of this point.

In practice, we can often avoid the exponential scaling
of the Khatri--Rao embedding dimension
or the injectivity parameter.%
First, some applications of Khatri--Rao sketching
involve small tensor order $\ell$.
For example, we have $\ell = 2$
for Sylvester equations (\cref{sec:khatri-rao-applications})
and for the matrix recovery problem discussed
below (\cref{sec:matrix-recovery}).
Second, even when $\ell$ is large,
Khatri--Rao embeddings only exhibit
exponential scaling for worst-case subspaces.
Indeed, \cref{fig:rsvd-gaussian-vs-structured}
already provides clear evidence that the RSVD
with a Khatri--Rao embedding can reliably
approximate matrices arising from applications.

\subsection{\newtext{Implication}: Approximate matrix recovery from bilinear queries} \label{sec:matrix-recovery}

In recent years, \emph{matrix recovery problems} have captured the imagination of many researchers in computational mathematics~\cite{amsel24,halikias24,chen25a,boulle23,boulle24,amsel25}.
Matrix recovery problems ask us to find a near-optimal approximation of an
input matrix $\mB \in \F^{n \times n}$ from within a class $\cF \subseteq \F^{n \times n}$ of structured matrices:
\begin{equation} \label{eqn:mtx-recovery}
\text{Find $\oldtilde\mB \in \F^{n \times n}$ : }
\quad
\norm{\mB - \tilde\mB}_{\rm F}
    \le \Gamma \cdot \min_{\mM \in \cF}\ \norm{\mB - \mM}_{\rm F}.
\end{equation}
The parameter $\Gamma \geq 1$ is a factor of suboptimality.
Examples of structured classes $\cF$ include patterned matrices
(e.g., circulant, Toeplitz, Hankel)
or sparse matrices with a fixed bandwidth. 

In many settings, we can only access the target matrix
$\mB \in \F^{n \times n}$ via a restricted class of queries.
For example, in the \emph{matvec access model}, the algorithm can interact with the matrix by extracting the product $\vx \mapsto \mB \vx$ with a test vector $\vx \in \F^n$.
The goal is to produce the approximation $\oldtilde{\mB}$ after performing the minimum number of queries.  This challenge arises in operator learning problems~\cite{boulle23,boulle24}.

In this paper, we consider the problem of approximating a matrix from a \emph{linearly parameterized matrix family}~\cite[Def.~2.1]{halikias24}.
More precisely, the class $\cF$ is a $d$-dimensional linear subspace spanned by a linearly independent family of matrices $\mM_1, \dots, \mM_d \in \F^{n \times n}$:
\begin{equation} \label{eqn:lin-mtx-fam}
    \cF \coloneqq \operatorname{span} \{ \mM_1,\ldots,\mM_d \}
    \subseteq \F^{n \times n}.
\end{equation}
We study the \emph{bilinear access model}, where we interact with the
input matrix $\mB \in \F^{n \times n}$ by evaluating bilinear%
forms $(\vx,\vy) \mapsto \vy^{\sf T} \mB \vx$ for input vectors $\vx, \vy \in \F^{n}$ of our choice.
Let us emphasize that the bilinear form involves the plain transpose,
\emph{without} a conjugate; this choice reflects the fact that
tensor products are \emph{multilinear}.
This access model provides an abstraction for resource-constrained
operator learning problems where we can only test the operator
by applying it to elements of the computational basis,
measuring the output in another computational basis.

A simple dimension-counting argument shows that $d$ bilinear queries are necessary and sufficient to perfectly reconstruct a target matrix $\mB \in \cF$ that belongs to the class.
For an \emph{arbitrary} target matrix $\mB \in \F^{n \times n}$,
we can prove that $\cO(d)$ bilinear queries suffice
to find a near-optimal approximation $\oldtilde{\mB} \in \cF$
from the structured class.
This result is a consequence of our OSI guarantees for Khatri--Rao embeddings.

\begin{theorem}[Approximate matrix recovery from bilinear queries] \label{thm:matrix-recovery}
    Let $\cF$ be a \(d\)-dimensional linearly parametrized matrix family,
    as in~\eqref{eqn:lin-mtx-fam}, and let \(\mB\in\bbF^{n \times n}\) be any matrix.
    With $p = \cO(d)$ bilinear queries, \cref{alg:matrix-recovery}
    returns a matrix \(\tilde\mB \in \F^{n \times n}\) that satisfies
    \begin{equation*}
    \norm{\mB - \tilde\mB}_{\rm F}
        \le \rC \cdot
        \min_{\mM \in \cF}\ \norm{\mB - \mM}_{\rm F}
    \quad\text{with probability at least $\tfrac{9}{10}$.}
    \end{equation*}
\end{theorem}

\begin{algorithm}[t]
    \caption{Approximate matrix recovery from bilinear queries}
    \label{alg:matrix-recovery}
    \begin{algorithmic}[1]
        \Require Target matrix $\mB \in \bbF^{n \times n}$ accessible by bilinear queries, matrices $\mM_1,\ldots,\mM_d$ defining a linearly parametrized matrix family $\cF$, as in~\eqref{eqn:lin-mtx-fam},
        number $p\ge d$ of bilinear queries 
        \Ensure Near-optimal approximation $\tilde\mB \in \cF$ to the target matrix $\mB$
        \State Generate iid standard normal vectors $\vx_1,\ldots,\vx_p,\vy_1,\ldots,\vy_p \sim \cN_{\F}(\bm{0}, \mI)$ in $\F^n$
        \State Form the matrix $\mF\in\bbF^{p\times d}$ with entries $\smash{f_{ij} = \vy_i^{\sf T}\mM_j^{\vphantom{\sf T}}
        \vx_i^{\vphantom{\sf T}}}$
        \State Form the vector $\vg \in \bbF^p$ with entries $\smash{g_i = \vy_i^{\sf T} \mB^{\vphantom{\sf T}} \vx_i^{\vphantom{\sf T}}}$
            \Comment{Bilinear queries!}
        \item $\oldtilde\vx \gets \mF \setminus \vg$ \Comment{Stable least-squares solver}
        \item $\tilde\mB \gets \sum_{j=1}^d \oldtilde x_j \mM_j$
        \Comment Approximation to target matrix
    \end{algorithmic}
\end{algorithm}

To the best of our knowledge, no one has investigated
the approximate matrix recovery problem
in the bilinear access model.
Thus, we believe \cref{thm:matrix-recovery} is the first result
to achieve a nontrivial query complexity for this problem.
Using the best existing OSE result for
Khatri--Rao embeddings~\cite[Thm.~5]{bujanovic25},
we can only attain the weaker query complexity
bound $p = \cO(d^{3/2})$.
The concurrent research~\cite{saibaba2025improved}
yields an improvement,
but it still falls short of the optimal query complexity.

\begin{proof}
    The matrix recovery problem~\eqref{eqn:mtx-recovery} is as an overdetermined linear least-squares problem.  Define%
    \begin{equation*}
        \mA \coloneqq \begin{bmatrix}
    \vertbar & \vertbar &        & \vertbar \\
    \Vec (\mM_1)    & \Vec (\mM_2)    & \cdots & \Vec (\mM_d)    \\
    \vertbar & \vertbar &        & \vertbar 
        \end{bmatrix} \quad\text{and}\quad
        \vec{b} \coloneqq \Vec(\mat{B}).
    \end{equation*}
    Given an input matrix $\mB$, the best Frobenius-norm approximation $\mB_{\star}$ from the subspace $\cF$ admits a representation
    as a linear combination of the basis matrices:
    $\mB_{\star} = \sum_{j=1}^d x_j^\star \mM_j$.
    The vector $\vx^{\star} \coloneqq (x_1^\star, \dots, x_d^{\star})$
    of optimal coefficients is obtained as the (unique) solution to the least-squares problem
    \begin{equation*}
        \min_{\vx \in \bbF^d}\ \norm{\smash{\vb - \mA \vx}}_2
            = \min_{\vx \in \bbF^d}\ \norm{ \mB - \sum_{j=1}^d x_j \mM_j }_{\rm F}
            = \min_{\mM \in \cF}\ \norm{\mB - \mM}_{\rm F}.
    \end{equation*}

    Construct the Khatri--Rao test matrix $\mOmega \in \F^{n^2 \times k}$
    whose columns $(\vx_i \otimes \vy_i : i = 1, \dots, k)$ are tensor products of independent standard normal random vectors.
    According to~\cref{corol:kr-large-sketch-osi}
    and \cref{table:KR-constants},
    (the entrywise complex conjugate of)
    this test matrix
    satisfies the $(d,\nicefrac{1}{2})$-OSI property
    for some $k = \mathcal{O}(d)$.
    
    \Cref{alg:matrix-recovery} attacks the vectorized least-squares problem
    by sketching $\mF = \mOmega^{\sf T} \mA$
    and $\vg = \mOmega^{\sf T} \vb$.
    It returns the approximate solution
    $\oldtilde\vx = (\mOmega^{\sf T} \mA)^\dagger (\mOmega^{\sf T} \vg)$.
    The resulting matrix solution $\oldtilde{\mB} = \sum_{j=1}^d \oldtilde{x}_j \mM_j$ satisfies
    \begin{equation*}
        \norm{\mB - \tilde\mB}_{\rm F}
        = \norm{\smash{\vb - \mA \oldtilde\vx}}_2
        \le \rC\cdot \min_{\vx \in \bbF^d}\ \norm{\smash{ \vb - \mA \vx}}_2
        = \rC \cdot \min_{\mM \in \cF}\ \norm{\mB - \mM}_{\rm F}
        = \rC \cdot \norm{\smash{\mB - \mB_{\star}}}_{\rm F}.
    \end{equation*}
The inequality follows from our analysis of sketch-and-solve with an OSI (\cref{thm:lsq-via-osi}).
\end{proof}

\section{Scientific applications}
\label{sec:science-applications}

We present two scientific applications of the structured test matrices studied in this paper.
In \cref{sec:science-pod-modes}, we employ the SparseStack test matrix for compression and analysis of scientific simulation data.
This approach produces an accurate approximation of the input data
faster than baseline algorithms that use Gaussian test matrices.
We apply this methodology to data from a simulation of Bose--Einstein condensates in condensed matter physics.

In \cref{sec:science-trace}, Khatri--Rao test matrices allow us to estimate the \emph{partition function} of a large quantum-mechanical system, a key quantity that describes the system's thermodynamic behavior.
The tensor structure of the quantum Hamiltonian
motivates us to deploy algorithms based on tensor-structured sketches.
We can use Khatri--Rao test matrices to implement several
variance-reduced trace estimation algorithms~\cite{meyer2021hutch++,persson22,epperly24trace}
that include a low-rank approximation stage.
For the quantum systems that we study, these
approaches achieve much better accuracy than
simpler trace estimation methods.

Our experiments were run on the UC Berkeley SCF cluster, using 64 cores of a dual-socket AMD EPYC 7543 (2.8 GHz) with 500 GB of memory. All experiments were run in MATLAB v24.1.
Our code is publicly available at
\textbf{https://github.com/chriscamano/OSI-and-Structured-Random-Matrices}.
\begin{figure}[t]
    \centering
    \includegraphics[width=1\linewidth]{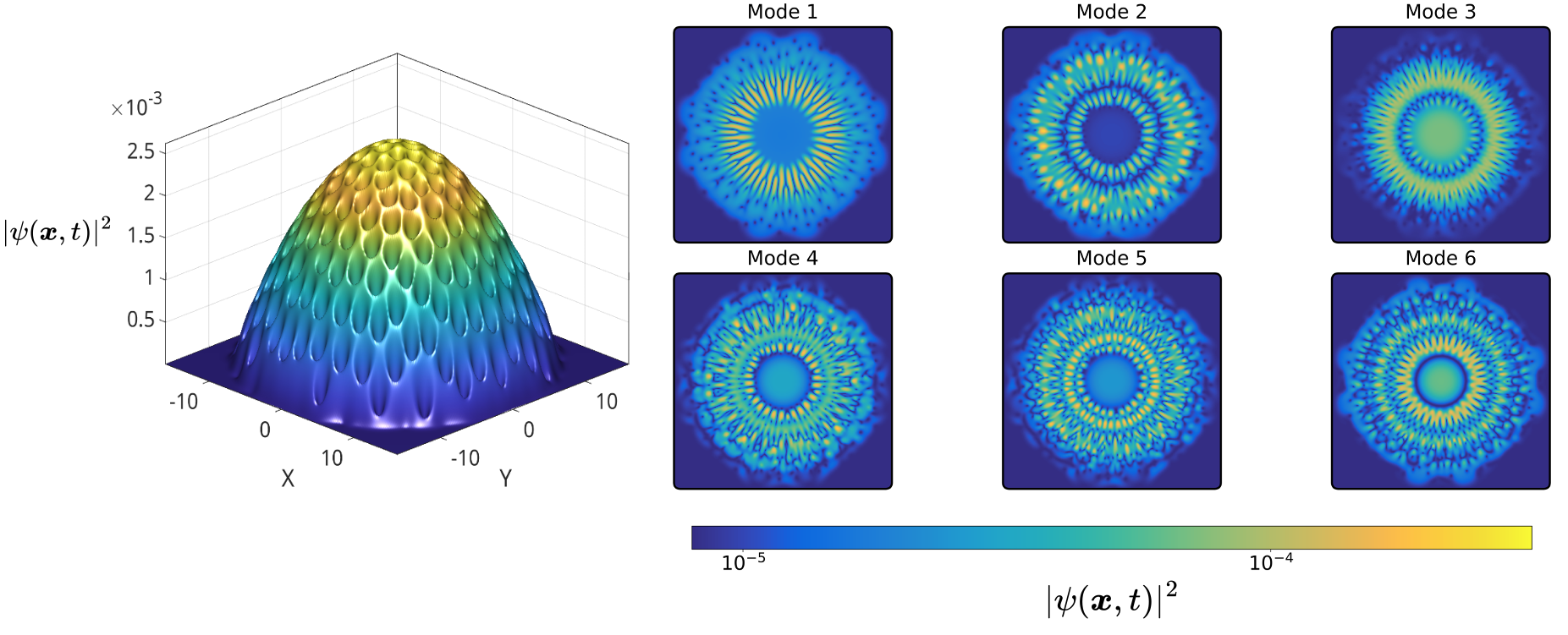}
    \caption[]{\textbf{Dynamics of imaginary-time Gross--Pitaevskii equation}.
    (\textit{left})
        Ground-state density $|{\psi}_{\text{final}}(\mathbf{x},t)|^2$ of a rotating Bose--Einstein condensate obtained by imaginary-time integration of the Gross--Pitaevskii equation.
    (\textit{right})
        The first six POD modes of the mean centered snapshot matrix $\mA$, computed using generalized Nystr\"om SVD (\cref{alg:gen_nystrom_svd}) with SparseStack test matrices (\cref{def:sparse-stack}).
        The POD modes display the dominant dynamics of quantized vortices forming within the BEC as it evolves towards a ground state.  See \cref{app:gross-pitaevskii} for more details.
        }
    \label{fig:gross-pitaevskii}
\end{figure}

\subsection{Application 1: Scientific data compression}
\label{sec:science-pod-modes}

In our first application, we use the SparseStack test matrix (\cref{def:sparse-stack}) with the generalized Nystr\"om algorithm (\cref{alg:gen_nystrom_svd}) to compress and analyze matrices of scientific simulation data.
Our experiments document up to a \textbf{12$\times$ speedup}
relative to algorithms with Gaussian test matrices
(\cref{def:gauss-test}).

Here is the basic setup.
Consider a dynamical system $\vx(t)\in\bbF^n$, and assemble a data matrix 
\begin{equation*}
    \mA \coloneqq \begin{bmatrix}
        \vertbar & \vertbar &        & \vertbar \\
        \vx(t_1) & \vx(t_2) & \cdots & \vx(t_d)    \\
        \vertbar & \vertbar &        & \vertbar 
    \end{bmatrix} \in \bbF^{n\times d}
\end{equation*}
by collating the states of the system at $d$ equally spaced times $t_1,\ldots,t_d$.
The snapshots $\vx(t_i)$ are typically obtained as outputs of a numerical differential equation solver,
although they could also be measurements from experiments.
We are interested in two tasks:
\begin{enumerate} \setlength{\itemsep}{0pt}
    \item \textbf{Compression.}
    To reduce the storage cost for the data matrix,
    we can compute and store a rank-$k$ approximation to $\mA$, where $k \ll \min\{n,d\}$.
    \item \textbf{POD modes.}
    To understand the dynamics, we can compute the $r$ dominant left singular vectors of $\mA$, which are known as the \textit{proper orthogonal decomposition} (POD) modes in this setting.
    The POD modes help us visualize the ``important structures'' in the dynamical system, and they provide a convenient basis for reduced-order modeling \cite{dawson23}.
\end{enumerate}
For both tasks, we can invoke the generalized Nystr\"om approximation (\cref{sec:gen-nystrom-via-osi}).
For compression, we use \cref{alg:gen_nystrom_outer}, which returns an approximation in outer-product form $\hat\mA = \mF\mG^\top$.
For the POD modes, we need the output in SVD form,
which is provided by \cref{alg:gen_nystrom_svd}.
\Cref{alg:gen_nystrom_outer,alg:gen_nystrom_svd} have essentially the same \emph{asymptotic runtime}, but the former is significantly faster because it avoids costly orthogonalization steps.
As an aside, the same methodology applies in the \emph{streaming setting}~\cite{tropp19}, where the snapshots $\vx(t_i)$ are delivered sequentially,
processed immediately, and then discarded.
Our paper does not include experiments with streaming data.
\begin{figure}[t]
    \centering
\resizebox{.9\linewidth}{!}{%
\input{figures/results/fig9.pgf}
}
\caption{\textbf{SparseStack: Scientific data compression.}  Relative error $\norm{\mA - \hat\mA}_{\rm F} / \norm{\mA}_{\rm F}$ (\emph{bottom left}), runtime  (\emph{top left}), and speed-up (\emph{right}) for compression of Gross--Pitaeksii (GPE) simulation data of size \(131{,}072\times 40{,}000\) using%
the generalized Nyström method (\cref{alg:gen_nystrom_outer}, GN, orange) or the generalized Nyström SVD (\cref{alg:gen_nystrom_svd}, GNSVD, blue) with Gaussian (solid circles, 
\cref{def:gauss-test}) or SparseStack test matrices (dashed squares, \cref{def:sparse-stack}).
We report the median of five runs; shaded regions are bounded by the 10\% and 90\% quantiles.
}
  \label{fig:scientific-compression}
\end{figure}
To demonstrate this approach, we apply it to data from the simulation of a rotating Bose--Einstein condensate (BEC), a quantum superfluid that exhibits quantum phenomena at macroscopic scales \cite{anderson95}.
To compute the ground state of this system, condensed matter physicists often solve the imaginary-time Gross--Pitaevskii equation,
a nonlinear partial differential equation whose solutions converge to a \emph{ground-state} in the large-time limit.
These ground states, visualized in the left panel of \cref{fig:gross-pitaevskii}, are known to display complex dynamics, such as quantized vortices whose cores are spatial regions where the probability density for the BEC vanishes.
These vortices are the ``holes'' visible in this plot.

Here is a summary of the experimental setup;
see \Cref{app:gross-pitaevskii} for details.
We simulate the imaginary-time Gross--Pitaevskii equation for $d = 40{,}000$ timesteps on a discretized $256\times 256$ grid,
resulting in a data matrix $\mA\in\bbC^{256^2\times 40{,}000}$.
Following standard practice, we center each snapshot so that
each column of $\mA$ has zero mean.%
We perform both the data compression and POD-mode tasks using generalized Nystr\"om approximation (\cref{alg:gen_nystrom_outer} or \ref{alg:gen_nystrom_svd}) with SparseStack test matrices
with row sparsity $\zeta = 4$.

The right-hand panels in~\cref{fig:gross-pitaevskii}
contain visualizations of the top $6$ POD modes (left singular vectors).
These POD modes reveal symmetric density patterns in the condensate, isolating the dominant dynamics that emerge during imaginary-time evolution.  They reflect%
the formation of quantized vortices in the BEC.

\Cref{fig:scientific-compression} charts
the errors and runtime of this method. 
As compared to a baseline experiment using Gaussian test matrices (open circles), the right panel shows that SparseStack test matrices (closed squares) improve runtimes by \textbf{up to $\boldsymbol{12\times}$ for compression (orange) and $\boldsymbol{4\times}$ for computation of POD modes (blue)}.
Absolute runtimes appear in the top left panel.
The bottom left panel confirms that sparse test matrices
result in no loss of accuracy.

\subsection{Application 2: Approximating the partition function
of a quantum system}
\label{sec:science-trace}

In our second application, we exploit
Khatri--Rao test matrices to approximate
the partition function of a quantum-mechanical system.
For this computation, we apply variance-reduced
trace estimators to a tensor-structured matrix.
As compared with earlier Khatri--Rao sketching methods,
our approach improves the accuracy by six to twelve orders
of magnitude---with minimal increase in runtime.

Let \(\mH\in\bbR^{2^\ell \times 2^\ell}\) denote the (symmetric) Hamiltonian matrix for a system of $\ell$ interacting quantum particles.
We are interested in computing the \emph{partition function} \(Z(\beta)\defeq \tr(\exp(-\beta\mH))\) of the system, where \(\beta > 0\) is a parameter called the \emph{inverse temperature}.
The partition function is a central quantity in statistical mechanics that reflects the properties of a quantum system in thermal equilibrium \newtext{\cite[Sec.~2]{kliesch2019properties}}.
Because the dimension of $\mH$ is exponential in the number of particles $\ell$, direct computation with $\mH$ rapidly becomes infeasible.

Fortunately, most quantum Hamiltonians natively possess tensor structure, and $\mH$ can often be represented using a compact tensor or tensor network format \cite{orus19}.
In particular, these tensor (network) formats allow us to access the matrix $\mH$ via Kronecker matvecs, enabling us to efficiently sketch $\mH$ with Khatri--Rao test matrices of order $\ell$.
Moreover, we can apply iterative methods to approximate tensor-structured matvecs with \(\mA\coloneqq\exp(-\beta\mH)\) \cite{al11,chen24}.
Thus, computing the partition function $Z(\beta)$ amounts to estimating
\(\tr(\mA)\) using a small number of Kronecker matvecs.

\begin{figure}[t]
    \centering
\resizebox{1\linewidth}{!}{%
\input{figures/results/trace_kr_multiple_beta2.pgf}
}
    \caption{\textbf{Khatri--Rao test matrices: Estimating a partition function}.
    Relative error $|Z(\beta)- \widehat{Z}(\beta)|/Z(\beta) $ for estimates $\widehat{Z}(\beta)$ of the partition function $Z(\beta)=\tr(\exp(-\beta\mH))$ computed using several stochastic trace estimators, i.e., Girard-Hutchinson, NA-Hutch++, Nyström++, and XNysTrace, with real spherical Khatri--Rao test matrices of tensor order $\ell=16$.
    We examine three values of the inverse temperature $\beta$.
    The markers track the median of 30 trials, and shaded regions are bounded by the 10\% and 90\% quantiles.
    }
    \label{fig:kr-trace}
\end{figure}

The most basic trace estimation algorithm is the \emph{Girard--Hutchinson trace estimator} \cite{girard87, hutchinson89}:
\[
    H_t(\mA) \defeq \tr(\mOmega^\top\mA\mOmega)
    \quad\text{where}\quad
    \text{$\mOmega$ is an isotropic test matrix with embedding dimension $t$.}
\]
This formula yields an unbiased estimate of $\tr(\mA)$,
and it requires $t$ matvecs with $\mA$.
This estimator%
is unbiased for any isotropic test matrix \mOmega and requires $t$ matrix--vector products with \mA to compute.
This estimator has been applied to partition function estimation in the past with ordinary (i.e., non-Kronecker) matvecs%
\cite{motta20,bravyi21,shen24}.
These earlier works do not fully exploit the tensor structure,
so the methods are limited to small tensor order $\ell$.
In order to reach larger values of $\ell$,
we can take $\mOmega$ to be a random Khatri--Rao test matrix.
The Girard--Hutchinson estimator $H_t$ with a Khatri--Rao test
matrix provably converges as the parameter $t$ increases,
but the convergence is very slow for large \(\ell\) \cite{feldman22,meyer23}, requiring \newtext{a large} number of matvecs to attain even one digit of accuracy; see \cref{fig:kr-trace}.

A recent line of work studies \emph{variance-reduced} algorithms for trace estimation, which augment the Girard--Hutchinson trace estimator with a low-rank approximation stage.
This strategy is exemplified by the \nahutchpp\, estimator \cite{meyer2021hutch++}, which returns
\[
    \nahutchpp(\mA)
    \defeq \tr(\hat\mA) + H_{t/2}(\mA - \hat\mA).
\]
In this expression, $\hat\mA$ is a rank-\(t/6\) generalized Nystr\"om approximation of $\mA$, computed with test matrices $\mOmega \in \bbF^{n\times (t/6)}$ and $\mPsi \in \bbF^{n\times (t/3)}$.
The algorithm expends $t$ matvecs to apply $\mA$ and $\mA^*$.

The variance of the NA-Hutch++ estimator equals \(\Var[H_{t/2}(\mA-\hat\mA)]\), which tends to be much smaller than the variance
of the Girard--Hutchinson trace estimator
when the input matrix $\mA$ is close to a low-rank matrix.
Fortunately, this is the case for quantum Hamiltonians,
where the eigenvalues of the matrix \(\mA = \exp(-\beta\mH)\)
often decay rapidly.

As we showed in \cref{sec:khatri-rao-small-sketch},
Khatri--Rao test matrices can serve as OSIs,
so they offer guarantees for low-rank approximation,
and they can be used to implement variance-reduced trace estimators,
such as \nahutchpp.
In the worst case, Khatri--Rao test matrices increase errors by a factor exponential in the tensor order $\ell$.
Even so, we can still hope that variance-reduced trace estimators with Khatri--Rao test matrices perform well in practice.

\Cref{fig:kr-trace} puts this hope to the test.
We deploy several variance-reduced trace estimators
using real spherical Khatri--Rao test matrices with tensor order $\ell$.
We estimate the partition function for the transverse field Ising model \cite{pfeuty1970}, a 1D quantum system whose Hamiltonian is given as 
\begin{equation*}
    \mH=-\sum_{i=1}^\ell \mZ_i \mZ_{\operatorname{mod}(i,\ell)+1}-h \sum_{i=1}^\ell \mX_i \in \mathbb{R}^{2^\ell \times 2^\ell}\quad \text{ with } \ell=16, \: h=10.
\end{equation*}
Here, $\mX_i$ and $\mZ_i$ denote Pauli operators acting on the $i$th site and $h$ is a coupling coefficient.
For simplicity, this experiment is restricted to the moderate value $\ell = 16$,
and our code works with an explicit representation of $\mH$ as a sparse matrix.
We test three variance-reduced trace estimators that are applicable to the Kronecker matvec access model, including NA-Hutch++, Nystr\"om++ \cite{persson22}, and XNysTrace \cite{epperly24trace}.
\Cref{app:trace-est} contains further implementation details.

\Cref{fig:kr-trace} shows the results.
We see that variance-reduced estimators significantly outperform the Girard--Hutchinson estimator when using the same number of matvecs with the input matrix.
\textbf{On the most extreme example, variance-reduced trace estimators attain accuracy $\boldsymbol{10^{-12}}$, while the Kronecker Girard--Hutchinson estimator fails to achieve even a single digit of accuracy.}
Additionally, we see that \nahutchpp, which uses generalized Nystr\"om, has consistently worse error than the other two variance-reduced estimators, which rely on psd Nystr\"om.
For partition function estimation, the overall evidence suggests that variance-reduced trace estimators with Khatri--Rao test matrices can be extremely effective.

\section{The OSI property: Proof strategies}
\label{sec:osi-proof-strats}

In \cref{sec:randnla-via-osi}, we saw that several randomized algorithms for low-rank approximation and linear regression are effective when implemented with OSI test matrices.  To complete the design of these algorithms, we must develop strategies for establishing the OSI property.

\Cref{prop:osi-spec} states that the OSI property admits
an equivalent spectral formulation.
As a consequence, we can
establish the OSI by producing lower bounds on the minimum
singular value of the test matrix, compressed to a subspace.
The literature on high-dimensional probability contains several
powerful techniques for bounding minimum singular values.
This section elaborates on these technical methods,
and it derives OSI guarantees for each of the major classes
of test matrices
discussed in this paper (SparseStack, SparseRTT, Khatri--Rao).

\Cref{sec:fourth-to-second} describes a method based on fourth-moment
bounds, which we use to identify Khatri--Rao test matrices
with constant injectivity.
\Cref{sec:small-ball} introduces the small-ball method,
which allows us to prove that certain Khatri--Rao test matrices
are subspace injections at an embedding dimension \(k=\cO(r)\).
\Cref{sec:gaussian-compare} outlines the Gaussian comparison method,
which lets us analyze sparse test matrices.
Last, \cref{sec:matrix-concentration} comments on some
matrix concentration tools that may be relevant for
other problems.

\subsection{Methodology 1: Fourth-moment bounds}
\label{sec:fourth-to-second}

The first strategy deduces the injectivity property
from a moment comparison assumption.
This theorem was established by Oliveira~\cite[Thm.~1.1]{oliveira16}
using the PAC-Bayesian method.  Tropp~\cite[Thm.~5.2]{tropp25} recently
obtained an alternative proof based on Gaussian comparison.

\begin{importedtheorem}[OSI via fourth moment bound, \protect{\cite[Thm.~5.2 and Rem.~5.1]{tropp25}}]
    \label{impthm:fourth-to-second}
    Let $\vomega \in \F^{d}$ be an isotropic random vector.
    Suppose that the second and fourth moments of all linear marginals of $\vomega$
    are comparable up to a constant factor $h_{4} \geq 1$.  That is,
    \begin{equation} \label{eqn:fourth-second-moment}
        \E\big[|\langle \vomega, \vu\rangle|^4 \big]
            \leq h_{4} \cdot
            \big(\E \big[|\langle\vomega,\vu\rangle|^2 \big]\big)^2
            = h_{4}%
        \hspace{1cm}
        \text{for each unit-norm vector $\vu\in\F^d$}.
    \end{equation}
    Consider the random test matrix 
    \begin{equation*}
        \mOmega = \frac{1}{\sqrt{k}} \begin{bmatrix}
            \vomega_1 & \cdots & \vomega_k
        \end{bmatrix} \in \F^{d\times k}
        \quad\text{where the columns $\vomega_{j} \sim \vomega$ iid.}
    \end{equation*}
    Then $\mOmega$ satisfies the $(r, \alpha)$-OSI property
    when the embedding dimension $k \geq \rC \cdot (1-\alpha)^{-2} \cdot h_{4} r$.
\end{importedtheorem}

The equivalence~\eqref{eqn:fourth-second-moment} of all second and fourth moments is a rather strong assumption.
It requires the random vector $\vomega$ to have controlled variance in all directions.
For example, standard Gaussian and standard Rademacher vectors satisfy the comparison~\eqref{eqn:fourth-second-moment} with $h_{4} = 3$.
On the other hand, for a sparse Rademacher random vector in $\F^d$ with $\xi$ nonzero entries, $h_{4} \geq d / \xi$, so \cref{impthm:fourth-to-second} gives pessimistic bounds for the embedding dimension of sparse test matrices.

\subsubsection{Application: Khatri--Rao embeddings with constant injectivity}

As an example of the fourth-moment theorem, we can establish \cref{corol:kr-large-sketch-osi}, which states that Khatri--Rao test matrices are OSIs
when the embedding dimension is exponential in the order $\ell$ of the
tensor product.  In fact, \cref{corol:kr-large-sketch-osi} follows immediately when we combine \cref{impthm:fourth-to-second} with the following result of Meyer \& Avron~\cite{meyer23}.

\begin{importedtheorem}[Random Kronecker products: Fourth-moment bounds, \protect{\cite[Thm.~4]{meyer23}}] \label{impthm:meyer-avron}
\newtext{Let $\vnu \in \F^{d_0}$ be an isotropic random vector} that satisfies the moment bound
\[
\E \big[ |\langle \vnu, \va \rangle|^4 \big]
    \leq \rC_{\vnu}%
\quad\text{for each unit-norm vector $\va \in \F^{d_0}$.}
\]
Construct the Kronecker product $\vomega \coloneqq \vomega^{(1)} \otimes \dots \otimes \vomega^{(\ell)} \in \F^{d_0^\ell}$ where the factors $\vomega^{(i)} \sim \vnu$ are iid.  Then
\begin{equation} \label{eqn:ma-vector}
\E \big[ |\langle \vomega, \vu\rangle|^4 \big]
    \leq \rC_{\vnu}^{\ell}%
    \quad\text{for each unit-norm vector $\vu \in \F^{d_0^\ell}$.}
\end{equation}
Equivalently,
\begin{equation} \label{eqn:ma-matrix}
\E \big[ (\vomega^\top\mM\vomega)^2 \big]
    \leq \rC_{\vnu}^{\ell} \cdot (\tr(\mM))^2
    \quad\text{for each psd $\mM \in \F^{d_0^\ell \times d_0^\ell}$.}
\end{equation}
\end{importedtheorem}

To recover the original formulation~\cite[Thm.~4]{meyer23}
of \cref{impthm:meyer-avron}, we confirm the equivalence of
\cref{eqn:ma-vector} and \cref{eqn:ma-matrix}.
By choosing a rank-one matrix $\mM = \vu \vu^\top$,
it is clear that the matrix formulation~\eqref{eqn:ma-matrix}
implies the vector formulation~\eqref{eqn:ma-vector}.
For the converse, rewrite~\eqref{eqn:ma-matrix} as
\[
    \sup_{\substack{\tr(\mM) \leq 1, \text{$\mM$ psd}}}\
    \E\bigl[ (\vomega^\top \mM \vomega)^2 \bigr] \leq \rC_{\vnu}^{\ell}.
\]
A convex function on a compact set attains its supremum
at an extreme point.  In this case, each extreme point
is a rank-one matrix $\mM = \vu \vu^{\top}$ where
$\vu$ is a unit-norm vector.

For a random Kronecker product vector, \cref{impthm:meyer-avron} indicates that the moment comparison factor $h_{4} = \rC_{\vnu}^\ell$ scales exponentially with the tensor order $\ell$.  As a consequence, our result \cref{corol:kr-large-sketch-osi} provides an $(r, \nicefrac{1}{2})$-OSI guarantee for a Khatri--Rao test matrix at some embedding dimension $k = \Theta( \rC_{\vnu}^{\ell} r )$, while it offers no information about the test matrix when $k = o( \rC_{\vnu}^{\ell} r)$.

This limitation reflects a genuine phenomenon.
Consider a Khatri--Rao test matrix $\mOmega \in \R^{2^{\ell} \times k}$, generated from a $2$-dimensional Rademacher base distribution $\vnu$ on $\R^2$.
For this base distribution,
the comparison factor $\rC_{\vnu} = 2$,
and \cref{corol:kr-large-sketch-osi} demands
an embedding dimension $k = \cO( 2^\ell r )$
to achieve an $(r, \nicefrac{1}{2})$-OSI.
The discussion in \cref{sec:khatri-rao-overwhelming-orthogonality}
implies that we must choose the embedding dimension $k = \Omega(2^{\ell})$
to achieve a nontrivial OSI guarantee; see also~\cite[Thm.~4]{meyer25}.
Confirming these predictions, our empirical work (\cref{sec:kr-rad-bad-numerics}) provides clear evidence that the Rademacher Khatri--Rao
test matrix fails catastrophically when $k = \Theta(r)$ and the tensor order \(\ell\) is large.

\subsection{Methodology 2: Small-ball probabilities}
\label{sec:small-ball}

We must exploit an alternative structure to obtain guarantees for Khatri--Rao test matrices with proportional embedding dimension $k = \cO(r)$.
To that end, we introduce an injectivity theorem based on \emph{the small-ball method}.
This strategy was originally developed by Mendelson and coauthors~\cite{Men14:Remark-Diameter,mendelson2015learning,koltchinskii2015bounding,Men18:Learning2}.
The version here%
is adapted from several other sources: a result~\cite[Thm.~1.6]{GKK+22:Geometry-Polytopes} of Gu{\'e}don et al., Tropp's lecture notes~\cite[Thm.~19.1]{tropp23hdp}, and Vershynin's textbook~\cite[Exer.~8.27]{Ver25:High-Dimensional-Probability-2ed}.
Our result improves the constants and extends the bound to the complex setting.

\begin{theorem}[OSI via small-ball method]
    \label{thm:osi-from-small-ball}
    Fix a constant $\tau > 0$.
    \newtext{Let $\vomega \in \F^{d}$ be an isotropic random vector that} satisfies the \emph{small-ball probability} bound
    \begin{equation} \label{eqn:small-ball}
    \prob\bigl\{ | \langle \vomega, \vu \rangle |^2 \leq \tau \bigr\}
        \leq \frac{1}{100}
        \quad\text{for each unit-norm vector $\vu \in \F^d$.}
    \end{equation}
    Consider the random test matrix 
    \begin{equation*}
        \mOmega = \frac{1}{\sqrt{k}} \begin{bmatrix}
            \vomega_1 & \cdots & \vomega_k
        \end{bmatrix} \in \bbF^{d\times k}
        \quad\text{where the columns $\vomega_i \sim \vomega$ iid.}
    \end{equation*}
    Then \(\mOmega\) enjoys the $(r,\nicefrac{\tau}{10})$-OSI property when the embedding dimension satisfies
    \[
    \begin{aligned}
    k &\geq \max\big\{3r,\ 131 \big\}, & \text{when $\:\F = \R$}; \\
    k &\geq \max\big\{\newtext{16r,\ 96} \big\}, & \text{when $\:\F = \C$}.
    \end{aligned}
    \]
\end{theorem}

The small-ball method captures a fundamental
insight about injectivity.
Indeed, the main obstacle to injectivity
is the existence of a fixed unit-norm vector
that is (almost) orthogonal
to the random vector $\vomega$.
The statistic $\tau$ in the small-ball
condition~\eqref{eqn:small-ball}
quantifies the severity of this issue.
\newtext{
    The constant \(1/100\) is arbitrary; decreasing the probability will increase the injectivity parameter.
    See the full proof for details.
}

\begin{proof}[Proof sketch]
The argument relies on two ideas.  First, we can bound the minimum singular value below by introducing indicator random variables:
\[
    \sigma_{\rm min}^2(\mOmega^{\top} \mQ)
    = \min_{\norm\vu_2=1}\ \frac1k\sum_{i=1}^k | \langle \vomega_i, \mQ\vu \rangle |^2
    \geq \min_{\norm\vu_2=1}\ \frac1k\sum_{i=1}^k \tau \cdot \mathbbm1_{\{| \langle \vomega_i, \mQ\vu \rangle |^2 \geq \tau \}}.
\]
The first identity follows from the variational definition of the
minimum singular value.  The lower bound depends on the numerical inequality
$t \geq \tau \cdot \mathbbm1_{\{t \geq \tau\}}$, valid for $t \geq 0$.

Second, we regard the indicators as Boolean classifiers,
which suggests an argument based on the \newtext{Vapnik–Chervonenkis (VC)} dimension.
The VC law of large numbers~\cite[Thm.~8.3.15]{Ver25:High-Dimensional-Probability-2ed} states that the empirical classification
probability concentrates uniformly around its expectation:
\[
    \frac1k\sum_{i=1}^k \mathbbm1_{\{ |\langle \vomega_i, \mQ\vu \rangle |^2 \geq \tau\}}
    \approx \prob\big\{ |\langle \vomega, \mQ\vu \rangle |^2 \geq \tau\big\} \geq 0.99.
\]
The family of linear classifiers in $\F^d$ has VC dimension $\Theta(d)$.
With some additional arguments to handle the absolute values, we determine that $k = \Theta(d)$ samples suffice to obtain concentration simultaneously for all $\vu \in \F^d$.
See~\cref{app:small-ball-analysis} for the details.
\end{proof}

The small-ball method yields OSI guarantees for a wide range
of random test matrices.  It is well-suited for problems where
the random vector $\vomega$ is continuous, with a bounded density.
As with \cref{impthm:fourth-to-second},
when the random vector $\vomega$ has controlled variance
in every direction, the small-ball condition~\eqref{eqn:small-ball}
holds.
In fact, we can sometimes activate the small-ball condition
in settings where the random vector $\vomega$ lacks finite
polynomial moments of any order.

On the other hand, the small-ball method is ill-suited for
problems where the distribution of the random vector $\vomega$
is spiky.
The condition~\eqref{eqn:small-ball} {never}
holds for a (very) sparse random vector $\vomega$.
The small-ball condition also fails when $\vomega$
is a Kronecker product of iid Rademacher vectors
because of the overwhelming orthogonality property
(\cref{sec:khatri-rao-overwhelming-orthogonality}).

\subsubsection{Application: Khatri--Rao embeddings with proportional embedding dimension} \label{sec:khatri-rao-proportional-proof}

We can deploy the small-ball technology to study Khatri--Rao test matrices.
To that end, we need a small-ball probability bound
for the Kronecker product of random vectors.
This property is somewhat delicate, especially when the base distribution
is discrete.
As a consequence, it is natural to assume that the base distribution is continuous, with an anticoncentration property.
Recently, Hu \& Paouris \cite{hu2024small} established a result of this type for \emph{real} random vectors.

\begin{importedtheorem}[Random Kronecker products: Small-ball probability, \protect{\cite[Thms.~1.1 and 1.4]{hu2024small}}] \label{thm:hupao}
    \newtext{Let  $\vnu \in \R^{d_0}$ be a \emph{real} random vector}
    with independent entries, each
    drawn from an absolutely continuous distribution whose density
    is uniformly bounded by $M$.  Construct the Kronecker product
    $\vomega \coloneqq \vomega^{(1)} \otimes \cdots \otimes \vomega^{(\ell)} \in \R^{d_0^\ell}$ where the factors $\vomega^{(i)} \sim \vnu$ iid.
    Then we have the small-ball probability bound
    \[
    \prob\left\{ | \langle \vomega, \vu \rangle | \leq (\rC M)^{-\ell} \right\}
    \leq \frac{1}{100}
    \quad\text{for each unit-norm vector $\vu \in \R^{d_0^\ell}$.}
    \]
    Alternatively, if the distribution of $\vnu \in \R^{d_0}$ is log-concave, then the same bound holds with $M = 1$.
\end{importedtheorem}

To establish our result on Khatri--Rao products at the proportional embedding dimension (\cref{corol:kr-small-sketch}), we simply invoke the small-ball method
(\cref{thm:osi-from-small-ball}) with the small-bound probability bound from \cref{thm:hupao}.

We can extend \cref{thm:hupao} to handle the
real spherical base distribution.
Here is a sketch of the argument.
Consider the base distribution $\vnu$ that is
uniform on the \emph{ball} with radius $\sqrt{d_0}$.
This base distribution is log-concave,
so \cref{thm:hupao} provides a small-ball
probability bound with $M = 1$.
To pass to the spherical distribution,
we dilate each tensor factor $\vomega^{(i)}$
to increase its norm to $\sqrt{d_0}$.
This operation only decreases the small-ball probability.

\subsection{Methodology 3: Comparison with a Gaussian matrix}
\label{sec:gaussian-compare}

\newtext{We cannot treat sparse test matrices
with either of the first two methodologies because
the columns are not independent, and the sparsity
prevents us from bounding either fourth moments or
small ball probabilities.
To analyze sparse test matrices,
we rely on the Gaussian comparison method
recently devised by Tropp~\cite{tropp25}.}
This technique allows us
to obtain a lower bound for the minimum eigenvalue of a random psd
matrix by comparison with a suitable Gaussian matrix.

To state Tropp's result, we introduce several functions that capture
the variability of a random matrix.
For a random self-adjoint matrix $\mX \in \F^{d\times d}$,
we define the variance and the second moment in the
direction of the self-adjoint matrix $\mM \in \F^{d\times d}$:
\[
\Var[\mX](\mM) \coloneqq \Var\bigl[ \tr(\mM \mX) \bigr]
\quad\text{and}\quad
\Mom[\mX](\mM) \coloneqq \E\bigl[ (\tr(\mM \mX))^2 \bigr].
\]
The \emph{weak variance} is defined as
$\sigma_*^2(\mX) \coloneqq
    \max_{\ \norm{\vu}_2=1}\ \Var[\vu^\top\mX\vu]$.

\begin{importedtheorem}[Gaussian comparison for random psd matrices, \protect{\cite[Thm.~2.3]{tropp25}}]
    \label{impthm:comparison-iid-sum}
    \newtext{
        Let \(\mW\in\F^{d \times d}\) be a random psd matrix with two finite moments:%
    }
    $\E{} \norm{ \mW }_{\rm F}^2 < + \infty$.
    Construct a self-adjoint Gaussian matrix $\mX \in \F^{d \times d}$ that satisfies
    \begin{equation}
    \label{eq:gauss-comparisson-needs}
    \E[ \mX ] = \E[ \mW ]
    \quad\text{and}\quad
    \Var[\mX](\mM) \geq \Mom[\mW](\mM)
    \quad\text{for all self-adjoint $\mM \in \F^{d \times d}$.}
    \end{equation}
    Define the random matrices
    \begin{equation} \label{eq:gauss-comparison-mtx}
    \begin{aligned}
        \mY &\coloneqq \sum_{i=1}^k \mW_i &&\quad\text{where $\mW_i \sim \mW$ iid;} \\
        \mZ &\coloneqq \sum_{i=1}^k \mX_i &&\quad\text{where $\mX_i \sim \mX$ iid.}
    \end{aligned} \end{equation}
    Then, with probability at least $1 - \delta$,
    the minimum eigenvalues admit the comparison
    \[
    \lambda_{\min}(\mY) \geq \E[ \lambda_{\min}(\mZ) ]
        - \sqrt{2 \sigma_*^2(\mZ) \log(2d/\delta)}.
    \]
\end{importedtheorem}

\Cref{impthm:comparison-iid-sum} requires more effort than
the other methodologies, but it has a wide scope of application.
(For instance, Tropp invokes the result%
to establish \cref{impthm:fourth-to-second}.)
In our setting, the Gaussian comparison method can be used to verify
the $(r, \alpha)$-OSI property for a test matrix $\mOmega$
whose outer product decomposes%
as a sum of%
iid terms:
\[
\mOmega \mOmega^\top = \sum_{j=1}^k \mW_j
\quad\text{where%
$\mW_j \sim \mW$ iid.}
\]
As an example of this strategy, \cref{sec:sparsecol-pf}
provides a detailed analysis of the
SparseCol test matrix.%
\Cref{app:sprase-stack} presents the analysis of the
SparseStack test matrix.%

\subsubsection{Additional background} \label{sec:gauss-background}

\newtext{
    To employ \cref{impthm:comparison-iid-sum}, we need to construct a self-adjoint Gaussian matrix \mX that satisfies \cref{eq:gauss-comparisson-needs}.
    We construct \mX on a case-by-case basis, relating the exact form of \(\Mom[\mW]\) to the variances of well-studied Gaussian matrix models.
    To formalize this process, we collect a few additional facts from Tropp's paper~\cite{tropp25}.
}
Let $\mW, \mX \in \F^{d \times d}$ be random self-adjoint matrices,
and let $\mQ \in \F^{d \times r}$ be a fixed matrix.
Then
\begin{equation} \label{eqn:gaussian-compare-conjugation}
\begin{aligned}
\E[\mX] &= \E[\mW] & \text{implies} && \E[ \mQ^\top \mX \mQ ] &= \E[ \mQ^\top \mW \mQ]; \\
\Var[\mX] &\geq \Mom[\mW] & \text{implies} &&
    \Var[ \mQ^\top \mX \mQ ] &\geq \Mom[ \mQ^\top \mW \mQ ].
\end{aligned}
\end{equation}
The inequalities between the variance function and the second-moment
function hold pointwise for each self-adjoint argument $\mM \in \F^{d \times d}$.
In words, first and second moments are equivariant under
conjugation~\cite[Prop.~3.1]{tropp25}.

Next, we summarize the statistics of several common Gaussian matrix models.
\Cref{table:gaussian-library} lists the mean, variance, expected maximum eigenvalue, and weak variance for each instance.
The first example is a \emph{scalar Gaussian} matrix, defined as $g \mI_d$
where $g \sim \cN(0,1)$.
The second is an \emph{iid diagonal Gaussian} matrix, namely the diagonal matrix $\mD_d \in \F^{d \times d}$ with iid real diagonal entries drawn from
$\cN(0,1)$.
The third is the \emph{compressed diagonal Gaussian} matrix $\mQ^\top \mD_d \mQ$, where $\mQ \in \F^{d \times r}$ is an orthonormal matrix that depends on the context.
Last, we consider a \emph{Gaussian unitary ensemble (GUE)} matrix:
\[
\mG_{\rm gue}^{(d)} = (\mG + \mG^{\top}) / \sqrt{2}
\quad\text{where $\mG \in \C^{d \times d}$ has iid complex $\cN_{\C}(0,1)$ entries.}
\]
The GUE matrix is rotationally invariant.
\begin{importedlemma}[Rotational invariance, \protect{\cite[Sec.~3.7.4]{tropp25}}]
    \label{implem:ortho-conjugate-gue}
    \newtext{Let \(\mQ\in\bbC^{d \times r}\) be orthonormal.}
    Then the matrix \(\mQ^\top\mG_{\rm gue}^{(d)}\mQ\) follows the same distribution as \(\mG_{\rm gue}^{(r)}\).
\end{importedlemma}

\begin{table}[t]
\centering
\begin{tabular}{@{}lcclll@{}} \toprule
	Description & Notation $(\mZ)$ & \(\E[\mZ]\) & \(\Var[\mZ](\mM)\)
    & \(\E[\lambda_{\rm max}(\mZ)]\) & \(\sigma_*^2(\mZ)\) \\ \midrule
	Scalar Gaussian & \(g\mI_d\)
        & \(\mat0\)
        & \((\tr(\mM))^2\)
		& \(0\)
		& \(1\)
		\\[0.25em]
	IID real diagonal & \(\mD_d\)
        & \(\mat0\)
        & \(\sum_{i=1}^d (\mM_{ii})^2\)
		& \(\leq \sqrt{2\log d}\)
		& \(1\)
		\\[0.25em]
	Compressed diagonal & \(\mQ^\top\mD_d\mQ\)
        & \(\mat0\)
        & \(\sum_{i=1}^r ((\mQ\mM\mQ^\top)_{ii})^2\)
        & \(\leq \sqrt{2\mu(\mQ)\log r}\) 
        & \(\leq \mu(\mQ)\)
		\\[0.25em]
    GUE & \(\mG_{\rm gue}^{(d)}\)
        & \(\mat0\)
        & \(\norm\mM_{\rm F}^2\)
		& \(2\sqrt{d}\)
		& \(1\)
		\\[0.25em]
		\bottomrule
\end{tabular}
\caption{
    \textbf{Gaussian matrix models,
    \protect{\cite[Secs.~3 and~6.3.2]{tropp25}}.}
    Entries with a leading \(\leq\) indicate inequalities;
    entries without a leading \(\leq\) are exact computations.
    The symbol \(\mu(\mQ)\) denotes the coherence of the matrix $\mQ\in\bbF^{d \times r}$ (\cref{sec:prelims}).
    See~\cref{sec:gauss-background} for other definitions.
}
\label{table:gaussian-library}
\end{table}

\subsubsection{Application: SparseCol embeddings} \label{sec:sparsecol-pf}

This section deploys the Gaussian comparison method to
establish \cref{thm:sparse-indep-cols-osi},
which studies the injectivity properties
of the SparseCol test matrix (\cref{def:sparsecol}).
The argument depends on a moment calculation,
encapsulated in the next lemma.
Throughout this section, we work in the complex
field, as the real field is a special case.

\begin{lemma}[SparseCol: Moment computations]
    \label{lem:indep-cols-moments}
    Suppose that \(d \geq 2\).
    Let \(\vomega \in \C^d\) be the sparse random vector
    with $\xi$ nonzero entries, appearing in \Cref{def:sparsecol} of the SparseCol matrix.
Define the random matrix $\mW = k^{-1} \vomega \vomega^\top \in \C^{d\times d}$.
The first moment satisfies~ $\E{}[\mW] = k^{-1} \Id$.
The second moment in the direction of the self-adjoint
matrix $\mM \in \C^{d\times d}$ satisfies
\newtext{
    \begin{align*}
        \Mom[\mW](\mM)
        &\leq \frac1{k^2}\left(
            \frac{d}{\xi} \sum_{i=1}^d (\mM_{ii})^2
            + 2 \frac{\xi-1}{\xi}(\tr(\mM))^2
            + 4 \frac{\xi-1}{\xi}\norm{\mM}_{\rm F}^2
        \right) \\
        &\leq
        \frac1{k^2}\left(
            \frac{d}{\xi} \sum_{i=1}^d (\mM_{ii})^2
            + 2 (\tr(\mM))^2
            + 4 \norm{\mM}_{\rm F}^2 \right).
    \end{align*}
}
\end{lemma}

The proof of Lemma~\ref{lem:indep-cols-moments}
is similar with moment computations from~\cite{tropp25},
so we postpone it to \Cref{app:sparse-col}.
Let us continue with the proof of
\Cref{thm:sparse-indep-cols-osi}.

\begin{proof}[Proof of \Cref{thm:sparse-indep-cols-osi}]
To apply~\cref{impthm:comparison-iid-sum},
we must develop a Gaussian model
that compares with the SparseCol random matrix.
Define the random matrix $\mW$ as in \cref{lem:indep-cols-moments}.
Consider a Gaussian random matrix \(\mX\in\bbC^{d \times d}\)
of the form
\begin{equation} \label{eqn:sparsecol-X-def}
    \mX = \frac1k \left( \mI_d + \sqrt{\frac{d}{\xi}} \mD_d + \sqrt{2} g \mI_d + 2 \mG_{\rm gue}^{(d)}\right).
\end{equation}
As in~\cref{table:gaussian-library},
the independent Gaussian random matrices $\mD_d$ and $ g\mI_d$ and $\mG_{\rm gue}^{(d)}$ are diagonal, scalar, and GUE.
Using this dictionary, we quickly determine that
\(\E[\mX] = k^{-1} \mI = \E[\mW] \)
and that%
\begin{align*}
	\Var[\mX](\mM)
	&= \frac1{k^2}\left(\frac{d}{\xi}\Var[\mD](\mM) + 2 \Var[g \mI_d](\mM) + 4\Var[\mG_{\rm gue}^{(d)}](\mM)\right) \\
	&= \frac1{k^2}\left(\frac{d}{\xi}\sum_{i=1}^d (\mM_{ii})^2 + 2 (\tr(\mM))^2 + 4\norm{\mM}_F^2\right)
	\geq \Mom[\mW](\mM).
\end{align*}
The inequality follows from \cref{lem:indep-cols-moments}.
In other words, the moments of $\mX$ compare with
the moments of $\mW$.
Now, fix any matrix \(\mQ\in\bbF^{d \times r}\)
with orthonormal columns.
Since moments are equivariant under
conjugation~\cref{eqn:gaussian-compare-conjugation},
the random matrices $\mQ^{\top} \mX \mQ$
and $\mQ^{\top} \mW \mQ$ admit the same moment comparison.

Introduce a SparseCol random matrix $\mS \in \F^{d \times k}$,
as in the statement of \cref{thm:sparse-indep-cols-osi}.
Observe that
\[
\mS\mS^\top = \sum_{i=1}^k \mW_i
\quad\text{where $\mW_i \sim \mW$ iid.}
\]
We wish to control the quantity \(\lambda_{\rm min}(\mQ^\top\mS\mS^\top\mQ) = \lambda_{\rm min}(\sum_{i=1}^k \mQ^\top\mW_i\mQ)\).
To that end, select
\[
	\mY \defeq \sum_{i=1}^k \mQ^\top\mW_i\mQ
	\quad\text{and}\quad
	\mZ \defeq \sum_{i=1}^k \mQ^\top\mX_i\mQ
    \quad\text{where $\mX_i \sim \mX$ iid.}
\]
Both $\mY$ and $\mZ$ are the sums of iid matrices
in $\F^{r \times r}$
with the appropriate moment relationships,
so we can invoke \Cref{impthm:comparison-iid-sum}.
With probability at least \(1-\delta\),
\[
	\lambda_{\rm min}(\mQ^\top\mS\mS^\top\mQ)
	= \lambda_{\rm min}(\mY)
	\geq \E[\lambda_{\rm min}(\mZ)] - \sqrt{2\sigma_*^2(\mZ)\log(2r/\delta)}
\]
It remains to bound the expected minimum eigenvalue below and
to bound the weak variance above.

To achieve this goal, we must examine the distribution of $\mZ$.%
Owing to the definition~\eqref{eqn:sparsecol-X-def} of $\mX$,
\[
	\sum_{i=1}^k \mX_i
	\sim \mI_d + \sqrt{\frac{d}{k \xi}} \mD_d + \sqrt{\frac{2}{k}} g \mI_d + \sqrt{\frac{4}{k}} \mG_{\rm gue}^{(d)}.
\]
This statement follows from%
the stability of Gaussian distributions.
\Cref{implem:ortho-conjugate-gue} states that
\(\mQ^\top\mG_{\rm gue}^{(d)}\mQ \sim \mG_{\rm gue}^{(r)}\)
because $\mQ \in \F^{d\times r}$ is an orthonormal matrix.  Thus,
\[
    \mZ
    = \mQ^\top\left(\sum_{i=1}^k \mX_i\right)\mQ
    \sim \mI_r + \sqrt{\frac{d}{k \xi}} \mQ^\top\mD_d\mQ + \sqrt{\frac{2}{k}} g \mI_r + \sqrt{\frac{4}{k}} \mG_{\rm gue}^{(r)}.
\]
According to \cref{table:gaussian-library},
the expected minimum eigenvalue of $\mZ$ admits the lower bound
\begin{align*}
    \E\left[\lambda_{\rm min}(\mZ)\right]
    &\geq 1 - \sqrt{\frac d{k\xi}} \E\left[\lambda_{\rm max}(\mQ^\top\mD_d\mQ)\right] - \sqrt{\frac{2}{k}}\E\left[\lambda_{\rm max}(g \mI_r)\right] - \sqrt{\frac{4}{k}} \E\bigl[\lambda_{\rm max}\bigl(\mG_{\rm gue}^{(r)}\bigr)\bigr] \\
    &\geq 1 - \sqrt{\frac d{k\xi} \cdot 2\mu(\mQ)\log r} - \sqrt{\frac{4}{k} \cdot 4r}.%
\end{align*}
The coherence \(\mu(\mQ)\) is defined in \cref{sec:prelims}.
An upper bound for the weak variance \(\sigma_*^2(\mZ)\)
follows from the subadditivity property: $\sigma_*^2(\mA+\mB) \leq \sigma_*^2(\mA)+\sigma_*^2(\mB)$ for statistically independent $\mA$ and $\mB$ \cite[Sec.~3.6]{tropp25}.  Using the calculations from~\cref{table:gaussian-library},
\begin{align*}
	\sigma_*^2(\mZ)
	\leq%
    \frac{d}{k \xi} \sigma_*^2(\mQ^\top\mD_d\mQ) + \frac{2}{k} \sigma_*^2(g \mI_r) + \frac{4}{k} \sigma_*^2(\mG_{\rm gue}^{(r)})
	\leq \frac{d \mu(\mQ)}{k \xi} + \frac{6}{k}.
\end{align*}
Combine these results to reach a bound
for the minimum eigenvalue of $\mY$.
With probability at least \(1-\delta\),%
\[
    \lambda_{\rm min}(\mQ^\top\mS\mS^\top\mQ)
    \geq 1 - \sqrt{\frac{2d\mu(\mQ)}{k\xi} \log r} - \sqrt{\frac{16r}k} - \sqrt{2\left(\frac{d \mu(\mQ)}{k \xi} + \frac{6}{k}\right)\log(2\newtext{r}/\delta)}.
\]
In particular, we can choose the column sparsity parameter \(\xi = \cO\bigl(\newtext{\frac dr} \mu(\mQ) \log(r/\delta)\bigr)\)
and the embedding dimension \(k = \cO(\max\{r, \log(\newtext{r}/\delta)\})\).
We arrive at the bound
\[
	\lambda_{\rm min}(\mQ^\top\mS\mS^\top\mQ)
	\geq \nicefrac{1}{2}
\]
with probability at least \(1-\delta\).
\end{proof}

\subsection{Methodology 4: Matrix concentration}
\label{sec:matrix-concentration}

We would be remiss not to mention another established
technique for obtaining injectivity properties of
random matrices, based on matrix concentration tools.
\newtext{
    The articles \cite{Tro15:Introduction-Matrix,
    BvH24:Universality-Sharp} give a starting point for readers new to this literature.
}

Several of the standard matrix concentration inequalities
from Tropp's monograph~\cite{Tro15:Introduction-Matrix}
yield lower bounds on the minimum
eigenvalue of a self-adjoint random matrix,
expressed as an independent sum.
In particular, the lower matrix
Chernoff inequality~\citep[Thm.~5.1.1]{Tro15:Introduction-Matrix}
provides a good bound on the minimum eigenvalue
of a psd random matrix under minimal conditions.
For example, Tropp used this strategy
to analyze SRTT test matrices in his paper~\cite{tropp11SRHT}.
Unfortunately, classic matrix concentration tools
are usually not powerful enough to achieve the $(r, \alpha)$-OSI
property unless the embedding dimension $k = \Omega(r \log r)$.
Since we aim to reach the optimal scaling $k = \cO(r)$,
we must employ alternative methods.

Over the last several years, van Handel and coauthors
have developed more advanced matrix concentration tools,
including universality methods~\cite{BvH24:Universality-Sharp}
and intrinsic freeness results~\cite{BBvH23:Matrix-Concentration}.
These techniques can also be used
to prove injectivity results.
For example, Chenakkod \etal \cite{chenakkod24,chenakkod25} employed the universality method to study the SparseStack test matrix.
They achieve an \((r,1-\eps,1+\eps)\)-OSE with%
optimal dependence on \(\eps\),
with an optimal embedding dimension \(k=\cO(r)\),
but with a suboptimal row sparsity \(\zeta = \cO(\log^3 r)\).
As compared with the analysis via Gaussian comparison
(\cref{thm:sparse-stack-osi}), this approach results
in excess logarithmic factors.%
It is likely that contemporary matrix concentration
techniques can be used to establish injectivity properties
for other types of random test matrices.
\section{\newtext{Conclusion}} \label{sec:conclusion}

\newtext{This paper presents the theory of oblivious subspace injections, justifying the use of such test matrices in randomized linear algebra algorithms and presenting the SparseStack, SparseRTT, and Khatri--Rao constructions as examples.
Taken together, the results of this paper support the use of dimensionality reduction maps that are sparser or more structured than previously understood.}
\section*{Acknowledgements}

CC was supported by the Caltech Kortschak Scholars Program
and the NSF GRFP Award 2139433.
RAM, ENE, and JAT were supported by ONR Award N00014-24-1-2223,
a Caltech Center for Sensing to Intelligence grant,
and the Caltech Carver Mead New Adventures Fund.
ENE was also supported
by the DOE CSGF under Award DE-SC0021110.

\appendix

\section{Analysis of the SparseCol test matrix}
\label{app:sparse-col}

This section establishes \cref{lem:indep-cols-moments},
which controls the first two moments of the SparseCol test matrix.

\begin{proof}[Proof of \cref{lem:indep-cols-moments}]
    Recall that $\mW = k^{-1} \vomega \vomega^\top$,
    where $\vomega \in \F^d$ is the sparse random
    vector introduced in the construction of the
    SparseCol random matrix (\cref{def:sparsecol}).
    We can easily obtain the first moment of the
    matrix $\mW$
    from the linearity of expectation,
    so we focus on the second-moment bound.
    The quantity of interest can be expressed as
    \begin{align}
        \Mom[\mW](\mM)
        = \E\left[\bigl(\tr(k^{-1} \vomega\vomega^\top \mM)\bigr)^2\right]
        = \frac1{k^2} \sum_{i,j,u,t=1}^d m_{ij} m_{ut} \E\bigl[\omega_i \omega_j \omega_u \omega_t \bigr].
        \label{eq:sparse-moment-m-sum}
    \end{align}
    Here, $m_{ij}$ and $\omega_i$ refer to the entries of $\mM$ and $\vomega$.
	In view of \cref{def:sparsecol}, the entries of the sparse random vector take the form \(\omega_i = \sqrt{d/\xi} \sum_{q=1}^\xi \rad_q \mathbbm1_{\{s_q = i\}}\).
	Therefore, the product of four such entries%
    expands as a sum:
	\begin{align}
		\E\bigl[\omega_i \omega_j \omega_u \omega_t\bigr]
		= \frac{d^2}{\xi^2}
			\sum_{q_1,q_2,q_3,q_4=1}^\xi
				\E\bigl[\rad_{q_1} \rad_{q_2} \rad_{q_3} \rad_{q_4}\bigr] \cdot
				\E\bigl[\mathbbm1_{\{s_{q_1}=i\}}
				\mathbbm1_{\{s_{q_2}=j\}}
				\mathbbm1_{\{s_{q_3}=u\}}
		\mathbbm1_{\{s_{q_4}=t\}}\bigr]
		\label{eq:sparse-moment-omega-foil}
	\end{align}
    Each random variable \(\rad_q \sim \textsc{rademacher}\), so
	\[
		\E\bigl[\rad_{q_1} \rad_{q_2} \rad_{q_3} \rad_{q_4}\bigr]
		= \mathbbm1_{\{q_1 = q_2 = q_3 = q_4\}}
		+ \mathbbm1_{\{q_1 = q_2 \neq q_3 = q_4\}}
		+ \mathbbm1_{\{q_1 = q_3 \neq q_2 = q_4\}}
		+ \mathbbm1_{\{q_1 = q_4 \neq q_2 = q_3\}}.
	\]
    Therefore, in \cref{eq:sparse-moment-omega-foil},
    it suffices to examine the sum
    over tuples \((q_1,q_2,q_3,q_4)\)
    where the indicators are nonzero.
    We proceed through each of the four expectations.
    In the first case, we have \(q_1 = q_2 = q_3 = q_4\),
    whence%
    \[
        \E\bigl[\mathbbm1_{\{s_{q_1}=i\}}
            \mathbbm1_{\{s_{q_1}=j\}}
            \mathbbm1_{\{s_{q_1}=u\}}
            \mathbbm1_{\{s_{q_1}=t\}}\bigr]
        = \prob\bigl\{s_{q_1} = i\bigr\} \cdot \mathbbm1_{\{i=j=u=t\}}
        = \frac{1}{d} \mathbbm1_{\{i=j=u=t\}}.
    \]
    In the second case, we have \(q_1 = q_2 \neq q_3 = q_4\),
    whence%
    \[
        \E\bigl[\mathbbm1_{\{s_{q_1}=i\}}
            \mathbbm1_{\{s_{q_1}=j\}}
            \mathbbm1_{\{s_{q_3}=u\}}
            \mathbbm1_{\{s_{q_3}=t\}}\bigr]
        = \prob\bigl\{ s_{q_1} = i \bigr\}
        \cdot \Pr\bigl\{ s_{q_3} = u\bigr\} \cdot \mathbbm1_{\{i=j\neq u=t\}}
        = \frac1{d(d-1)} \mathbbm1_{\{i=j \neq u=t\}}.
    \]
    The third and fourth cases are symmetric to the second, but with the \(i,j,u,t\) symbols permuted appropriately.
    Substitute the last three displays into \cref{eq:sparse-moment-omega-foil} and then into \cref{eq:sparse-moment-m-sum}, and simplify the result.
\end{proof}

\section{Analysis of the SparseStack test matrix}
\label{app:sprase-stack}

In this section, we prove \Cref{thm:sparse-stack-osi},
which establishes an OSI guarantee for the SparseStack
test matrix (\cref{def:sparse-stack}).
We work in the complex field, as the real field is
just a special case.
The SparseStack random matrix $\mOmega \in \C^{d \times k}$
has a sparsity parameter $\zeta \in \bbN$
and a block size parameter $b \in \bbN$.
The embedding dimension $k \coloneqq b \zeta$.

We can express the SparseStack matrix in terms of a
CountSketch matrix~\cite{charikar04,Clarkson13,meng13,nelson13}:
\begin{equation} \label{eq:countsketch}
    \mPhi = \bmat{
            \rad_1 \ve_{s_1}^\top \\
            \vdots \\
            \rad_d \ve_{s_d}^\top
        } \in \C^{d \times b}
        \quad\text{where}\quad
        \begin{aligned}
        &\text{$\rad_i \sim \textsc{rademacher}$ iid;} \\
        &\text{$s_i \sim \textsc{uniform}\{1,\dots, b\}$ iid.}
        \end{aligned}
\end{equation}
As usual, $\mathbf{e}_i \in \C^b$ is a standard basis vector.
Then the SparseStack matrix takes the form
\begin{equation} \label{eqn:sparsestack-app}
        \mOmega = \frac{1}{\sqrt\zeta} \bmat{
            \mPhi_1 & \cdots & \mPhi_\zeta
        } \in \C^{k \times d}
        \quad
        \text{where $\mPhi_1, \ldots, \mPhi_\zeta \sim \mPhi$ iid}.
    \end{equation}
This characterization is valuable because it
allows us to write $\mOmega \mOmega^{\top} = \zeta^{-1} \sum_{i=1}^\zeta \mPhi_i \mPhi_i^\top$, which is an average of iid random psd matrices.

To prove \cref{thm:sparse-stack-osi}, we employ the Gaussian comparison approach (\cref{impthm:comparison-iid-sum}).
As with the SparseCol matrix, the first step requires us
to compute first and second moments of our test matrix.

\begin{lemma}[CountSketch moments]
    \label{lem:sparse-stack-moments}
    Let \(\mPhi \in \bbC^{d \times b}\) be the CountSketch matrix \cref{eq:countsketch}.
    Then~ \(\E[\mPhi\mPhi^\top]=\mI\).
    In addition, for each self-adjoint matrix $\mM \in \C^{d \times d}$,
    \[
        \Mom[\mPhi\mPhi^\top](\mM) \leq (\tr(\mM))^2 + \frac2b\norm{\mM}_F^2.
    \]
\end{lemma}
We defer proof of \cref{lem:sparse-stack-moments}
to \cref{sec:countsketch-moments}.
With the result at hand, we turn to the proof of~\cref{thm:sparse-stack-osi}.
\begin{proof}[Proof of \Cref{thm:sparse-stack-osi}]
Let $\mOmega \in \C^{k \times d}$ be the SparseStack
test matrix constructed in~\eqref{eqn:sparsestack-app}.
It is easy to verify that $\mOmega$ is isotropic~\eqref{eqn:osi-isotropy}, so we focus on establishing the injectivity property~\eqref{eq:injectivity}.

To begin, write the outer product \(\mOmega\mOmega^\top\) as an iid sum.  Define $\mW \coloneqq \zeta^{-1} \mPhi \mPhi^\top \in \C^{d \times d}$,
and note that
\[
    \mOmega\mOmega^\top = \sum_{i=1}^\zeta \mW_i%
    \quad
    \text{where $\mW_i \sim \mW$ iid.}
\]
\Cref{lem:sparse-stack-moments} implies that the first moment \(\E[\mW] = \zeta^{-1} \mI\) and that the second moment
\[
    \Mom[\mW](\mM)
    \leq \frac1{\zeta^2} (\tr(\mM))^2 + \frac{2}{\zeta^2b} \norm{\mM}_F^2.
\]
Now, we build a Gaussian comparison model.
Let \(\mX \in \C^{d\times d}\) be the Gaussian matrix
\[
    \mX \defeq \frac1\zeta\mI_d + \frac1\zeta g\mI_d + \frac1\zeta\sqrt{\frac2b} \mG_{\rm gue}^{(d)},
\]
where \(g\mI_d\) and \(\mG_{\rm gue}^{(d)}\) are the scaled identity and GUE matrices defined in \Cref{table:gaussian-library}.
We can compare the moments of $\mX$ and $\mW$.
The first moments are equal: \(\E[\mX]=\frac1\zeta\mI_d = \E[\mW]\). The second moments satisfy
\[
    \Var[\mX](\mM)
    = \frac1{\zeta^2} (\tr(\mM))^2 + \frac{2}{b\zeta^2}\norm{\mM}_F^2
    \geq \Mom[\mW](\mM).
\]
In other words, $\mX$ is a Gaussian comparison model for $\mW$.
For any orthonormal matrix $\mQ \in \C^{d \times r}$, we can
use equivariance~\eqref{eqn:gaussian-compare-conjugation} to confirm that
$\mQ^\top \mX \mQ$ is a Gaussian comparison model for $\mQ^\top \mW \mQ$.

We seek to establish control on the minimum eigenvalue \(\lambda_{\rm min}(\mQ^\top\mOmega\mOmega^\top\mQ) = \lambda_{\rm min}(\sum_{i=1}^\zeta \mQ^\top\mW_i\mQ)\).
To do so, we introduce iid copies \(\mX_1,\ldots,\mX_\zeta\) of the Gaussian matrix $\mX$, and we define
\[
    \mY \defeq \sum_{i=1}^\zeta \mQ^\top\mW_i\mQ
    \quad\text{and}\quad
    \mZ \defeq \sum_{i=1}^\zeta \mQ^\top\mX_i\mQ.
\]
Since \mY and \mZ are sums of iid matrices in $\C^{d\times d}$, \cref{impthm:comparison-iid-sum} implies that%
\[
    \lambda_{\rm min}(\mQ^\top\mOmega\mOmega^\top\mQ)
    = \lambda_{\rm min}(\mY)
    \geq \E[\lambda_{\rm min}(\mZ)] - \sqrt{2\sigma_*^2(\mZ)\log(2r/\delta)} \quad \text{with probability at least \(1-\delta\)}.
\]
To prove the theorem, we must establish a lower bound on the expected minimum eigenvalue $\E[\lambda_{\rm min}(\mZ)]$ and an upper bound on the weak variance $\sigma_*^2(\mZ)$.
To do so, we first examine the distribution of \mZ more closely.
By linearity, Gaussianity, and the definition $k = b \zeta$,%
\[
    \sum_{i=1}^\zeta \mX_i
    \sim \mI_d + \frac1{\sqrt\zeta} g\mI_d + \sqrt{\frac2{k}} \mG_{\rm gue}^{(d)}.
\]
Now, we conjugate by $\mQ \in \F^{d\times r}$.
By rotational invariance (\cref{implem:ortho-conjugate-gue}), \(\mQ^\top\mG_{\rm gue}^{(d)}\mQ \sim \mG_{\rm gue}^{(r)}\).
Thus,
\[
    \mZ =  \sum_{i=1}^\zeta \mQ^\top\mX_i\mQ \sim \mI_r + \frac1{\sqrt\zeta} g\mI_r + \sqrt{\frac2k} \mG_{\rm gue}^{(r)}.
\]
By symmetry of the Gaussian distribution with respect to negation,
we can bound the expected minimum eigenvalue of $\mZ$ below:%
\begin{align*}
    \E\bigl[\lambda_{\rm min}(\mZ)\bigr]
    &\geq 1 - \frac1{\sqrt\zeta}\E\bigl[\lambda_{\rm max}(g\mI_r)\bigr] - \sqrt{\frac2k} \E\bigl[\lambda_{\rm max}(\mG_{\rm gue}^{(r)})\bigr] 
    = 1 - \sqrt{\frac{4r}k}.%
\end{align*}
We have introduced results from \Cref{table:gaussian-library}.
Next, we bound the weak variance above
via the subadditivity property~\cite[Sec.~3.6]{tropp25}
and the results from \Cref{table:gaussian-library}:
\begin{align*}
    \sigma_*^2(\mZ)
    &\leq%
    \frac1\zeta \sigma_*^2(g\mI_r) + \frac2k \sigma_*^2(\mG_{\rm gue}^{(d)}) 
    \leq%
    \frac1\zeta + \frac1k \leq \frac3\zeta,
\end{align*}
where the last inequality relies on the fact \(\zeta \leq k\).
We conclude that, with probability at least \(1-\delta\),%
\[
    \lambda_{\rm min}(\mQ^\top\mOmega\mOmega^\top\mQ)
    \geq 1 - \sqrt{\frac{4r}{k}} - \sqrt{\frac{6}{\zeta}\log(2r/\delta)}.
\]
Set $\delta = \nicefrac{1}{20}$, and note that
the right-hand side exceeds $\nicefrac{1}{2}$
at an embedding dimension \(k = \cO(r)\) and a sparsity level \(\zeta = \cO(\log r)\).  This argument confirms the injectivity
property~\eqref{eq:injectivity}.
\end{proof}

\subsection{Moments of the CountSketch matrix} \label{sec:countsketch-moments}
In this section, we calculate the first two moments of the CountSketch matrix.

\begin{proof}[Proof of \cref{lem:sparse-stack-moments}]
Introduce the CountSketch matrix, $\mPhi \in \C^{d \times b}$
described in~\eqref{eq:countsketch}, whose entries are denoted
$\phi_{i\alpha}$.
The first moment $\E{}[ \mPhi \mPhi^\top ] = \Id$ by linearity of
expectation, so we focus on the second moment.
    To begin, we use the definition of the second moment
    function%
    to reach
    the expansion%
    \begin{multline}\label{eq:indep-subrow-sum-of-c}
        \Mom[\mPhi\mPhi^\top](\mM)
        = \E[(\tr(\mM(\mPhi\mPhi^\top)))^2]= \E\left[\left|\sum_{i,j=1}^d \sum_{\alpha=1}^b m_{ij}\phi_{i\alpha}\phi_{j\alpha}\right|^2\right] \\ = \sum_{i,j,t,u=1}^d \sum_{\alpha,\beta=1}^b m_{ij}m_{t u} \E[\phi_{i \alpha} \phi_{j \alpha} \phi_{t \beta} \phi_{u \beta}] \eqqcolon \sum_{i,j,t,u=1}^d m_{i j} m_{t u} \theta_{i j t u}.
    \end{multline}
    In the last line, we defined the quantities \(\theta_{i j t u} 
    \defeq \sum_{\alpha,\beta=1}^b \E[\phi_{i \alpha} \phi_{j \alpha} \phi_{t \beta} \phi_{u \beta}]\).
    We will compute each of these terms
    before returning to \eqref{eq:indep-subrow-sum-of-c}.
    
    By construction of the CountSketch matrix, its entries take the form \(\phi_{i \alpha} = \rad_i \mathbbm1_{\{s_i = \alpha\}}\).  Thus,
    \begin{align}
        \theta_{i j t u}
        = \sum_{\alpha,\beta=1}^b \E\bigl[\rad_i \rad_j \rad_{t} \rad_{u}\bigr] \cdot \E\bigl[\mathbbm1_{\{s_i=\alpha\}}\mathbbm1_{\{s_j=\alpha\}}\mathbbm1_{\{s_t=\beta\}}\mathbbm1_{\{s_{u}=\beta\}}\bigr].
        \label{eq:indep-subrow-expand-c-as-sum}
    \end{align}
    Each random variable \(\rad_q \sim \textsc{rademacher}\), so
    \[
        \E\bigl[\rad_i \rad_j \rad_{t} \rad_{u}\bigr]
        = \mathbbm1_{\{i = j = t = u\}}
        + \mathbbm1_{\{i = j \neq t = u\}}
        + \mathbbm1_{\{i = t \neq j = u\}}
        + \mathbbm1_{\{i = u \neq j = t\}}.
    \]
    This equation follows from%
    the fact that the rows of $\mPhi$ are independent.
    Therefore, in \cref{eq:indep-subrow-expand-c-as-sum},
    it suffices to compute \(\theta_{i j t u}\) for tuples \((i, j, t, u)\) that make the indicators nonzero.
    We proceed through each of the four indicators.
    In the first case, we have \(i=j=t=u\).
    Then
    \[
        \E\bigl[\mathbbm1_{\{s_i=\alpha\}}\mathbbm1_{\{s_i=\alpha\}}\mathbbm1_{\{s_i=\beta\}}\mathbbm1_{\{s_i=\beta\}}\bigr]
        = \frac1b \mathbbm1_{\{\alpha=\beta\}}.
    \]
    Thus, \(\theta_{i i i i} = \sum_{\alpha,\beta=1}^b b^{-1} \mathbbm1_{\{\alpha=\beta\}} = 1\).
    In the second case, \(i = j \neq t = u\), so%
    \[
        \E\bigl[\mathbbm1_{\{s_i=\alpha\}}\mathbbm1_{\{s_i=\alpha\}}\mathbbm1_{\{s_t=\beta\}}\mathbbm1_{\{s_t=\beta\}}\bigr]
        = \frac1{b^2}.
    \]
    Thus, \(\theta_{i i t t} = \sum_{\alpha,\beta=1}^b b^{-2} = 1\).
    In the third case, \(i = t \neq j = u\), so%
    \[
        \E \bigl[ \mathbbm1_{\{s_i=\alpha\}} \mathbbm1_{\{s_j=\alpha\}} \mathbbm1_{\{s_i=\beta\}} \mathbbm1_{\{s_j=\beta\}} \bigr]
    = \frac1{b^2} \mathbbm1_{\{\alpha=\beta\}}.
    \]
    Thus, \(\theta_{i j i j} = \sum_{\alpha,\beta=1}^b b^{-2} \mathbbm1_{\{\alpha=\beta\}} = b^{-1} \).
    In the fourth case, \(i = u \neq j = t\).
    This case is the same as the third, and we determine that \(\theta_{i j j i} = b^{-1} \).
    Substitute these identities for $\theta_{ijtu}$ back into \cref{eq:indep-subrow-sum-of-c} to reach
    \begin{align*}
        \Mom[\mPhi\mPhi^\top](\mM)
        &= (\tr(\mM))^2 + \frac1b \sum_{\substack{i,j=1\\i\neq j}}^d |m_{i j}|^2 + \frac1b \sum_{\substack{i,j=1\\i\neq j}}^d m_{i j}^2
        \leq (\tr(\mM))^2 + \frac2b \norm{\mM}_F^2.
    \end{align*}
    In the last step, since $\mM$ is self-adjoint, we may bound $\sum_{i \neq j} m_{ij}^2 = \sum_{i \neq j} \Re(m_{ij}^2)
    \leq \norm{\mM}_{\rm F}^2$.
\end{proof}

\section{Proof of \cref{thm:osi-from-small-ball}}
\label{app:small-ball-analysis}

To prove \cref{thm:osi-from-small-ball}, we use the following result from VC theory:
\begin{theorem}[Symmetric classifiers]
    \label{thm:symmetric-classifier-concentration}
    Let \(\cF\) be the class of symmetric two-sided linear classifiers
    \[
        \cF \defeq 
        \left\{
            \vx \mapsto \mathbbm1_{\{|\langle \vx,\vz\rangle|^2 \geq \tau\}}
            ~:~
            \vz\in\bbR^r,\tau\geq0
        \right\},
    \]
    and let \(\vx_1,\ldots,\vx_k\) be iid copies of a random vector \(\vx\in\bbR^r\).
    With probability at least \(1-\delta\), for all \(f\in\cF\), it holds that
    \[
        \frac1k \sum_{i=1}^k f(\vx_i)
        ~\geq~
        \E[f(\vx)]
        ~- \sqrt{
            \tfrac{r}{k}\log\left(\tfrac{\e k}{r}\right)
        } - \sqrt{
            \newtext{\tfrac{1}{8k}\log\left(\tfrac1\delta\right)}
        }.
    \]
\end{theorem}
Asymptotically equivalent results are known \cite{tropp23hdp,mohri2012foundations}; we provide sharper constants.
Our version of \cref{thm:symmetric-classifier-concentration} is proven in \cref{app:vc-dim-analysis}.
This result simplifies the results in \cite{tropp2015convex} and obtains tighter constants.
The proof is by the small-ball method \cite{GKK+22:Geometry-Polytopes,tropp23hdp,Ver25:High-Dimensional-Probability-2ed}.
With this preparation in place, we frame a bound on the minimum singular value of a matrix with iid columns:
\begin{theorem}[Minimum singular value by small-ball method, real case]
    \label{thm:small-ball-min-sing-val}
    Let \(\mX = [\vx_1 ~ \cdots ~ \vx_k] \in \bbR^{r \times k}\), where \(\vx_i \in \bbR^r\) are iid copies of random vector \vx.
    Then, for all \(\tau > 0\) and \(\delta\in(0,1)\), we have that
    \[
        \sigma_{\rm min}^2(\mX) \geq \tau k \left(
            \min_{\norm\vu_2=1}
            \Pr\left\{|\langle \vu, \vx\rangle|^2 \geq \tau\right\}
            -   
            \sqrt{\tsfrac{r}{k}\log\left(\tsfrac{\e k}{r}\right)} - \sqrt{\tsfrac{1}{8k}\log\left(\tsfrac1\delta\right)}
        \right)
    \]
    holds with probability at least \(1-\delta\).
\end{theorem}
\begin{proof}
    We expand the minimum singular value of \mX:
    \[
        \sigma_{\rm min}^2(\mX)
        = \min_{\norm\vu_2=1} \norm{\vu^\top\mX}_2^2
        = \min_{\norm\vu_2=1} \sum_{i=1}^k |\langle \vu, \vx_i\rangle|^2
    \]
    Then, since \(t \geq \tau \mathbbm1_{\{t \geq \tau\}}\) for positive \(t\), we have that
    \[
        \sigma_{\rm min}^2(\mX)
        \geq \tau k \left( \min_{\norm\vu_2=1} \frac1k \sum_{i=1}^k \mathbbm1_{\left\{|\langle \vu, \vx_i\rangle|^2 \geq \tau\right\}} \right).
    \]
    Notice that the term inside the minimum is an empirical mean of indicator variables.
    Thus, the stated result follows immediately by \cref{thm:symmetric-classifier-concentration}.
\end{proof}

Notice that \cref{thm:small-ball-min-sing-val} is locked into the real field \(\bbF=\bbR\).
The real case of \cref{thm:osi-from-small-ball} follows immediately from this result and a few short numerical calculations.
\Cref{app:small-ball-complex} addresses the complex case.
We also state a high-probability version of this result with small explicit constants:

\begin{theorem}[High probability subspace injection by small-ball method]
    Let \(\mOmega = \frac1{\sqrt k}[\vomega_1 ~ \cdots ~ \vomega_k] \in \bbF^{d \times k}\) be generated from iid copies \(\vomega_i\) of an isotropic random vector \(\vomega\).
    Suppose that \(\Pr\bigl\{|\langle \vu, \vomega\rangle|^2 \geq \tau\bigr\} \geq 0.99\) for all unit vectors \(\vu\in\bbF^d\) and some value \(\tau>0\).
    Then for each orthonormal matrix \(\mQ\in\bbF^{d \times r}\), it holds that \(\newtext{\sigma_{\rm min}^2}(\mQ^\top\mOmega) \geq \frac{\tau}{10}\) with probability at least \(1-\delta\) provided that \(k \geq \max\{4r,10\log(1/\delta)\}\) for \(\bbF=\bbR\) or else \(k \geq \max\{\newtext{16r,32}\log(1/\delta)\}\) for \(\:\bbF=\bbC\).
\end{theorem}

\subsection{Refined VC dimension bounds}
\label{app:vc-dim-analysis}

In this section, we develop \Cref{thm:symmetric-classifier-concentration} using standard tools from VC theory carefully.
We first prove the real case, and extend our results to the complex case in \cref{app:small-ball-complex}.
We start by importing useful definitions and concentration results.
We start with the growth function:
\begin{definition}[Growth function, \protect{\cite[Def.~3.6]{mohri2012foundations}}]
    \label{def:growth-func}
    Let \cX be a set, and let \(\cF\) be a set of functions mapping from \cX to \(\{0,1\}\).
    Then, the \emph{growth function} of \cF is defined as 
    \[
        \Pi_{\cF}(k)
		\defeq
		\max_{\vx_1,\ldots,\vx_k\in\cX}
		|\{
            (f(\vx_1),\ldots,f(\vx_k))
            ~:~
            f\in\cF
        \}|.
	\]
    The growth function $\Pi_{\cF}(k)$ counts the number of ways functions in \cF can assign labels to $k$ points in \cX.
\end{definition}
We list two standard properties of the growth function:
\begin{importedtheorem}[Concentration of empirical mean, adapted from \protect{\cite[Cor.~3.9]{mohri2012foundations}}]
	\label{impthm:growth-func-generalization}
	Let \cX be a set, and let \cF be a set of functions that map from \cX to \(\{0,1\}\).
	Let \(\vx_1,\ldots,\newtext{\vx_k}\) be iid copies of a random variable \(\vx\in\cX\).
	Then, with probability at least \(1-\delta\), it holds simultaneously for all \(f\in\cF\) that
	\[
		\frac1k \sum_{i=1}^k f(\vx_i)
		\geq
		\E[f(\vx)]
		- \sqrt{
			\frac{\log(\Pi_{\cF}(k))}{2k}
		} - \sqrt{
			\frac{\log(\newtext{\nicefrac1\delta})}{8k}
		}
	\]
\end{importedtheorem}

\begin{importedlemma}[\newtext{Self-union, adapted from} \protect{\cite[Exercise 3.23]{mohri2012foundations}}]
	\label{implem:growth-func-self-intersec}
    \newtext{
    Let \cF be a set of functions that maps from \cX to \(\{0,1\}\), and define the \emph{self-union} \(\cF_{\cup} \defeq \bigl\{
        \vx \mapsto \mathbbm1_{\{f_1(\vx)=1 \text{ or } f_2(\vx)=1\}}
        ~:~
        f_1,f_2\in\cF
    \bigr\}\).
	Then
	\(
		\Pi_{\cF_\cup} \leq (\Pi_{\cF_{}})^2.
	\)}
\end{importedlemma}
\newtext{
    Exercise 3.23 of \cite{mohri2012foundations} states the last bound
    for the self-intersection of a set of functions;
    our variant follows after applying De Morgan's law.
}
Next, we import the definition of the VC dimension and recall its relation to the growth function:
\begin{definition}[VC dimension, \protect{\cite[Def.~3.10]{mohri2012foundations}}]
	\label{def:vc-dim}
	Let \cX be a set, and let \cF be a set of functions that map from \cX to \(\{0,1\}\).
	The \emph{VC dimension} of \cF is defined as
	\(
		\text{VCdim}(\cF)
		\defeq
		\max \{k ~:~ \Pi_{\cF}(k) = 2^k\}.
	\)
\end{definition}
\begin{importedlemma}[VC dimensions of linear classifiers, \protect{\cite[Thm.~9.2]{shalev14}}]
    \label{implem:linear-vc-dim}
    The set of linear classifiers
    \(
        \cH \defeq \left\{
            \vx \mapsto \mathbbm1_{\{\langle\vx,\vu\rangle>\tau\}}
            ~:~
            \vu\in\bbR^r, \tau\in\bbR
        \right\}
    \) has VC dimension equal to \(r\).
\end{importedlemma}
\begin{importedlemma}[Sauer's Lemma, \protect{\cite[Cor.~3.18]{mohri2012foundations}}]
	\label{implem:sauer}
	Let \cF be a set of functions with VC dimension \(r\).
	For \(k \geq r\), the growth function is bounded as
	\(
		\Pi_{\cF}(k) \leq ( \e k/r )^r.
	\)
\end{importedlemma}
\begin{proof}[Proof of \cref{thm:symmetric-classifier-concentration}]
    Let \(\cH\) be the set of linear classifiers in \cref{implem:linear-vc-dim}.
	For any two-sided linear classifier \(f\in\cF\), we can write
	\[
        \newtext{
            f(\vx)
            = \mathbbm1_{\{\langle\vx,\vu\rangle^2\geq\tau\}}
            = \mathbbm1_{\left\{\langle\vx,\vu\rangle\geq\sqrt\tau \text{ or } \langle\vx, -\vu\rangle\geq\sqrt\tau\right\}}
            = \mathbbm1_{\{h_1(\vx)=1 \text{ or } h_2(\vx)=1\}}
             \quad \text{for some \(h_1,h_2\in\cH\).}
        }
	\]
    We find that $\cF \subseteq \newtext{\cH_\cup}$.
	By \Cref{implem:growth-func-self-intersec,implem:linear-vc-dim,implem:sauer}, we infer
	\[
        \newtext{\Pi_{\cF}(k)}
		\leq \newtext{\Pi_{\cH_{\cup}}(k)}
		\leq (\Pi_{\cH}(k))^2
		\leq \left(\e k/{r}\right)^{2r} \quad \text{for all } k \ge r.
	\]
	According to \Cref{impthm:growth-func-generalization}, simultaneously for all \(f\in\cF\),%
	\begin{align*}
		\frac1k \sum_{i=1}^k f(\vx_i)
		&\geq
		\E[f(\vx)]
		- \sqrt{
			\frac{\log(\Pi_{\cF}(k))}{2k}
		} - \sqrt{
			\frac{\log(\nicefrac1\delta)}{8k}
		}%
		\geq
		\E[f(\vx)]
		- \newtext{\sqrt{
			\frac{r}{k}\log\left(\tfrac{\e k}{r}\right)
		} - \sqrt{
			\frac{1}{8k}\log\left(\tfrac1\delta\right)
		}}.
	\end{align*}
    This is the desired result.
\end{proof}

\subsection{Analysis in the complex field}
\label{app:small-ball-complex}

To adapt our results to the complex setting, we first develop a version of \cref{thm:symmetric-classifier-concentration} for unions of symmetric two-sided linear classifiers and a complex-valued version of \cref{thm:small-ball-min-sing-val}.

\begin{theorem}[Union of symmetric classifiers]
    \label{thm:symmetric-classifier-concentration-complex}
    Let \(\newtext{\cF_{\mathrm{pair}}}\) be the class of unions of symmetric classifiers
    \[
        \newtext{\cF_{\mathrm{pair}}} \defeq 
        \left\{
            \vx \mapsto \mathbbm1_{\left\{|\langle \vx,\vz\rangle|^2 \geq \tau \text{ or } |\langle \vx,\vw\rangle|^2 \geq \tau\right\}}
            ~:~
            \vz,\vw\in\bbR^r,\tau\geq0
        \right\},
    \]
    and let \(\vx_1,\ldots,\vx_k\) be iid copies of a random vector \(\vx\in\bbR^r\).
    With probability at least \(1-\delta\), \newtext{all \(f\in\newtext{\cF_{\mathrm{pair}}}\) have}
    \[
        \newtext{
            \frac1k \sum_{i=1}^k f(\vx_i)
            ~\geq~
            \E[f(\vx)]
            ~- \sqrt{
                \tfrac{2r}{k}\log\left(\tfrac{\e k}{r}\right)
            } - \sqrt{
                \tfrac{1}{8k}\log\left(\tfrac1\delta\right)
            }.
        }
    \]
\end{theorem}
\begin{proof}
    Notice that \(\newtext{\cF_{\mathrm{pair}}}\) is \newtext{contained in} the self-union of the class of two-sided symmetric linear classifiers defined in \cref{thm:symmetric-classifier-concentration}.
    The proof of the advertised result tracks the proof of \cref{thm:symmetric-classifier-concentration}, instead using \cref{implem:growth-func-self-intersec} to bound \(\Pi_\cF(k) \leq (\e k/r)^{4r}\).
\end{proof}

\begin{theorem}[Minimum singular value by small-ball method, complex case]
    \label{thm:small-ball-min-sing-val-complex}
    Let \(\mX = [\vx_1 ~ \cdots ~ \vx_k] \in \bbC^{r \times k}\), where \(\vx_i \in \bbC^r\) are iid copies of random vector \vx.
    Then, for all \(\tau > 0\) and \(\delta\in(0,1)\), we have that
    \[
        \sigma_{\rm min}^2(\mX) \geq \newtext{2}\tau k \left(
            \min_{\norm\vu_2=1}
            \Pr\bigl\{|\langle \vu, \vx\rangle|^2 \geq \tau\bigr\}
            -   
            \sqrt{\tsfrac{4r}{k}\log\left(\tsfrac{\e k}{2r}\right)} - \sqrt{\tsfrac{1}{8k}\log\left(\tsfrac1\delta\right)}
        \right)
    \]
    holds with probability at least \(1-\delta\).
\end{theorem}
\begin{proof}
    As in the proof of \cref{thm:small-ball-min-sing-val}, we obtain a lower bound
    for the minimum singular value of \mX via
    \[
        \sigma_{\rm min}^2(\mX)
        \geq \newtext{2\tau k} \left( \min_{\norm\vu_2=1} \frac1k \sum_{i=1}^k \mathbbm1_{\left\{|\langle \vu, \vx_i\rangle|^2 \geq \newtext{2}\tau \right\}} \right).
    \]
    \newtext{
        Note that 
        \(|\langle \vu, \vx_i\rangle|^2 \geq 2\tau\)
        is guaranteed if either \(|\Re \: \langle \vu, \vx_i\rangle|^2 \geq 2\tau\) or \(|\Im \: \langle \vu, \vx_i\rangle|^2 \geq 2\tau\).
    }
    Now let \(\varphi:\bbC^r \rightarrow \bbR^{2r}\) be the map \(\varphi(\vx) \defeq \sbmat{\Re \vx \\ \Im \vx}\).
    We have that \(\Re \langle \vu, \vx_i\rangle = \langle \varphi(\vu), \varphi(\vx_i)\rangle\) and \(\Im \langle \vu, \vx_i\rangle = \langle \mP\varphi(\vu), \varphi(\vx_i)\rangle\) for permutation matrix \(\mP = \sbmat{\mat0 & \mI \\ \mI & \mat0}\).
    By introducing these observations, we get
    \[
        \sigma_{\rm min}^2(\mX)
        \geq \newtext{2\tau k} \left(
            \min_{\norm\vu_2=1} \frac1k \sum_{i=1}^k \mathbbm1_{
            \left\{
                |\langle \varphi(\vu), \varphi(\vx_i)\rangle|^2 \geq \newtext{2}\tau
                \text{ or }
                |\langle \mP\varphi(\vu), \varphi(\vx_i)\rangle|^2 \geq \newtext{2}\tau
            \right\}}
        \right).
    \]
    Notice that the term inside the minimum is an empirical mean of indicator variables involving vectors in \(\bbR^{2r}\).
    Thus, the stated result follows immediately by \cref{thm:symmetric-classifier-concentration-complex} \newtext{ and the fact that}
    \[
        \newtext{
            \Pr\bigl\{
                |\Re \langle \vu, \vx\rangle|^2 \geq 2\tau
                \text{ or }
                |\Im \langle \vu, \vx\rangle|^2 \geq 2\tau
            \bigr\}
            \geq
            \Pr\bigl\{|\langle\vu,\vx\rangle|^2 \geq \tau\bigr\}.
        }
    \]
    \newtext{This step completes the proof.}
\end{proof}

The complex case of \cref{thm:osi-from-small-ball} now follows as a corollary of \cref{thm:small-ball-min-sing-val-complex}.

\section{Implementation of scientific data compression experiment}
    \label{app:gross-pitaevskii}

In this appendix, we describe the experiment from \cref{sec:science-pod-modes} in greater detail.
The dynamics of a Bose--Einstein condensate trapped in an external potential $V_{\mathrm{trap}}$ rotating around the $z$-axis at an angular frequency $\Omega$ are governed by the \textit{Gross--Pitaevskii equation} (GPE):
$$\mathrm{i} \hbar \frac{\partial \psi(\mathbf{x}, t)}{\partial t}=\left[-\frac{\hbar^2}{2 m} \nabla^2+V_{\text {trap }}(\mathbf{x})+g|\psi|^2-\Omega L_z\right] \psi(\mathbf{x}, t) \quad \text{for} t \ge 0,$$
where $\hbar$ is the reduced Planck constant, $m$ is the atomic mass, $g$ is the interparticle interaction strength, and 
$
  L_z = -i\hbar\bigl(x\,\partial_y - y\,\partial_x\bigr)
$
is the \(z\)-component of the angular-momentum operator.

To study ground states, we follow the standard procedure \cite{zeng09} and consider the evolution of this equation in \emph{imaginary time}.
We choose the initial state $\psi_0(\vx,0)$ to be a Gaussian profile, carrying a phase defect in a harmonic potential $V_h$.
In order to simulate the imaginary-time GPE evolution, we use a Fourier pseudospectral method in space and a semi-implicit backward/forward Euler scheme in time following \cite{zeng09}. 
The initial state $\psi_0(\vx,0)$ is evolved for $40{,}000$ time steps using a time step of $\Delta t=10^{-4}$ on a $256\times 256$ domain.%
The PDE data $\vx(t)\in\bbC^{65{,}536}$ is natively complex, but we convert to real data by concatenating the real and imaginary parts $(\Re \vx(t),\Im \vx(t)) \in \bbR^{131{,}072}$.

\section{Implementation of partition function estimation experiment}
\label{app:trace-est}

In this appendix, we describe the experimental setup for \cref{sec:science-trace} in greater detail.
Partition function estimation for the transverse-field Ising model was studied in \cite{epperly24trace}.
Following their lead, we shift the Hamiltonian matrix by a constant $b=(1+h)\ell$ before applying stochastic trace estimation so that $\mH+b\mI$ is psd, allowing us to use variance-reduced trace estimators for psd matrices \mA like Nyström++ \cite{persson22} and XNysTrace \cite{epperly24trace}.
To compute \(\exp(-\beta\mH)\mOmega\) we follow \cite{al11}, evaluating the matrix exponential matvec implicitly via a truncated Taylor approximation.
We observe that when paired with Khatri--Rao test matrices, the previously mentioned variance-reduced trace estimators strongly outperform the naive Girard--Hutchinson trace estimator for this problem.

\bibliographystyle{modified_siamplain}

{\footnotesize
\bibliography{references}}

\end{document}